\pdfoutput=1

\documentclass[twocolumn, prx, aps, 10pt, floatfix, superscriptaddress, nofootinbib]{revtex4-2}

\usepackage[english]{babel}
\usepackage[dvipsnames,table]{xcolor}

\usepackage{amsmath,amsthm,amssymb}
\usepackage{mathtools}
\usepackage{physics}

\DeclareMathOperator*{\argmin}{arg\,min}
\DeclareMathOperator*{\argmax}{arg\,max}

\usepackage[export]{adjustbox}
\usepackage{tikz,pgfplots}
\pgfplotsset{compat=newest}

\usetikzlibrary{pgfplots.groupplots}
\usetikzlibrary{graphs,graphs.standard}
\usetikzlibrary{quantikz2}
\usetikzlibrary{tikzmark}
\usetikzlibrary{fadings,shapes.arrows,shadows,intersections}
\usepackage{float}
\usepackage{algorithmicx}
\usepackage{algpseudocode}
\usepackage{subcaption}

\usepackage{capt-of}
\usepackage{caption}

\DeclareMathOperator{\cost}{cost}

\usepackage{array}
\usepackage{booktabs}

\theoremstyle{plain}
\newtheorem{theorem}{Theorem}
\newtheorem{corollary}{Corollary}
\newtheorem{proposition}{Proposition}

\theoremstyle{definition}
\newtheorem{lemma}{Lemma}

\newcommand{\C}[0]{\mathbb{C}}

\newcommand{\Span}[1]{\operatorname{span}\!\left(#1\right)}

\newcommand{\Nof}[1]{\mathcal{N}({#1})}
\newcommand{\Ball}[2]{B_{#1}(#2)}
\newcommand{\Pof}[1]{\Pi_{\mathcal{N}({#1})}}

\newcommand{\inducedG}[1][G]{\mathcal G_{\mathrm{VC}}(#1)}

\usepackage[colorlinks=true,linkcolor=blue,citecolor=blue,urlcolor=blue]{hyperref}
\usepackage[nameinlink,noabbrev]{cleveref}

\newcounter{algbox}
\renewcommand{\thealgbox}{\arabic{algbox}}
\crefname{algbox}{algorithm}{algorithms}
\Crefname{algbox}{Algorithm}{Algorithms}

\crefname{appendix}{appendix}{appendices}
\Crefname{appendix}{Appendix}{Appendices}

\begin{document}

\title{Iterative quantum algorithms for the minimum vertex cover problem based on continuous-time quantum walks}
\date{\today}

\author{Ruben Pariente Bassa}
\affiliation{SINTEF AS, Department of Mathematics and Cybernetics, Oslo, Norway}

\author{Finley A. Quinton}
\affiliation{Department of Industrial Economics and Technology Management, NTNU, Trondheim}

\author{Franz G. Fuchs}
\affiliation{SINTEF AS, Department of Mathematics and Cybernetics, Oslo, Norway}
\affiliation{University of Oslo, Department of Mathematics, Oslo, Norway}

\author{Pascal Halffmann}
\affiliation{Fraunhofer Institute for Industrial Mathematics ITWM, Department of Financial Mathematics, Kaiserslautern, Germany}

\begin{abstract}
    We introduce a constraint-preserving hybrid quantum-classical greedy framework for the minimum vertex cover problem, which extends directly to maximum independent set by bitwise complementation.
    The framework uses projected Pauli-X terms whose sum preserves the feasible subspace and acts within it exactly as the adjacency matrix of a layered graph of feasible covers.
    This graph is connected, so every feasible cover is linked to the configuration containing all vertices by a sequence of allowed single-vertex flips.
    Starting from this configuration, the corresponding continuous-time quantum walk propagates amplitude into layers containing progressively smaller covers.
    We rank vertices using either their marginal cover probabilities or the expected cover size obtained after fixing each candidate vertex in the cover, and use these rankings to guide recursive greedy reductions.
    Across several random-graph families, with walk times fixed using independent calibration ensembles, the quantum-informed algorithms achieve lower mean approximation ratios and solve a larger fraction of instances optimally than their corresponding classical greedy baselines.
    The conditioned-energy strategy performs best on the tested instances and retains algorithmic performance close to the exact continuous-time limit under low-depth Trotterisation.
    For bounded-degree graphs, each Trotter layer has circuit depth independent of system size, and the framework requires neither penalty terms nor variational training.
\end{abstract}

\maketitle

\onecolumngrid
\vspace{-7\baselineskip}
\begin{center}
  \captionsetup{hypcap=false}
    \centering
    \begin{adjustbox}{width=\textwidth}
    \begin{tikzpicture}[
    >=Latex,
    node distance=1.6cm and 1.5cm,
    box/.style={
        draw,
        rounded corners=2pt,
        align=center,
        minimum width=3.2cm,
        minimum height=1.1cm,
        font=\small
    },
    smallbox/.style={
        draw,
        rounded corners=2pt,
        align=center,
        minimum width=3.2cm,
        minimum height=1.1cm,
        font=\small
    },
    vtx/.style={circle, draw, inner sep=1.2pt, minimum size=14pt, font=\scriptsize},
    chosen/.style={circle, draw, fill=black!20, inner sep=1.2pt, minimum size=14pt, font=\scriptsize},
    removed/.style={circle, draw=black!35, text=black!40, fill=black!8, inner sep=1.2pt, minimum size=14pt, font=\scriptsize},
    statenode/.style={
        rounded rectangle,
        draw=blue!60!black,
        fill=blue!6,
        inner xsep=1.pt,
        inner ysep=1.pt,
        font=\ttfamily\scriptsize
    },
    statechosen/.style={
        rounded rectangle,
        draw=blue!60!black,
        fill=blue!18,
        inner xsep=1.pt,
        inner ysep=1.pt,
        font=\ttfamily\scriptsize
    },
    edge/.style={line width=0.6pt},
    fadedge/.style={line width=0.6pt, draw=black!30},
    scorebar/.style={fill=black!70, draw=none},
    inducedge/.style={blue!55!black, line width=0.75pt},
    qrw/.style={red!70!black, line width=1.0pt},
    lab/.style={font=\scriptsize},
    weightlab/.style={font=\scriptsize\itshape}
]

\node[box] (state) {CTQW on induced graph\\$\mathcal{G}_{\mathrm{VC}}(G)$};

\node[smallbox, right=of state] (measure) {Vertex scores\\$s_i \in \{P_i,E_i\}$};

\node[smallbox, right=of measure] (select) {Greedy selection\\$v^\ast=\arg\max/\min_i s_i$};

\node[smallbox, right=of select] (reduce) {Reduce instance\\$G\leftarrow G \setminus\{v^\ast\}$};

\draw[->] (state) -- (measure);
\draw[->] (measure) -- (select);
\draw[->] (select) -- (reduce);

\def\loopheight{1.0}
\def\labelheight{0.8}

\draw[->, rounded corners=8pt]
    ([yshift=0.15cm]reduce.north) -- ++(0,\loopheight)
    -- ($(state.north)+(0,\loopheight+0.15)$)
    -- ([yshift=0.15cm]state.north);

\node[lab] at ($(state.north)!0.5!(reduce.north)+(0,\labelheight)$)
    {repeat until no edges remain/classically solvable};

\begin{scope}[shift={(-1.5,2.75)}]
\begin{adjustbox}{scale=0.4}
\node[vtx] (g1) at (-1.0,0.7) {};
\node[vtx] (g2) at (0.2,1.2) {};
\node[vtx] (g3) at (1.0,0.2) {};
\node[vtx] (g4) at (0.1,-0.9) {};
\node[vtx] (g5) at (-1.2,-0.5) {};

\draw[edge] (g1)--(g2);
\draw[edge] (g2)--(g3);
\draw[edge] (g1)--(g3);
\draw[edge] (g3)--(g4);
\draw[edge] (g4)--(g5);
\draw[edge] (g5)--(g1);
\draw[edge] (g2)--(g4);
\end{adjustbox}
\end{scope}

\begin{scope}[shift={(-2.,-2.0)}]



\def\xy{2,1.}
\def\xdist{1}
\def\ydist{.75}

\coordinate (rootvc) at (\xy);

\node[statechosen] (s11111) at (rootvc) {11111};

\node[statenode] (s01111) at ($(rootvc)+(-2*\xdist,-\ydist)$) {01111};
\node[statenode] (s10111) at ($(rootvc)+(-\xdist,-\ydist)$) {10111};
\node[statenode] (s11011) at ($(rootvc)+(0,-\ydist)$) {11011};
\node[statenode] (s11101) at ($(rootvc)+(\xdist,-\ydist)$) {11101};
\node[statenode] (s11110) at ($(rootvc)+(2*\xdist,-\ydist)$) {11110};

\node[statenode] (s01101) at ($(rootvc)+(-4*\xdist/3,-2*\ydist)$) {01101};
\node[statenode] (s11010) at ($(rootvc)+(0,-2*\ydist)$) {11010};
\node[statenode] (s10110) at ($(rootvc)+(4*\xdist/3,-2*\ydist)$) {10110};

\draw[inducedge] (s11111)--(s01111);
\draw[inducedge] (s11111)--(s10111);
\draw[inducedge] (s11111)--(s11011);
\draw[inducedge] (s11111)--(s11101);
\draw[inducedge] (s11111)--(s11110);

\draw[inducedge] (s01111)--(s01101);
\draw[inducedge] (s11101)--(s01101);

\draw[inducedge] (s10111)--(s10110);
\draw[inducedge] (s11110)--(s10110);

\draw[inducedge] (s11011)--(s11010);
\draw[inducedge] (s11110)--(s11010);

\end{scope}

\begin{scope}[shift={(5.2,-2.1)}]

\node[vtx,fill=red!90!green] (b1) at (-1.0,0.7) {1};
\node[vtx,fill=red!30!green] (b2) at (0.2,1.2) {2};
\node[vtx,fill=green]        (b3) at (1.0,0.2) {3};
\node[vtx,fill=red!60!green] (b4) at (0.1,-0.9) {4};
\node[vtx,fill=red]          (b5) at (-1.2,-0.5) {5};

\draw[edge] (b1)--(b2);
\draw[edge] (b2)--(b3);
\draw[edge] (b3)--(b4);
\draw[edge] (b4)--(b5);
\draw[edge] (b5)--(b1);
\draw[edge] (b1)--(b3);
\draw[edge] (b2)--(b4);

\end{scope}

\begin{scope}[shift={(9.8,-2.1)}]

\node[vtx]     (c1) at (-1.0,0.7) {1};
\node[vtx]     (c2) at (0.2,1.2) {2};
\node[removed] (c3) at (1.0,0.2) {$v^\ast$};
\node[vtx]     (c4) at (0.1,-0.9) {4};
\node[vtx]     (c5) at (-1.2,-0.5) {5};

\draw[edge]    (c1)--(c2);
\draw[fadedge] (c2)--(c3);
\draw[fadedge] (c1)--(c3);
\draw[fadedge] (c3)--(c4);
\draw[edge]    (c4)--(c5);
\draw[edge]    (c5)--(c1);
\draw[edge]    (c2)--(c4);

\node[vtx] (d1) at (3.8,0.7) {1};
\node[vtx] (d2) at (5.0,1.2) {2};
\node[vtx] (d4) at (4.9,-0.9) {4};
\node[vtx] (d5) at (3.6,-0.5) {5};

\draw[edge] (d1)--(d2);
\draw[edge] (d4)--(d5);
\draw[edge] (d5)--(d1);
\draw[edge] (d2)--(d4);

\end{scope}

\end{tikzpicture}
    \end{adjustbox}
    
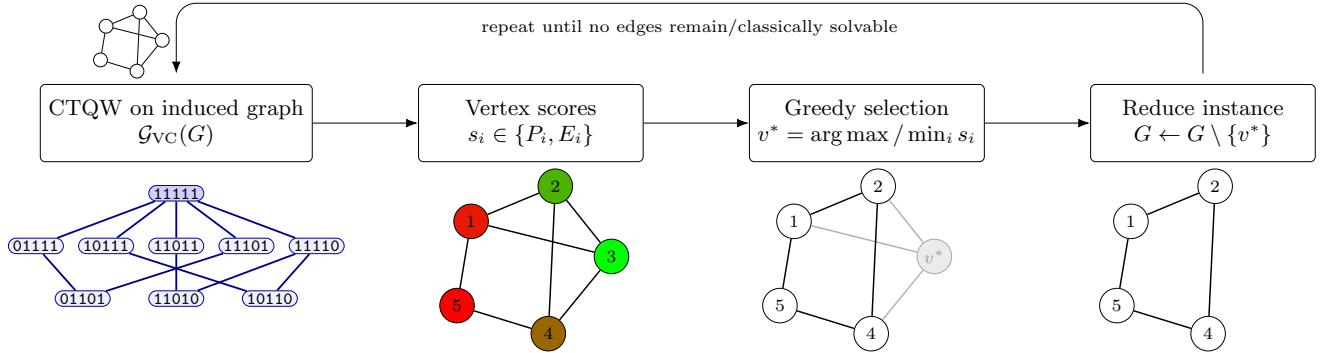
\captionof{figure}{
    Schematic of the iterative quantum algorithm.
    The input graph $G$ defines the induced feasible-state graph $\inducedG$, on which a constraint-preserving quantum walk prepares a state supported entirely on valid vertex covers. 
    Observables and conditioned cost evaluations of the CTQW state produce per-vertex scores $s_i\in\{P_i,E_i\}$ on the original graph, where $P_i$ is the marginal probability that vertex $i$ belongs to the cover and $E_i$ is the expected cover size obtained from a walk in which vertex $i$ is fixed in the cover.
    This converts global information in the feasible configuration space into locally accessible greedy decision signals.
    A vertex is then fixed according to the chosen rule, the graph is reduced accordingly, and the procedure is repeated on the reduced instance.
    In this way, the method combines constrained quantum exploration with recursive classical reduction, without penalty terms or variational optimisation.
    }
    \label{fig:recursive_vertex_fixing}
  \captionsetup{hypcap=true}
\end{center}
\twocolumngrid

\section{Introduction and related work}

The minimum vertex cover (MVC) and maximum independent set (MIS) problems are canonical NP-complete problems in combinatorial optimisation, dating back to Karp's seminal work on polynomial-time reductions~\cite{karp2009reducibility}. Given a graph $G=(V,E)$, MVC asks for a minimum-size subset of vertices incident to every edge, while MIS asks for a maximum-size subset of pairwise non-adjacent vertices. These problems are complementary: selecting a vertex in an independent set is equivalent to excluding it from the corresponding vertex cover. Thus, if $x_i\in\{0,1\}$ indicates whether vertex $i$ is included in a vertex cover and $y_i=1-x_i$ indicates inclusion in the complementary independent set, then
\begin{equation*}
    |\mathrm{MVC}| + |\mathrm{MIS}| = |V|.
\end{equation*}
This duality makes MVC and MIS a natural testbed for classical approximation algorithms, heuristic methods, and quantum optimisation techniques.

We formulate the minimum vertex cover problem as
\begin{equation}
    \min_{x\in\{0,1\}^{|V|}} \sum_{i \in V} c_i x_i
    \quad \text{s.t.} \quad
    x_i + x_j \ge 1
    \;\; \forall (i,j) \in E,
\label{eq:mvc}
\end{equation}
where $x_i=1$ indicates that vertex $i$ is selected in the cover and
$c_i\ge 0$ denotes its weight.
Equivalently, we identify a vertex cover $C\subseteq V$ with its indicator string $C\in\{0,1\}^{|V|}$, where $C_i=1$ if and only if $i\in C$.
For an algorithm returning a cover $C$, we define the approximation ratio as $r(C):=|C|/\tau(G)$, where $\tau(G)$ is the minimum vertex-cover size; thus $r(C)=1$ is optimal, and larger values indicate worse solutions.
The constraint enforces that every edge has at least one endpoint in the cover.
Unless stated otherwise, all MVC instances considered in this work are unweighted, so $c_i\equiv 1$ and the cost of a cover $C\subseteq V$ is $|C|$.
Classically, MVC admits a factor-2 approximation via maximal matchings, linear-programming relaxations, and primal-dual methods~\cite{vazirani2001approximation};
under the Unique Games Conjecture, a standard assumption in hardness-of-approximation theory, no polynomial-time algorithm can achieve an approximation factor of $2-\epsilon$ for any fixed $\epsilon>0$~\cite{khot2002power}.
Greedy algorithms remain important in practice because they are simple, scaleable, and often competitive. Largest-degree-first rules select vertices that cover many edges at once, while smallest-degree-first rules exploit low-degree structure either by excluding low-degree vertices from the cover or by selecting their neighbours~\cite{greedisgood,krysta2024ultimate}. More broadly, greedy ideas underlie repair heuristics, local improvement methods, primal heuristics in mixed-integer optimisation, and metaheuristics such as the greedy randomized adaptive search procedure (GRASP)~\cite{FeoResende1995}. Their main limitation is that the decision rule is typically local, depending only on degrees or nearby graph structure, and therefore may miss global correlations that determine high-quality solutions~\cite{BangJensen2004}.

Quantum optimisation provides several routes to incorporate such global structure. A common approach maps the constrained problem in \cref{eq:mvc} to a quadratic unconstrained binary optimisation (QUBO) problem or Ising Hamiltonian by adding penalty terms for violated constraints.
This strategy has been used in adiabatic quantum computation, quantum annealing, the quantum approximate optimisation algorithm (QAOA), and variational quantum eigensolvers (VQEs)~\cite{Lucas2014Ising,Pelofske2019AnnealerMVC,Farhi2014QAOA,peruzzo2014variational}.
Penalty encodings are flexible, but they introduce parameter-tuning issues and can create energy landscapes with many low-energy infeasible states~\cite{Fuchs_2022,herman2023constrained}.

A complementary line of work enforces feasibility by construction. The quantum alternating operator ansatz framework~\cite{hadfield2019quantum} provides a general method for designing mixers that preserve the feasible subspace throughout the evolution. Related constructions include XY-type mixers for cardinality and covering constraints~\cite{Wang2020XYMixers}, Grover-style mixers~\cite{Bartunik2020GroverMixers}, general subspace-preserving approaches~\cite{Fuchs_2022}, and stabiliser-based methods for constructing restricted mixers~\cite{fuchs2023LX}. Constraint-preserving methods have also been studied directly for MIS, where finite-depth ans\"atze can bias the dynamics toward larger independent sets~\cite{saleem2020max}. Other penalty-free mechanisms, including constrained evolution and quantum Zeno dynamics, similarly aim to keep the state inside the feasible space rather than penalizing constraint violations after they occur~\cite{herman2023constrained,bucher2025penalty}.

In parallel, quantum tree-search algorithms have established provable speedups for exact combinatorial optimisation. Quantum-walk techniques yield near-quadratic speedups for branch-and-bound and related tree-search procedures~\cite{montanaro2020quantum,chakrabarti2022universal}, while quantum tree-generation methods have been proposed for problems such as 0--1 knapsack~\cite{Wilkening_2025}. These results are important for fault-tolerant settings, but they are largely distinct from near-term heuristic approaches. Hybrid quantum-classical methods instead use shallow quantum circuits to generate information that guides classical reduction steps. Examples include warm-start methods based on classical relaxations~\cite{Egger2021WarmStart}, recursive QAOA~\cite{Bravyi2022hybridquantum}, quantum-informed recursive optimisation~\cite{Finzgar2024QIRO}, and scalable quantum-walk heuristics that evolve on the original graph and enforce feasibility through recursive freezing~\cite{luiz2025scalable}. These methods belong to the broader class of iterative quantum algorithms, where quantum subroutines and classical simplification are combined recursively~\cite{PhysRevA.110.052435,Dupont_2023}.

In this work we introduce a penalty-free, constraint-preserving quantum greedy framework for MVC and MIS.
The local feasibility-preserving operators used in our construction are the
projected Pauli mixer operators for MVC described in~\cite{hadfield2019quantum}.
Our contribution is to use these projected Pauli-$X$ operators in a different role. Instead of treating them as variational mixers inside a QAOA-type ansatz, we use them to define a continuous-time quantum-walk Hamiltonian on the feasible vertex-cover subspace. Restricted to this subspace, the Hamiltonian is the adjacency matrix of the induced feasible-state graph, whose vertices are feasible covers and whose edges connect covers that differ by a single feasibility-preserving bit flip, see \cref{fig:recursive_vertex_fixing}.

This induced-graph viewpoint turns the constrained walk into a source of vertex scores for classical greedy reduction. Starting from the all-in-the-cover state, the walk explores feasible covers without generating infeasible configurations. Local observables of the evolved state are then used to rank vertices on the original graph. Although the measured quantities are local, their values depend on the global structure explored by the constrained dynamics. At each iteration, the algorithm fixes either a high-scoring vertex into the cover or, in a complementary smallest-score rule, fixes the neighbours of a low-scoring vertex. The graph is reduced and the procedure is repeated on the residual instance.

We develop two main quantum-informed greedy algorithms.
The Quantum Probabilistic Greedy (QPG) algorithm scores vertices by their marginal occupation probability in the CTQW state, while the Quantum Energy Greedy (QEG) algorithm scores vertices by the expected cover size under a walk conditioned on fixing that vertex in the cover, thereby aggregating occupation information from across the graph.
Both algorithms are studied with largest-degree-first-like and smallest-degree-first-like reduction rules.
A third variant, Quantum Fidelity Greedy (QFG), uses an echo experiment to quantify how strongly a candidate vertex operator fails to commute with the remaining active operators and is discussed separately in \cref{sec:QFG}.
Numerical benchmarks show that the quantum-informed rules can outperform their classical greedy analogues on the graph families considered.

The remainder of the paper is organised as follows.
\Cref{sec:ConstrainedOperators} recalls the feasible Hilbert space and the
projected Pauli operators used to preserve the MVC constraints.
\Cref{sec:reachability_connectivity} proves that these generators connect the
all-in-the-cover state to every feasible vertex cover.
The induced feasible-state graph and the associated CTQW dynamics are developed
in \cref{sec:induced_graph_walk}, where we analyse the layered transport
structure.
The QPG and QEG algorithms, together with their greedy reduction rules, are
introduced in \cref{sec:IterativeQuantum}.
Numerical benchmarks on 3-regular, 4-regular, and Erd\H{o}s--R\'enyi random
graphs are reported in \cref{sec:Results}, followed by conclusions in
\cref{sec:conclusion}.
Additional material on the dynamical Lie algebra, entanglement structure,
circuit depth, Trotter error, Monte Carlo simulation, and QFG is provided in
the appendices.

\section{Constraint-preserving operators for MVC}
\label{sec:ConstrainedOperators}

The feasible vertex-cover subspace is generated by bit strings satisfying the edge constraints of the MVC problem.
On this subspace, local bit flips must be restricted so that feasibility is preserved.
The projected Pauli operators used below are the MVC constraint-preserving mixer operators described in~\cite{hadfield2019quantum};
they are recalled here to fix notation and make the continuous-time-walk construction self-contained.

Let $G=(V,E)$ be a graph with $n=|V|$ vertices. We represent a candidate vertex cover by a bit string $C \in \{0,1\}^n$, where $C_i=1$ indicates that vertex $i$ is included in the cover. The feasible Hilbert space is
\begin{equation*}
\mathcal{H}_{\mathrm{VC}}
:=
\Span{
\ket{C}
\;\middle|\;
C_i + C_j \ge 1 \;\; \forall (i,j)\in E
}
\subseteq (\C^2)^{\otimes n}.
\end{equation*}
Equivalently, $\mathcal{H}_{\mathrm{VC}}$ is the image of the feasibility projector
\begin{equation*}
\Pi_{\mathrm{VC}}
=
\sum_{C\in \mathrm{VC}(G)} \ket{C}\!\bra{C}
=
\prod_{(i,j)\in E}\left(I-\ket{00}\bra{00}_{ij}\right),
\end{equation*}
where $\mathrm{VC}(G)$ denotes the set of all vertex covers of $G$.

Let
\begin{equation*}
\Nof{S}:=\{\, j\in V\setminus S \mid \exists\, i\in S \text{ such that } \{i,j\}\in E \,\}
\end{equation*}
denote the open neighbourhood of a set $S\subseteq V$.
For notational simplicity, we omit set braces, e.g. $\Nof{i}:=\Nof{\{i\}}$.
We define the projector onto the subspace where all neighbours of $i$ belong to the cover by
\begin{equation*}
\Pof{i}
\coloneqq
\bigotimes_{j\in \Nof{i}} \ket{1}\!\bra{1}_j \;\otimes\; I_{V\setminus \Nof{i}}.
\end{equation*}
The projector $\Pof{i}$ enforces the MVC flip condition: vertex $i$ may be flipped only when all of its neighbours are already included in the cover.
In that case, changing $C_i$ cannot uncover any edge incident to $i$.
\Cref{fig:star} illustrates the open neighbourhood of a vertex and the corresponding neighbourhood-controlled rotation.
\begin{figure}
    \centering
    \begin{subfigure}{0.42\linewidth}
        \centering
        \scalebox{0.8}{\begin{tikzpicture}[
scale=.75,
  line/.style={draw=black!60, line width=0.8pt},
  outer/.style={circle, fill=cyan!35, draw=black, minimum size=8mm},
  center/.style={circle, fill=red!35, draw=black, line width=0.8pt, minimum size=8mm}
]
  \node[center,
    path picture={
    \path[fill=red!35]
      (path picture bounding box.north west) rectangle
      (path picture bounding box.center);
    \path[fill=cyan!35]
      (path picture bounding box.west |- path picture bounding box.center)
      rectangle
      (path picture bounding box.south east);
  }
    ] (c) at (0,0) {$q_6$};

    \foreach \i/\q in {90/1,162/2,234/3,306/4,18/5} {
        \node[outer] (n\i) at (\i:2.3) {$q_{\q}$};
        \draw[line] (c) -- (n\i);
    }
    
\end{tikzpicture} }
        \caption{}
    \end{subfigure}
    \hspace{0.03\linewidth}
    \begin{subfigure}{0.42\linewidth}
        \centering
        \scalebox{1.2}{\input{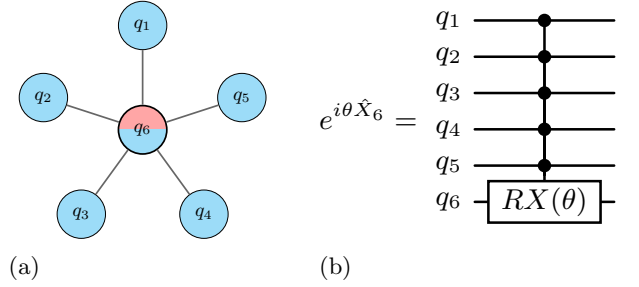}}
        \caption{}
    \end{subfigure}
    \caption{
    (a) Graph illustrating a vertex $q_6$ (red) and its neighbourhood $\Nof{6}$ (blue).
    The neighbours may have additional edges among themselves or to other vertices not shown.
    (b) Circuit representation of the constraint preserving unitary $U_6(\alpha)=e^{i\alpha \hat{X}_6}$.
    A single-qubit rotation is applied to vertex $6$, conditioned on all neighbouring qubits being in the state $\ket{1}$, thereby restricting the dynamics to the vertex-cover subspace.
    }
    \label{fig:star}
\end{figure}
We then define the neighbourhood-controlled Pauli operators
\begin{equation*}
    \hat{\sigma}_i \coloneqq \Pof{i} \, \sigma_i,
    \qquad
    \sigma\in\{X,Y\}.
\end{equation*}
Here $\Pof{i}$ is understood as an operator on the full $n$-qubit Hilbert space, and $\sigma_i$ acts nontrivially only on qubit $i$. Since $\Pof{i}$ acts only on the neighbours of $i$, it commutes with $\sigma_i$.

The operators $\hat{X}_i$ and $\hat{Y}_i$ generate single-vertex transitions that preserve feasibility: if $\ket{C}\in\mathcal{H}_{\mathrm{VC}}$, then $\hat{\sigma}_i\ket{C}$ is either zero or another state in $\mathcal{H}_{\mathrm{VC}}$. Indeed, the only possible nontrivial action occurs when all neighbours of $i$ are already in the cover, so flipping vertex $i$ cannot leave any incident edge uncovered. This is the property used throughout the construction.

Exponentiating these generators gives the corresponding constraint-preserving unitaries,
\begin{equation*}
U_i(\alpha)
=
e^{i\alpha \hat{\sigma}_i}
=
(I-\Pof{i}) + \Pof{i} e^{i\alpha \sigma_i}.
\end{equation*}
Thus, $U_i(\alpha)$ is a multi-controlled Pauli rotation on qubit $i$, conditioned on all neighbouring qubits being in the state $\ket{1}$. For the algorithms developed later, the generators are the projected Pauli-$X$ operators
\begin{equation*}
\hat{X}_i = \Pof{i} X_i.
\end{equation*}

A completely analogous construction applies to the maximum independent set (MIS) problem by exchanging occupied and unoccupied vertices. In the MIS convention, a bit value $1$ denotes inclusion in the independent set, and feasibility requires that no adjacent vertices are both occupied. The MVC and MIS unitary dynamics $U_{VC}, U_{MIS}$ are related by the global Pauli-$X$ transformation
\begin{equation*}
U_{\mathrm{MIS}} = X^{\otimes n} U_{\mathrm{VC}} (X^{\otimes n})^\dagger,
\qquad
\ket{\psi_{\mathrm{MIS}}} = X^{\otimes n}\ket{\psi_{\mathrm{VC}}},
\end{equation*}
which maps vertex covers to their complementary independent sets. Under this transformation, the control condition $\ket{1}\!\bra{1}$ on neighbouring qubits is replaced by $\ket{0}\!\bra{0}$, so the MIS dynamics is generated by the same construction with the roles of $0$ and $1$ exchanged.

\section{Reachability and connectivity of the feasible graph}
\label{sec:reachability_connectivity}

The projected operators $\{\hat X_i\}_{i\in V}$ define the allowed one-bit moves between feasible vertex covers.
Before using their uniform sum as a continuous-time quantum-walk Hamiltonian, we record the corresponding connectivity property: starting from the all-in-the-cover configuration, these moves can reach every feasible vertex cover.

For a coefficient vector $\vec{\alpha}\in\mathbb{R}^{|V|}$, define
\begin{equation*}
    H(\vec{\alpha})=\sum_{i\in V}\alpha_i\hat{X}_i,
\end{equation*}
and write 
$
    U(\vec{\alpha}) := e^{iH(\vec{\alpha})}
    $.
We define
\begin{equation*}
    \Omega := 1^{|V|}
\end{equation*}
for the all-in-the-cover configuration, and
\begin{equation*}
    \ket{\Omega} := \ket{1}^{\otimes |V|}
\end{equation*}
for the corresponding computational basis state. This state is always feasible
and corresponds to the maximal vertex cover.

The relevant connectivity statement is that the root configuration $\Omega$
is connected to every feasible configuration $C\in\mathrm{VC}(G)$.
Equivalently, every feasible basis state $\ket{C}$ is reachable from
$\ket{\Omega}$ by a sequence of feasibility-preserving single-vertex flips.

\begin{theorem}[Reachability of feasible vertex-cover states]
\label{thm:reachability}
Let $G=(V,E)$ and let
\begin{equation*}
    \mathcal{H}_{\mathrm{VC}}
    =
    \Span{\ket{C}\mid C\in \mathrm{VC}(G)}
\end{equation*}
be the feasible vertex-cover subspace. Then, for any vertex cover
$C\in \mathrm{VC}(G)$, there exists a parameter vector
$\boldsymbol{\alpha}^C$ with entries
\begin{equation*}
    \alpha_i^C \in \{0,\pi/2\}
\end{equation*}
such that
\begin{equation*}
    U(\boldsymbol{\alpha}^C)\ket{\Omega}
    =
    e^{i\sum_{i\in V}\alpha_i^C \hat{X}_i}\ket{\Omega}
    \propto
    \ket{C}.
\end{equation*}
In particular, every feasible basis state is reachable from $\ket{\Omega}$
by a sequence of constraint-preserving local flips.
\end{theorem}

\begin{proof}
Let $C\in \mathrm{VC}(G)$ be any feasible vertex cover, and define
\begin{equation*}
    S:=\{i\in V \mid C_i=0\}.
\end{equation*}
Thus, $S$ is the set of vertices removed from the all-in-the-cover
configuration.

Since $C$ is a vertex cover, no edge of $G$ can have both endpoints in $S$.
Therefore $S$ is an independent set.
Hence, for any distinct $i,j\in S$, one has $\{i,j\}\notin E$.
For these non-adjacent vertices, the corresponding projected generators commute:
\begin{equation*}
    [\hat{X}_i,\hat{X}_j]=0
    \qquad
    \text{for all } i,j\in S.
\end{equation*}
Indeed, if $i$ and $j$ are not adjacent, then the target qubit of each
projected flip is not contained in the control set of the other. Their
control projectors are diagonal and commute, and the single-qubit Pauli
operators act on distinct qubits.

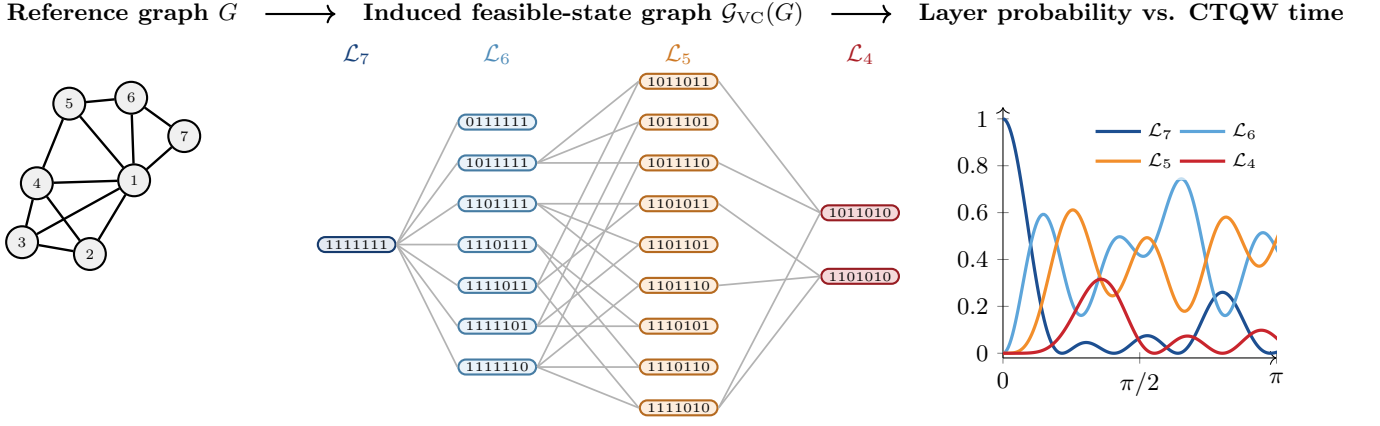
\begin{figure*}[t]
\centering
\definecolor{w7}{HTML}{1D4F91} 
\definecolor{w6}{HTML}{5DA5DA} 
\definecolor{w5}{HTML}{F28E2B} 
\definecolor{w4}{HTML}{C9252D} 

\begin{tikzpicture}[
    x=1cm,y=1cm,
    origv/.style={circle, draw=black, fill=gray!12, thick, minimum size=5.4mm, inner sep=0pt, font=\scriptsize},
    vcnode/.style={rounded rectangle, thick, inner xsep=2.2pt, inner ysep=1.2pt, font=\ttfamily\tiny},
    inducedge/.style={black!30, line width=0.65pt},
    origedge/.style={black, line width=0.9pt},
    lab/.style={font=\small\bfseries},
    layerlab/.style={font=\small},
    layer7node/.style={vcnode, draw=w7!75!black, fill=w7!12},
    layer6node/.style={vcnode, draw=w6!75!black, fill=w6!13},
    layer5node/.style={vcnode, draw=w5!75!black, fill=w5!16},
    layer4node/.style={vcnode, draw=w4!75!black, fill=w4!18},
    layer7lab/.style={layerlab, text=w7!85!black},
    layer6lab/.style={layerlab, text=w6!85!black},
    layer5lab/.style={layerlab, text=w5!85!black},
    layer4lab/.style={layerlab, text=w4!85!black}
]

\node[lab] (titleG) at (-5.1,2.05) {Reference graph $G$};

\begin{scope}[shift={(-5.1,-0.15)}, scale=0.78, every node/.style={transform shape}]
\node[origv] (v1) at ( 0.20, 0.00) {$1$};
\node[origv] (v2) at (-0.55,-1.25) {$2$};
\node[origv] (v3) at (-1.70,-1.05) {$3$};
\node[origv] (v4) at (-1.45,-0.05) {$4$};
\node[origv] (v5) at (-0.90, 1.30) {$5$};
\node[origv] (v6) at ( 0.15, 1.40) {$6$};
\node[origv] (v7) at ( 1.05, 0.75) {$7$};

\draw[origedge] (v1)--(v2);
\draw[origedge] (v1)--(v3);
\draw[origedge] (v1)--(v4);
\draw[origedge] (v1)--(v5);
\draw[origedge] (v1)--(v6);
\draw[origedge] (v1)--(v7);
\draw[origedge] (v2)--(v3);
\draw[origedge] (v2)--(v4);
\draw[origedge] (v3)--(v4);
\draw[origedge] (v4)--(v5);
\draw[origedge] (v5)--(v6);
\draw[origedge] (v6)--(v7);
\end{scope}

\node[lab] (titleVC) at (1.,2.05) {Induced feasible-state graph $\inducedG$};

\begin{scope}[shift={(-2.0,-0.70)}, scale=1.2]

\node[layer7lab] at (0.00, 1.85) {$\mathcal{L}_7$};
\node[layer6lab] at (1.55, 1.85) {$\mathcal{L}_6$};
\node[layer5lab] at (3.55, 1.85) {$\mathcal{L}_5$};
\node[layer4lab] at (5.55, 1.85) {$\mathcal{L}_4$};

\node[layer7node] (s1111111) at (0.00,0.00-.25) {$1111111$};

\node[layer6node] (s0111111) at (1.55, 1.35-.25) {$0111111$};
\node[layer6node] (s1011111) at (1.55, 0.90-.25) {$1011111$};
\node[layer6node] (s1101111) at (1.55, 0.45-.25) {$1101111$};
\node[layer6node] (s1110111) at (1.55, 0.00-.25) {$1110111$};
\node[layer6node] (s1111011) at (1.55,-0.45-.25) {$1111011$};
\node[layer6node] (s1111101) at (1.55,-0.90-.25) {$1111101$};
\node[layer6node] (s1111110) at (1.55,-1.35-.25) {$1111110$};

\node[layer5node] (s1011011) at (3.55, 1.80-.25) {$1011011$};
\node[layer5node] (s1011101) at (3.55, 1.35-.25) {$1011101$};
\node[layer5node] (s1011110) at (3.55, 0.90-.25) {$1011110$};
\node[layer5node] (s1101011) at (3.55, 0.45-.25) {$1101011$};
\node[layer5node] (s1101101) at (3.55, 0.00-.25) {$1101101$};
\node[layer5node] (s1101110) at (3.55,-0.45-.25) {$1101110$};
\node[layer5node] (s1110101) at (3.55,-0.90-.25) {$1110101$};
\node[layer5node] (s1110110) at (3.55,-1.35-.25) {$1110110$};
\node[layer5node] (s1111010) at (3.55,-1.80-.25) {$1111010$};

\node[layer4node] (s1011010) at (5.55, 0.35-.25) {$1011010$};
\node[layer4node] (s1101010) at (5.55,-0.35-.25) {$1101010$};

\draw[inducedge] (s1111111.east)--(s0111111.west);
\draw[inducedge] (s1111111.east)--(s1011111.west);
\draw[inducedge] (s1111111.east)--(s1101111.west);
\draw[inducedge] (s1111111.east)--(s1110111.west);
\draw[inducedge] (s1111111.east)--(s1111011.west);
\draw[inducedge] (s1111111.east)--(s1111101.west);
\draw[inducedge] (s1111111.east)--(s1111110.west);

\draw[inducedge] (s1011111.east)--(s1011011.west);
\draw[inducedge] (s1111011.east)--(s1011011.west);

\draw[inducedge] (s1011111.east)--(s1011101.west);
\draw[inducedge] (s1111101.east)--(s1011101.west);

\draw[inducedge] (s1011111.east)--(s1011110.west);
\draw[inducedge] (s1111110.east)--(s1011110.west);

\draw[inducedge] (s1101111.east)--(s1101011.west);
\draw[inducedge] (s1111011.east)--(s1101011.west);

\draw[inducedge] (s1101111.east)--(s1101101.west);
\draw[inducedge] (s1111101.east)--(s1101101.west);

\draw[inducedge] (s1101111.east)--(s1101110.west);
\draw[inducedge] (s1111110.east)--(s1101110.west);

\draw[inducedge] (s1110111.east)--(s1110101.west);
\draw[inducedge] (s1111101.east)--(s1110101.west);

\draw[inducedge] (s1110111.east)--(s1110110.west);
\draw[inducedge] (s1111110.east)--(s1110110.west);

\draw[inducedge] (s1111011.east)--(s1111010.west);
\draw[inducedge] (s1111110.east)--(s1111010.west);

\draw[inducedge] (s1011011.east)--(s1011010.west);
\draw[inducedge] (s1011110.east)--(s1011010.west);
\draw[inducedge] (s1111010.east)--(s1011010.west);

\draw[inducedge] (s1101011.east)--(s1101010.west);
\draw[inducedge] (s1101110.east)--(s1101010.west);
\draw[inducedge] (s1111010.east)--(s1101010.west);

\end{scope}

\node[lab, align=center] (titleP) at (8.25,2.05) {Layer probability vs. CTQW time};

\begin{scope}[shift={(6.55,-2.5)}]
\begin{axis}[
  width=5.2cm,
  height=4.9cm,
  xlabel style={at={(axis description cs:0.5,-0.14)}, anchor=north},
  ylabel style={at={(axis description cs:-0.14,0.5)}, anchor=south, rotate=90},
  xmin=0, xmax=3.14159265,
  xtick={0,1.5707963,3.14159265},
  xticklabels={$0$,$\pi/2$,$\pi$},
  ymin=-0.02, ymax=1.05,
  ytick={0,0.2,0.4,0.6,0.8,1.0},
  axis lines=left,
  axis line style={->},
  grid=none,
  tick label style={font=\small},
  label style={font=\normalsize},
  legend style={
    font=\scriptsize,
    draw=none,
    fill=white,
    fill opacity=0.65,
    text opacity=1,
    at={(0.98,0.98)},
    anchor=north east,
    row sep=1pt
  },
  legend columns=2,
  every axis plot/.append style={line width=1.2pt},
]
\addplot+[color=w7, mark=none, solid]
  table[
    x=beta,
    y=P_k,
    col sep=comma,
    restrict expr to domain={\thisrow{k}}{7:7}
  ] {final_reduced_plotting/evolution_times/P_k_vs_beta.csv};
\addlegendentry{$\mathcal{L}_7$}

\addplot+[color=w6, mark=none, solid]
  table[
    x=beta,
    y=P_k,
    col sep=comma,
    restrict expr to domain={\thisrow{k}}{6:6}
  ] {final_reduced_plotting/evolution_times/P_k_vs_beta.csv};
\addlegendentry{$\mathcal{L}_6$}

\addplot+[color=w5, mark=none, solid]
  table[
    x=beta,
    y=P_k,
    col sep=comma,
    restrict expr to domain={\thisrow{k}}{5:5}
  ] {final_reduced_plotting/evolution_times/P_k_vs_beta.csv};
\addlegendentry{$\mathcal{L}_5$}

\addplot+[color=w4, mark=none, solid]
  table[
    x=beta,
    y=P_k,
    col sep=comma,
    restrict expr to domain={\thisrow{k}}{4:4}
  ] {final_reduced_plotting/evolution_times/P_k_vs_beta.csv};
\addlegendentry{$\mathcal{L}_4$}
\end{axis}
\end{scope}

\draw[->, thick] ($(titleG.east)+(0.25,0)$) -- ($(titleVC.west)+(-0.25,0)$);
\draw[->, thick] ($(titleVC.east)+(0.25,0)$) -- ($(titleP.west)+(-0.25,0)$);

\end{tikzpicture}
\vspace{-1\baselineskip}
\caption{
Example of the implicitly induced feasible-state graph underlying the
constrained continuous-time quantum walk (CTQW). The input graph
$G=(V,E)$ in the left panel defines the feasible-state graph
$\inducedG$ in the middle panel, whose vertices are the
feasible vertex covers of $G$ ordered by Hamming weight and whose
edges connect pairs of covers that differ by a single feasibility-preserving bit flip. The walk is initialised in the all-ones
state $\ket{\Omega} = \ket{1}^{\otimes |V|}$ and propagates
coherently down through layers of decreasing weight.
The right panel
shows the layer-resolved probabilities
$P_k(t) = \sum_{C\in\mathcal{L}_k}|\langle C|\psi(t)\rangle|^2$:
amplitude reaches the lowest-weight occupied layer before
$t \approx \pi/2$ and subsequently flows back toward the root, with
interference patterns emerging at later times. The induced graph is an
explanatory object only and is not constructed explicitly by the
algorithm.}
\label{fig:induced_graph_small_example}
\end{figure*}

Now choose
\begin{equation*}
    \alpha_i^C=
    \begin{cases}
        \pi/2, & i\in S,\\
        0, & i\notin S.
    \end{cases}
\end{equation*}
Since the generators with nonzero coefficients commute, the exponential
factorises as
\begin{equation*}
    e^{i\sum_{i\in V}\alpha_i^C \hat{X}_i}
    =
    \prod_{i\in S} e^{i\frac{\pi}{2}\hat{X}_i}.
\end{equation*}

Starting from the reference state $\ket{\Omega}$, all neighbours of every
vertex are initially included in the cover. Moreover, because $S$ is an
independent set, flipping one vertex in $S$ never changes the control
condition for any other vertex in $S$.
Therefore, on the states encountered along this sequence, each projected generator $\hat{X}_i=\Pof{i}X_i$ with $i\in S$ acts as the ordinary Pauli-$X$ operator on qubit $i$.
Consequently,
\begin{equation*}
    \prod_{i\in S} e^{i\frac{\pi}{2}\hat{X}_i}\ket{\Omega}
    =
    i^{|S|}
    \left(\prod_{i\in S} X_i\right)\ket{\Omega}
    =
    i^{|S|}
    \ket{0}_{S}\ket{1}_{V\setminus S}.
\end{equation*}
Since $\ket{0}_{S}\ket{1}_{V\setminus S}=\ket{C}$ and
$|S|=|V|-|C|$, we obtain
\begin{equation*}
    e^{i\sum_{i\in V}\alpha_i^C \hat{X}_i}\ket{\Omega}
    =
    i^{|V|-|C|}
    \ket{C}.
\end{equation*}
This proves the claim.
\end{proof}

To state the graph-theoretic consequence of the reachability theorem, define the induced feasible-state graph
\begin{equation*}
    \inducedG
    =
    \bigl(\mathcal V_{\mathrm{VC}},\mathcal E_{\mathrm{VC}}\bigr),
\end{equation*}
where
$
    \mathcal V_{\mathrm{VC}}
    =
    \mathrm{VC}(G)
$
is the set of all vertex covers of $G$.
Two feasible configurations $C,C'\in\mathrm{VC}(G)$ are adjacent, i.e. $(C,C')\in\mathcal E_{\mathrm{VC}}$, if and only if
\begin{equation*}
    d_{\mathrm H}(C,C')=1,
\end{equation*}
where $d_{\mathrm H}$ denotes Hamming distance.
Equivalently, adjacent vertices of $\inducedG$ are exactly those feasible covers that differ by one allowed projected Pauli-$X$ move generated by some $\widehat X_i$.

\Cref{thm:reachability} has an immediate graph-theoretic interpretation.

\begin{corollary}[Connectivity from the root]
    \label{cor:root_connectivity}
    Every feasible configuration $C\in \mathrm{VC}(G)$ is connected to the root configuration $\Omega = 1^{|V|}$ by a path in the induced feasible-state graph $\inducedG$.
    Equivalently, all vertices of $\inducedG$ lie in the same connected component as the root.
\end{corollary}

\begin{proof}
The proof of \cref{thm:reachability} explicitly constructs a sequence of allowed one-bit flips taking $\Omega$ to $C$.
Each such flip preserves feasibility and therefore corresponds to an edge of $\inducedG$.
Hence these flips define a path from the root to $C$ in the induced feasible-state graph.
\end{proof}

The argument also shows that $\inducedG$ is naturally rooted at $\Omega$.
Every feasible configuration is obtained from the all-in-the-cover configuration by removing an independent set of vertices, and hence admits at least one monotone path from $\Omega$.
This connectivity property is the structural reason why the induced-graph picture is useful for optimisation: the constrained dynamics do not merely remain inside the feasible subspace; from $\Omega$, the allowed moves connect to every feasible cover.

\section{Induced layered graph and continuous-time walk perspective}
\label{sec:induced_graph_walk}

The local projected flips from \cref{sec:ConstrainedOperators} define more than a feasibility-preserving operation on bit strings: collectively, they define a graph structure on the set of feasible vertex covers.
This graph is the natural state space of the constrained dynamics.
Its vertices are feasible covers, and its edges are the single-bit moves generated by the operators $\{\widehat X_i\}_{i\in V}$.

The uniform generator used in the algorithm is
\begin{equation}
H_{\mathrm{MVC}}
\;=\;\sum_{i\in V}\widehat{X}_{i},
\label{eq:H_mvc}
\end{equation}
where $\widehat{X}_{i}$ is the multi-controlled Pauli-$X$ operator that
flips bit $i$ if and only if the resulting configuration remains a
vertex cover of $G$.

A useful way to understand the constrained dynamics is to step back
from the original graph $G$ and instead analyze the geometry of the
\emph{feasible configurations} themselves. This perspective makes the
support of the dynamics explicit and recasts the unitary evolution as
a CTQW on a layered state graph. \Cref{fig:induced_graph_small_example} illustrates the
correspondence on a small representative instance.

\subsection{The induced feasible-state graph}
Having defined the induced feasible-state graph $\inducedG$ and established its root connectivity, we now analyse its canonical Hamming-weight layering.
For $k\in\{0,\ldots,|V|\}$, define
\begin{equation*}
    \mathcal L_k
    =
    \{C\in\mathrm{VC}(G): |C|=k\},
    \qquad
    |C|=\sum_i C_i .
\end{equation*}
Because adjacent feasible configurations differ by a single bit, every edge of $\inducedG$ connects layers $\mathcal L_k$ and $\mathcal L_{k\pm1}$.
The graph is therefore bipartite with respect to the parity of the Hamming weight.
Consequently, the spectrum of the adjacency matrix
\begin{equation*}
    A=A(\inducedG)
\end{equation*}
is symmetric about zero, with eigenvalues occurring in pairs $\pm\lambda$.

The root configuration $\Omega$ lies in the unique top layer $\mathcal L_{|V|}$.
The lowest occupied layer is $\mathcal L_{\tau(G)}$, where $\tau(G)$ is the minimum vertex-cover number of $G$.
Minimum vertex covers are exactly the elements of $\mathcal L_{\tau(G)}$.
Thus the depth of $\inducedG$ is
\begin{equation*}
    D(G)
    =
    |V|-\tau(G)
    =
    \alpha(G),
\end{equation*}
where $\alpha(G)$ is the independence number of $G$.
The last equality follows from the complement relation between vertex covers and independent sets.
This structure is visible in \cref{fig:induced_graph_small_example}: the root configuration occupies the top layer, while the minimum vertex covers appear in the bottom layer $\mathcal L_{\tau(G)}$.

The MIS picture is obtained by bitwise complementation.
Under the global Pauli-$X$ map
\begin{equation*}
    \ket{C}\mapsto X^{\otimes |V|}\ket{C},
\end{equation*}
vertex covers are mapped to their complementary independent sets.
The induced feasible-state graph is unchanged up to this relabelling, but the Hamming-weight layering is inverted.
The MVC root $\Omega=1^{|V|}$ corresponds to the MIS root $\Omega_{\mathrm{MIS}}=0^{|V|}$, and the MVC minimum-cover layer corresponds to the MIS maximum-independent-set layer.

The original graph $G$ enters only through the feasibility constraint.
It determines which one-bit flips preserve vertex-cover feasibility, and therefore determines the sizes of the layers $|\mathcal L_k|$ and the connectivity between them.
The continuous-time quantum walk itself takes place on $\inducedG$, not on the original graph $G$.

\subsection{Continuous-time dynamics on the induced graph}

In the feasible computational basis $\{\ket{C} : C \in \mathrm{VC}(G)\}$,
the Hamiltonian in \cref{eq:H_mvc} has matrix elements
\begin{equation*}
\bra{C'} H_{\mathrm{MVC}} \ket{C} \;=\;
\begin{cases}
1 & \text{if } (C,C') \in \mathcal{E}_{\mathrm{VC}},\\
0 & \text{otherwise},
\end{cases}
\end{equation*}
so that
\begin{equation}
    H_{\mathrm{MVC}}\big|_{\mathcal{H}_{\mathrm{VC}}}
    \;=\; A(\inducedG).
    \label{eq:H_equals_A}
\end{equation}
The constrained unitary
\begin{equation*}
U(t)\;=\;e^{i t H_{\mathrm{MVC}}}
\end{equation*}
therefore acts as a continuous-time quantum walk on
$\inducedG$.

\subsection{Layer probabilities and the walk-count expansion}
\label{subsec:walk_count}

Starting from $\ket{\Omega}$ at $t=0$, the amplitude on a feasible
cover $C$ is given by the following path expansion
\begin{equation*}
\bra{C}U(t)\ket{\Omega}=\sum_{m=0}^{\infty}\frac{(it)^{m}}{m!}\bra{C}H_{\mathrm{MVC}}^m\ket{\Omega}
= \sum_{m=0}^{\infty}\frac{(it)^{m}}{m!}\,\gamma^{C}_{m},
\end{equation*}
where $\gamma^{C}_{m} = (A^{m})_{C,\Omega}$ is the number of walks of
length $m$ from $\Omega$ to $C$ on $\inducedG$.
Because $\inducedG$ is
bipartite with respect to Hamming-weight parity, $\gamma^{C}_{m}$
vanishes unless $m$ has the same parity as the graph distance
\begin{equation*}
d(C) = |V| - |C|
\end{equation*}
from the root to $C$.
Writing $d=d(C)$, setting $m=d+2j$ and factoring out the
overall phase, the amplitude is real up to $i^{d}$, and the
single-cover probability admits the following expansion
\begin{equation*}
P_{C}(t) =
\Bigl(\sum_{j=0}^{\infty}\frac{(-1)^{j}\, t^{d+2j}}{(d+2j)!}\,
       \gamma^{C}_{d+2j}\Bigr)^{2}.
\end{equation*}

Layer probabilities follow by summation. Throughout this section we
parameterise layers by the graph distance $d$ from the root rather
than by Hamming weight; the two are equivalent through
$d = |V| - k$, with $d = 0$ corresponding to the root configuration $\Omega$ and
$d = D(G)$ to the minimum-cover layer. The probability of finding the
walker in the layer at distance $d$ is
\begin{equation}
P_{d}(t) 
=\sum_{C \in \mathcal{L}_{|V|-d}}
\Bigl(\sum_{j=0}^{\infty}\frac{(-1)^{j} t^{d+2j}}{(d+2j)!}
       \gamma^{C}_{d+2j}\Bigr)^{2}.
\label{eq:Pd_expansion}
\end{equation}

\subsection{Transport and layer hitting times}
\label{subsec:ballistic}

\begin{figure}
\centering
\definecolor{seqBlueDark}{RGB}{33,102,172}
\definecolor{seqBlue}{RGB}{67,147,195}
\definecolor{seqBlueLight}{RGB}{146,197,222}
\definecolor{seqGreen}{RGB}{127,191,123}
\definecolor{seqOrange}{RGB}{244,165,130}
\definecolor{seqRed}{RGB}{178,24,43}
\definecolor{seqGray}{RGB}{150,150,150}


\begin{tikzpicture}
\begin{axis}[
  width=8.2cm,
  height=5.1cm,
  xlabel={Layer index $\ell$},
  ylabel={Peak time $t^*_l$},
  xmin=-0.2,
  xmax=13.2,
  ymin=0,
  ymax=1.32,
  xtick={0,2,4,6,8,10,12},
  ytick={0,0.2,0.4,0.6,0.8,1.0,1.2},
  grid=major,
  grid style={line width=0.25pt, draw=gray!18},
  axis line style={line width=0.45pt},
  tick align=outside,
  tick style={line width=0.45pt, black},
  major tick length=2pt,
  legend columns=3,
  legend cell align=left,
  legend style={
    font=\scriptsize,
    draw=none,
    fill=white,
    fill opacity=0.86,
    text opacity=1,
    at={(0.97,0.97)},
    anchor=north east,
    column sep=0.35em,
    row sep=0.02em
  },
  every axis plot/.append style={
    line width=1.05pt,
    mark size=1.9pt,
    error bars/error bar style={line width=0.5pt}
  }
]

\addplot+[color=seqBlueDark, mark=o, mark options={fill=white}, solid,
  error bars/.cd, y dir=both, y explicit]
  coordinates {
    (0,0.973682) +- (0,0.0121414)
    (1,0.616683) +- (0,0.011539)
    (2,0.355413) +- (0,0)
    (3,0) +- (0,0.011539)
  };
\addlegendentry{$d=4$}

\addplot+[color=seqBlue, mark=square, mark options={fill=white}, solid,
  error bars/.cd, y dir=both, y explicit]
  coordinates {
    (0,1.139) +- (0,0.053906)
    (1,0.900091) +- (0,0.0510995)
    (2,0.655745) +- (0,0.044606)
    (3,0.465572) +- (0,0.020228)
    (4,0.287186) +- (0,0.0176378)
    (5,0) +- (0,0.0141462)
  };
\addlegendentry{$d=6$}

\addplot+[color=seqBlueLight, mark=triangle, mark options={fill=white}, solid,
  error bars/.cd, y dir=both, y explicit]
  coordinates {
    (0,1.16625) +- (0,0.0514568)
    (1,0.999051) +- (0,0.0529432)
    (2,0.807994) +- (0,0.0543947)
    (3,0.63379) +- (0,0.0402684)
    (4,0.493288) +- (0,0.0273309)
    (5,0.373577) +- (0,0.0136654)
    (6,0.246644) +- (0,0.0136654)
    (7,0) +- (0,0.00515577)
  };
\addlegendentry{$d=8$}

\addplot+[color=seqGreen!85!black, mark=diamond, mark options={fill=white}, solid,
  error bars/.cd, y dir=both, y explicit]
  coordinates {
    (0,1.17597) +- (0,0.0414409)
    (1,1.04882) +- (0,0.0408999)
    (2,0.903218) +- (0,0.0487122)
    (3,0.747541) +- (0,0.0416634)
    (4,0.620608) +- (0,0.0272975)
    (5,0.511928) +- (0,0.0176811)
    (6,0.41185) +- (0,0.0136487)
    (7,0.311773) +- (0,0.0114109)
    (8,0.210226) +- (0,0.0114109)
    (9,0) +- (0,0.182865)
  };
\addlegendentry{$d=10$}

\addplot+[color=seqOrange!90!black, mark=pentagon, mark options={fill=white}, solid,
  error bars/.cd, y dir=both, y explicit]
  coordinates {
    (0,1.18318) +- (0,0.0355409)
    (1,1.08476) +- (0,0.033471)
    (2,0.97135) +- (0,0.0423558)
    (3,0.836926) +- (0,0.0418705)
    (4,0.708744) +- (0,0.0349462)
    (5,0.605117) +- (0,0.0193286)
    (6,0.526876) +- (0,0.0150054)
    (7,0.433445) +- (0,0.00663598)
    (8,0.357285) +- (0,0.00663598)
    (9,0.281125) +- (0,0.00663598)
    (10,0.195186) +- (0,0.0117565)
    (11,0) +- (0,0.0728768)
  };
\addlegendentry{$d=12$}

\addplot+[color=seqRed, mark=*, mark options={fill=seqRed}, solid,
  error bars/.cd, y dir=both, y explicit]
  coordinates {
    (0,1.26933) +- (0,0)
    (1,1.16778) +- (0,0)
    (2,1.10432) +- (0,0.0126933)
    (3,1.01546) +- (0,0)
    (4,0.888531) +- (0,0)
    (5,0.774292) +- (0,0.0126933)
    (6,0.660052) +- (0,0)
    (7,0.583892) +- (0,0)
    (8,0.507732) +- (0,0)
    (9,0.431572) +- (0,0)
    (10,0.355413) +- (0,0)
    (11,0.279253) +- (0,0)
    (12,0.177706) +- (0,0)
    (13,0) +- (0,0.0126933)
  };
\addlegendentry{$d=14$}

\end{axis}
\end{tikzpicture}
    \caption{
    First-peak time $t_\ell^*$ as a function of layer index $\ell$, measured from the minimum-cover layer, for several depths $D(G)$ of $\inducedG$.
    Markers show the mean over instances; error bars indicate the corresponding standard deviation.
    }
    \label{fig:firstpeak_scaling}
\end{figure}
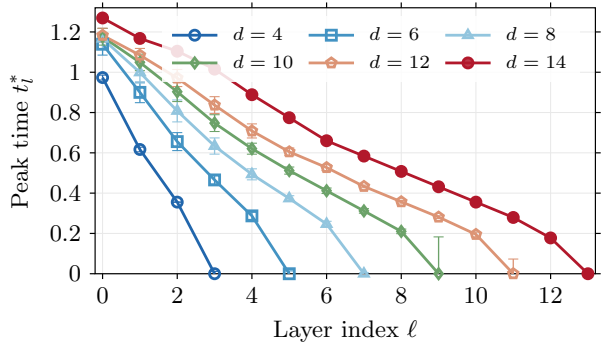

The walk-count expansion of the previous subsection makes the
short-time behaviour transparent. Retaining only the lowest-order term
in \cref{eq:Pd_expansion} gives
\begin{equation}
  P_d(t) \;\sim\; |\mathcal{L}_{|V|-d}| t^{2d},
  \label{eq:short-time-onset}
\end{equation}
since the number of walks of length $d$ from $\Omega$ to a cover $C$ at distance $d$ from the root is exactly $\gamma_{d}^C=d!$.
The factor $t^{2d}$ is the operational signature of the layered geometry: the
amplitude on a layer at distance $d$ is suppressed as $t^{2d}$ at short
times, so deeper layers remain essentially unpopulated until the
wavefront generated by $H_{\mathrm{MVC}}$ has had time to reach them.
This is exactly the qualitative picture of \cref{fig:induced_graph_small_example}:
amplitude leaves the root, propagates coherently downward through
layers of decreasing Hamming weight, and only later returns creating interference.

\paragraph{First-peak times.}
The location of the first maximum of $P_d(t)$ lies outside the regime
of validity of \cref{eq:short-time-onset}: by the time a layer
reaches its peak, many more terms of the walk-count series
contribute, and a closed-form prediction 
would require resumming the full expansion or deriving a controlled truncation.
We therefore characterise the peak times numerically.
Denote by $t^\star_\ell$ the smallest positive time at which $P_{D(G)-\ell}(t)$ attains a local maximum.
Here $\ell$ is the layer offset from the minimum-cover layer $\mathcal{L}_{\tau(G)}$: $\ell=0$ denotes the minimum-cover layer, and increasing $\ell$ moves upward toward the root.

Across all instances and depths considered, the first-peak times follow
a linear trend in the layer index,
\begin{equation}
  t^\star_\ell\big(D(G)\big)
  \;\approx\;
  t^0_{D(G)} \;-\; \frac{\ell}{v_{D(G)}}.
  \label{eq:hitting-linear}
\end{equation}
As shown in \cref{fig:firstpeak_scaling}, for fixed depth the dependence on $\ell$ is approximately linear, corresponding to a wavefront that advances through successive layers with an approximately constant effective group velocity $v_{D(G)}$, measured in layers per unit walk time.
The intercept $t^0_{D(G)}$, which is the extrapolated time at which the wavefront reaches the minimum-cover layer, increases with the depth $D(G)$ of the induced graph, while the magnitude of the fitted slope decreases as $D(G)$ grows.

\begin{figure}
\centering
\definecolor{seqBlueDark}{RGB}{33,102,172}
\definecolor{seqBlue}{RGB}{67,147,195}
\definecolor{seqBlueLight}{RGB}{146,197,222}
\definecolor{seqGreen}{RGB}{127,191,123}
\definecolor{seqOrange}{RGB}{244,165,130}
\definecolor{seqRed}{RGB}{178,24,43}
\definecolor{seqGray}{RGB}{150,150,150}

\begin{tikzpicture}
\begin{semilogyaxis}[
  width=8.2cm,
  height=5.1cm,
  xlabel={$n$},
  ylabel={$\langle P^{\star}_{\leq l}\rangle$},
  xmin=9.5,
  xmax=30.5,
  ymin=1.2e-2,
  ymax=1.6,
  xtick={10,12,14,16,18,20,22,24,26,28,30},
  ytick={1e-2,1e-1,1},
  grid=major,
  grid style={line width=0.25pt, draw=gray!18},
  minor y tick num=1,
  minor grid style={line width=0.2pt, draw=gray!10},
  axis line style={line width=0.45pt},
  tick align=outside,
  tick style={line width=0.45pt, black},
  major tick length=2pt,
  legend columns=3,
  legend cell align=left,
  legend style={
    font=\scriptsize,
    draw=none,
    fill=white,
    fill opacity=0.86,
    text opacity=1,
    at={(0.03,0.06)},
    anchor=south west,
    column sep=0.45em,
    row sep=0.05em
  },
  every axis plot/.append style={
    line width=1.05pt,
    mark size=1.9pt
  }
]

\addlegendimage{color=seqRed, mark=o, mark options={fill=white}, solid}
\addlegendentry{$l=1$}
\addlegendimage{color=seqOrange!90!black, mark=square, mark options={fill=white}, solid}
\addlegendentry{$l=2$}
\addlegendimage{color=seqGreen!85!black, mark=triangle, mark options={fill=white}, solid}
\addlegendentry{$l=3$}
\addlegendimage{color=seqBlueLight, mark=diamond, mark options={fill=white}, solid}
\addlegendentry{$l=4$}
\addlegendimage{color=seqBlueDark, mark=pentagon, mark options={fill=white}, solid}
\addlegendentry{$l=5$}
\addlegendimage{color=seqGray, mark=o, mark options={fill=white}, solid}
\addlegendentry{$l=13$}

\addplot+[color=seqRed, mark=o, mark options={fill=white}, only marks, forget plot]
  coordinates {
    (10,0.409608)
    (12,0.411496)
    (14,0.261845)
    (16,0.122281)
    (18,0.100113)
    (20,0.0761498)
    (22,0.0694886)
    (24,0.0505598)
    (26,0.0294747)
    (28,0.021828)
    (30,0.0166657)
  };
\addplot+[color=seqRed, no marks, solid, domain=10:30, samples=100, forget plot]
  {2.6212535*exp(-0.17076869*x)};

\addplot+[color=seqOrange!90!black, mark=square, mark options={fill=white}, only marks, forget plot]
  coordinates {
    (10,0.877374)
    (12,0.857844)
    (14,0.75799)
    (16,0.599378)
    (18,0.484374)
    (20,0.395458)
    (22,0.345982)
    (24,0.29781)
    (26,0.225226)
    (28,0.172819)
    (30,0.139042)
  };
\addplot+[color=seqOrange!90!black, no marks, solid, domain=10:30, samples=100, forget plot]
  {2.6639619*exp(-0.095434165*x)};

\addplot+[color=seqGreen!85!black, mark=triangle, mark options={fill=white}, only marks, forget plot]
  coordinates {
    (10,0.996655)
    (12,0.985596)
    (14,0.966929)
    (16,0.915937)
    (18,0.839489)
    (20,0.783152)
    (22,0.720909)
    (24,0.664364)
    (26,0.576513)
    (28,0.49094)
    (30,0.425646)
  };
\addplot+[color=seqGreen!85!black, no marks, solid, domain=10:30, samples=100, forget plot]
  {1.7216939*exp(-0.042670855*x)};

\addplot+[color=seqBlueLight, mark=diamond, mark options={fill=white}, only marks, forget plot]
  coordinates {
    (10,0.999937)
    (12,0.999668)
    (14,0.997833)
    (16,0.991267)
    (18,0.972634)
    (20,0.954801)
    (22,0.926733)
    (24,0.897205)
    (26,0.849678)
    (28,0.788904)
    (30,0.733579)
  };
\addplot+[color=seqBlueLight, no marks, solid, domain=10:30, samples=100, forget plot]
  {1.2265298*exp(-0.01466278*x)};

\addplot+[color=seqBlueDark, mark=pentagon, mark options={fill=white}, only marks, forget plot]
  coordinates {
    (10,1)
    (12,0.999999)
    (14,0.99994)
    (16,0.999581)
    (18,0.997472)
    (20,0.994309)
    (22,0.987045)
    (24,0.978347)
    (26,0.9636)
    (28,0.939985)
    (30,0.91327)
  };
\addplot+[color=seqBlueDark, no marks, solid, domain=10:30, samples=100, forget plot]
  {1.0591673*exp(-0.0039348503*x)};

\addplot+[color=seqGray, mark=o, mark options={fill=white}, only marks, forget plot]
  coordinates {
    (28,1)
    (30,1)
  };
\addplot+[color=seqGray, no marks, solid, domain=28:30, samples=2, forget plot]
  {1};

\end{semilogyaxis}
\end{tikzpicture}
    \caption{
    Cumulative bottom-$\ell$ peak probability
    $P^{\star}_{\leq \ell}(n) \equiv \max_t \sum_{j \leq \ell} P_{D(G)-j}(t)$
    as a function of $n$ at fixed offset $\ell$ from the minimum-cover layer.
    Markers show numerical data; lines show exponential fits
    $A_\ell e^{-\alpha_\ell n}$.
    For $\ell \geq 5$, the fitted decay rate $\alpha_\ell$ is statistically indistinguishable from zero over the simulated range. Error bars show the corresponding standard deviation.
    }
    \label{fig:cummulative_scaling}
\end{figure}
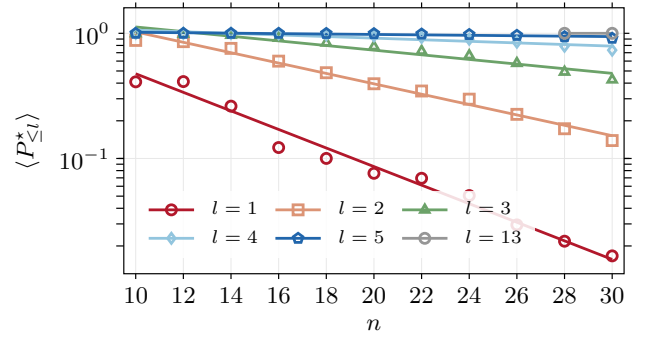

\paragraph{From transport to the greedy criterion.}
The transport picture builds up the intuition behind the central
difficulty of any single-shot scheme. \Cref{eq:hitting-linear}
says that the wavefront does reach the minimum-cover layer
$\mathcal{L}_{\tau(G)}$, and does so at a time that is not too large.
One might therefore hope to prepare the CTQW state at $t \approx
t^0_{D(G)}$, measure, and read off a minimum cover directly
with probability $P_{D(G)}\big(t^0_{D(G)}\big)$.
As numerically analysed in \cref{fig:cummulative_scaling}, reaching a layer is not the same as concentrating probability on it: although the wavefront arrives at $\mathcal{L}_{\tau(G)}$, the
\emph{single-layer} probability of being exactly on the minimum-cover
layer is exponentially small in the system size.

Hitting the minimum-cover layer is, structurally, a graph-traversal problem
related to continuous-time quantum-walk transport, as illustrated by the
glued-trees construction~\cite{Childs_2003}. A CTQW started from the
all-in-the-cover root must propagate across the layers of
$\inducedG$ and deposit appreciable amplitude where a
solution lives.
In that construction, the walk starts at the \textsc{entrance} root, spreads ballistically through
columns whose sizes grow and then shrink, and arrives at the \textsc{exit}
root in linear time with inverse-polynomial probability. Two features make
this succeed. First, the tree symmetry induces a column
partition, so the dynamics reduce \emph{exactly} to a $1$D tight-binding chain
on which transport is ballistic. Second, and decisively, the target is a single 
\textsc{exit} root.

Both properties are absent in $\inducedG$. The reduction to a $1$D chain
is only approximate; and, more importantly, the minimum-cover layer is not a
single exit vertex. The same wavefront traversal that
concentrates at a single root in the glued-trees problem fails to concentrate
on a bottom layer here; it is the favourable terminal-node geometry, not the
transport itself, that the vertex-cover graph lacks.
 Rather than demanding the exact minimum-cover layer, one may accept any cover within the lowest few
layers and ask for the cumulative probability of landing in the bottom
$\ell+1$ layers
\begin{equation*}
  P^\star_{\leq \ell}(n)
  \;\equiv\;
  \max_{t}\;\sum_{j \leq \ell} P_{D(G)-j}(t).
\end{equation*}
This quantity measures the probability of producing a feasible cover of size at most $\tau(G)+\ell$.
For the small fixed offsets $\ell<5$ shown in \cref{fig:cummulative_scaling}, the numerical data over the simulated range are well described by exponential decay in the graph size $n$.
For $\ell\geq5$, however, the fitted decay rates are statistically indistinguishable from zero over the same range.
At fixed $n$, the cumulative probability grows rapidly with $\ell$ and approaches unity for modest offsets.
These results support exponential suppression for the narrowest bands near the minimum-cover layer, while wider bottom bands appear substantially easier to reach in the investigated 3-regular instances.
An open question is how the minimum layer offset required to capture at least $1-\epsilon$ of the peak probability scales with $n$:
\begin{equation*}
  \ell^\star(n;\epsilon)
  \;=\;
  \min\big\{\,\ell \;:\; P^\star_{\leq \ell}(n) \geq 1-\epsilon\,\big\}.
\end{equation*}
Thus, $\ell^\star(n;\epsilon)$ is the smallest number of layers above the minimum-cover layer that must be included for the cumulative peak probability to reach the prescribed threshold.
Equivalently, it measures how wide a bottom band of the layered graph is required to retain probability at least $1-\epsilon$.
What survives the suppression of the deepest layers is not the single-shot
bottom-layer probability but the per-vertex statistics extracted from the walk
state along the way. This motivates the iterative greedy procedures of
\cref{sec:IterativeQuantum}: rather than attempting to measure a minimum cover
directly, we extract quantities from the walk state that remain informative
even when the bottom-layer probability is negligible.

\section{Iterative continuous-time quantum walk optimisation algorithms for MVC}
\label{sec:IterativeQuantum}

\begin{table}[t]

\centering
\caption{Score functions and vertex-selection rules for the primary algorithmic variants.}
\label{tab:score_selection}
\begin{tabular}{lccc}
\toprule
Variant & Score & LDF-like $v^*$ & SDF-like $v^*$ \\
\midrule
QPG & $P_v(t)$ & $\argmax_v P_v(t)$ & $\argmin_v P_v(t)$ \\
QEG & $E_v(t)$ & $\argmin_v E_v(t)$ & $\argmax_v E_v(t)$ \\
\bottomrule
\end{tabular}
\end{table}

The constrained walk developed above provides a way to extract vertex-level
information from the feasible vertex-cover subspace. We now use this information
to define iterative greedy algorithms for MVC. At each recursion step, the walk
is run on the current residual graph, vertex scores are computed from the
resulting state, one or more vertices are fixed according to a greedy rule, and
the graph is reduced.

The ideal object underlying these algorithms is the continuous-time walk state
\begin{equation*}
    \ket{\psi(t)}
    =
    U(t)\ket{\Omega},
    \qquad
    U(t)=e^{itH_{\mathrm{MVC}}}.
\end{equation*}
Since the restriction of $H_{\mathrm{MVC}}$ to the feasible subspace is the
adjacency matrix of the induced feasible-state graph,
\begin{equation*}
    H_{\mathrm{MVC}}\big|_{\mathcal{H}_{\mathrm{VC}}}
    =
    A(\inducedG),
\end{equation*}
the exact evolution is a CTQW on $\inducedG$.
This continuous time dynamics is the object analysed in \cref{sec:induced_graph_walk} and provides the conceptual basis for the vertex scores used below.

For circuit implementation, the exact unitary is replaced by a first-order
Suzuki--Trotter approximation,
\begin{equation*}
    U_p(t)
    =
    \left(
        \prod_{i\in V} e^{i(t/p)\hat{X}_i}
    \right)^p ,
\end{equation*}
up to the chosen ordering of the projected generators. The implemented state is
therefore
\begin{equation*}
    \ket{\psi_p(t)}=U_p(t)\ket{\Omega}.
\end{equation*}
In the numerical results, $p=\infty$ denotes the exact continuous-time
evolution, while finite $p$ denotes the corresponding trotterised circuit.
The multicontrol operators are ordered by ascending vertex degree, recomputed on the residual graph at each recursion step, with ties broken by vertex index.
On regular instances the ordering is therefore determined by the index at the first recursion level and becomes degree-dependent only after the graph has been reduced.
Implementation aspects of the trotterised state, including circuit depth,
Trotter error, and entanglement growth, are analysed in
\cref{sec:CTQW_state_analysis}.
In the special case $p=1$, the QPG and QEG scores can be estimated by sequential Monte Carlo sampling without explicitly constructing the full quantum state, as described in \cref{sec:montecarlo}.

\begin{figure}[t]
\refstepcounter{algbox}
\label{alg:quantum_greedy}
\noindent\textbf{Algorithm~\thealgbox: Generic quantum-informed greedy reduction}
\vspace{2pt}
\hrule
\vspace{4pt}
\begin{algorithmic}[1]
    \State \textbf{Input:} Graph $G$, score rule $\mathsf{s}$, reduction rule $R\in\{\mathrm{LDF},\mathrm{SDF}\}$, quantum parameters $\boldsymbol{\theta}$, and a tie-breaking rule
    \State Initialise the cover $C\gets\emptyset$
    \While{$E(G)\neq\emptyset$}
        \State Remove isolated vertices from $G$
        \State Evaluate the scores $\{s_v\}_{v\in V(G)}$ on the current residual graph using $\mathsf{s}$ and $\boldsymbol{\theta}$
        \State Select $v^*$ using the ordering associated with $\mathsf{s}$ and $R$, resolving ties with the specified rule
        \If{$R=\mathrm{LDF}$}
            \State $C\gets C\cup\{v^*\}$
            \State $G\gets G[V(G)\setminus\{v^*\}]$
        \Else
            \State $C\gets C\cup\mathcal{N}_G(v^*)$
            \State $G\gets G[V(G)\setminus(\{v^*\}\cup\mathcal{N}_G(v^*))]$
        \EndIf
    \EndWhile
    \State \Return $C$
\end{algorithmic}
\end{figure}

The main text considers two score functions: the marginal occupation probability used by QPG and the conditioned expected cover size used by QEG.
Both score functions are combined with the same recursive reduction procedure, while the direction of the score ordering depends on the chosen variant.
Both reduction rules are inspired by the classical greedy algorithm: the
largest-degree-first (LDF) rule adds the selected vertex itself to the cover,
whereas the smallest-degree-first (SDF) rule adds all of its neighbours.
\Cref{tab:score_selection} summarises the corresponding selection directions.
For QPG and QEG, the quantum parameters in \cref{alg:quantum_greedy} are $\boldsymbol{\theta}=(p,t)$.
All scores and neighbourhoods are evaluated on the current residual graph.

In the SDF-like update, vertex $v^*$ is removed from the residual graph without being added to the cover because all of its neighbours have been added and no incident edge remains unresolved.
For notational simplicity, the score definitions below are written for the ideal evolution, while finite $p$ denotes the corresponding Trotterized implementation.

\subsection{Quantum Probabilistic Greedy algorithm}
The \emph{Quantum Probabilistic Greedy} algorithm ranks vertices by their marginal probability of belonging to the vertex cover in the CTQW state.
For suitable walk times, these marginal probabilities can separate vertices that are consistently present in low-weight covers from vertices that are more often excluded.
The resulting ranking is then converted into a recursive graph reduction by either fixing a high-probability vertex into the cover or fixing the neighbours of a low-probability vertex, as described below.

The marginal probability that vertex $v$ belongs to the cover is the expectation value of the projector onto $\ket{1}_v$,
\begin{equation*}
    P_v(t)
    =
    \frac{1-\bra{\psi(t)}Z_v\ket{\psi(t)}}{2},
    \quad
    \ket{\psi(t)}
    =
    U(t)\ket{\Omega},
\end{equation*}
where $U(t)=e^{itH_{\mathrm{MVC}}}$ is the CTQW unitary at walk time $t$.
Using the Hadamard lemma, this probability can be rewritten in terms of nested commutators,
\begin{equation}
    P_v(t)=\frac{1}{2}-\frac{1}{2}\sum_{j=0}^{\infty}\frac{(-it)^{j}}{j!}\bra{\Omega}[\sum_k\hat{X}_k,Z_v]_j\ket{\Omega},
\label{eq:prob_campbell}
\end{equation}
where $[A,B]_j$ denotes the $j$-fold nested commutator,
using the convention
$
    [A,B]_0
    :=
    B,
    [A,B]_{j+1}
    :=
    [A,[A,B]_j]
$.

\paragraph{Nonlocal information structure.}

A crucial distinction between quantum and classical greedy algorithms emerges from the structure of the commutator expansion in \cref{eq:prob_campbell}.
Each additional commutator can enlarge the support of the observable by at most one graph-distance shell.
To make this precise, for a vertex $v\in V$ and integer $j\ge 0$, let
\begin{equation*}
    \Ball{j}{v}
    :=
    \{u\in V:\operatorname{dist}_G(u,v)\le j\}
\end{equation*}
denote the radius-$j$ ball around $v$ in $G$.
Thus $\Ball{0}{v}=\{v\}$ and $\Ball{1}{v}=\{v\}\cup\Nof{v}$.

The light-cone analysis in \cref{sec:StructureLieAlgebra} shows that the nested commutator $[H_{\mathrm{MVC}},Z_v]_j$ has support within $\Ball{j}{v}$.
It also implies the even-order selection rule used in \cref{eq:prob_campbell}: by symmetry of the projected operators $\widehat X_i$, all odd-order terms vanish in the expectation value.

This distinguishes the quantum score from purely local greedy rules.
While classical greedy algorithms (LDF, SDF) make decisions based solely on local degree information $\deg(v)$ or at most the open neighbourhood $\Nof{v}$, the quantum probability $P_v(t)$ incorporates contributions from progressively larger balls around $v$.
As time $t$ increases, and the continuous-time walk reaches deeper layers on the induced feasible-state graph, higher-order terms in the expansion become significant, and the effective information radius grows (\cref{cor:global}).
This is the mechanism by which QPG encodes nonlocal structural information into local vertex selection.

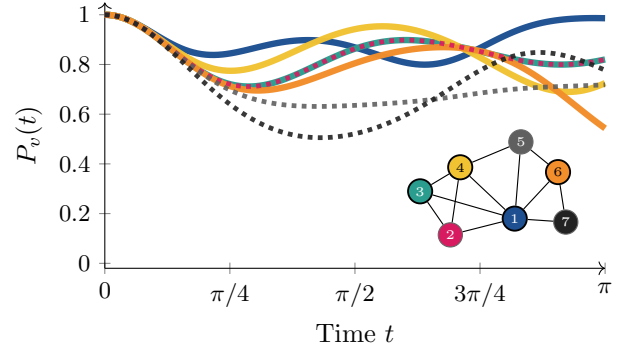
\begin{figure}
    \centering

\definecolor{c1}{HTML}{1D4F91}  
\definecolor{c2}{HTML}{D81B60}  
\definecolor{c3}{HTML}{2A9D8F}  
\definecolor{c4}{HTML}{F1C232}  
\definecolor{c5}{HTML}{5F5F5F}  
\definecolor{c6}{HTML}{F28E2B}  
\definecolor{c7}{HTML}{222222}  

\pgfplotsset{
  good curve/.style={
    no markers,
    solid,
    line width=2.4pt,
    opacity=0.98
  },
  bad curve/.style={
    no markers,
    solid,
    line width=1.8pt,
    opacity=0.90
  }
}

\begin{tikzpicture}
\begin{axis}[
  width=8.2cm,
  height=5.1cm,
  xlabel={Time $t$},
  ylabel={$P_v(t)$},
  xmin=0,
  xmax=3.14159265,
  ymin=-0.02,
  ymax=1.05,
  xtick={0,0.7853981,1.5707963,2.3561944,3.14159265},
  xticklabels={$0$,$\pi/4$,$\pi/2$,$3\pi/4$,$\pi$},
  ytick={0,0.2,0.4,0.6,0.8,1.0},
  axis lines=left,
  axis line style={->},
  grid=none,
  tick label style={font=\small},
  label style={font=\normalsize},
  clip=false
]

\addplot+[good curve, color=c1]
  table[x=beta,y=v0,col sep=comma]
  {final_reduced_plotting/evolution_times/p_vs_beta.csv};

\addplot+[good curve, color=c3]
  table[x=beta,y=v2,col sep=comma]
  {final_reduced_plotting/evolution_times/p_vs_beta.csv};

\addplot+[good curve, color=c4]
  table[x=beta,y=v3,col sep=comma]
  {final_reduced_plotting/evolution_times/p_vs_beta.csv};

\addplot+[good curve, color=c6]
  table[x=beta,y=v5,col sep=comma]
  {final_reduced_plotting/evolution_times/p_vs_beta.csv};
  
\addplot+[bad curve, color=c2, dotted]
  table[x=beta,y=v1,col sep=comma]
  {final_reduced_plotting/evolution_times/p_vs_beta.csv};

\addplot+[bad curve, color=c5, dotted]
  table[x=beta,y=v4,col sep=comma]
  {final_reduced_plotting/evolution_times/p_vs_beta.csv};

\addplot+[bad curve, color=c7, dotted]
  table[x=beta,y=v6,col sep=comma]
  {final_reduced_plotting/evolution_times/p_vs_beta.csv};

\node[
  anchor=south west,
  fill=white,
  fill opacity=0.84,
  text opacity=1,
  inner sep=0.6pt,
  rounded corners=1pt
]
at (rel axis cs:0.6,0.070)
{
\begin{tikzpicture}[
  every node/.style={
    circle,
    draw=black,
    line width=0.42pt,
    minimum size=3.2mm,
    inner sep=0pt,
    font=\tiny
  },
  vgood/.style={
    draw=black,
    line width=0.70pt
  },
  vbad/.style={
    draw=black!60,
    line width=0.42pt
  },
  edge/.style={
    draw=black,
    line width=0.42pt
  }
]

\begin{scope}[x=0.6cm,y=0.6cm,rotate=-45]

\node[vgood, fill=c1, text=white] (1) at ( 0.20, 0.00) {1};
\node[vbad,  fill=c2, text=white] (2) at (-0.55,-1.25) {2};
\node[vgood, fill=c3, text=white] (3) at (-1.70,-1.05) {3};
\node[vgood, fill=c4, text=black] (4) at (-1.45,-0.05) {4};
\node[vbad,  fill=c5, text=white] (5) at (-0.90, 1.30) {5};
\node[vgood, fill=c6, text=black] (6) at ( 0.15, 1.40) {6};
\node[vbad,  fill=c7, text=white] (7) at ( 1.05, 0.75) {7};

\draw[edge] (1) -- (2);
\draw[edge] (1) -- (3);
\draw[edge] (1) -- (4);
\draw[edge] (1) -- (5);
\draw[edge] (1) -- (6);
\draw[edge] (1) -- (7);
\draw[edge] (2) -- (3);
\draw[edge] (2) -- (4);
\draw[edge] (3) -- (4);
\draw[edge] (4) -- (5);
\draw[edge] (5) -- (6);
\draw[edge] (6) -- (7);

\end{scope}
\end{tikzpicture}
};

\end{axis}
\end{tikzpicture}
    \caption{Vertex occupation probabilities $P_v(t)$ as a function of the continuous quantum walk time $t$ for the depicted small example graph. As $t$ increases, higher-order terms in the commutator expansion become relevant and the quantum walk propagates in the layered graph, leading to a progressive separation of probability profiles and encoding increasingly nonlocal structural information, revealing how the quantum state induces a ranking over vertices, which is exploited by the greedy selection rules.}
    \label{fig:probvst}
\end{figure}

\begin{figure}[t]
    \centering
\begin{tikzpicture}

\definecolor{c4}{RGB}{0,114,178}
\definecolor{cE}{RGB}{190,30,45}

\begin{axis}[
  width=8.2cm,
  height=7.1cm,
  xlabel={Time $t$},
  ylabel={Empirical success rate},
  xmin=0,
  xmax=17,
  xtick={2,4,6,8,10,12,14,16},
  xticklabels={
    $\frac{\pi}{8}$,
    $\frac{\pi}{4}$,
    $\frac{3\pi}{8}$,
    $\frac{\pi}{2}$,
    $\frac{5\pi}{8}$,
    $\frac{3\pi}{4}$,
    $\frac{7\pi}{8}$,
    $\pi$
  },
  x tick label style={
    font=\normalsize,
    rotate=45,
    anchor=north east
  },
  ymin=0,
  ymax=1,
  ytick={0,0.2,0.4,0.6,0.8,1.0},
  grid=both,
  tick align=inside,
  tick label style={font=\normalsize},
  label style={font=\normalsize},
  title style={font=\normalsize},
  xlabel style={yshift=7pt},
  every axis plot/.append style={
    line width=1.1pt,
    mark size=2.4pt,
    error bars/y dir=both,
    error bars/y explicit=true,
    error bars/error bar style={line width=0.8pt},
    error bars/error mark options={
      rotate=90,
      mark size=2.4pt
    }
  },
  legend columns=2,
  legend style={
    at={(axis description cs:0.02,0.02)},
    anchor=south west,
    draw=black,
    rounded corners=1pt,
    fill=white,
    fill opacity=0.95,
    text opacity=1,
    font=\scriptsize,
    column sep=0.08cm,
    row sep=-1pt,
    inner xsep=2pt,
    inner ysep=1pt,
    cells={anchor=west}
  },
  legend image code/.code={
    \draw[
      mark repeat=2,
      mark phase=2,
      #1
    ]
    plot coordinates {
      (0cm,0cm)
      (0.25cm,0cm)
      (0.50cm,0cm)
    };
  }
]

\addplot+[
  color=c4,
  mark=*,
  mark options={fill=c4,draw=c4},
  solid,
  error bars/.cd,
  y dir=both,
  y explicit
]
table[
  x expr=\coordindex+1,
  y=prob_optimal,
  y error=sem,
  col sep=comma
]
{new_plots/t_analysis/prob/QPG_LDF_3_regular_t_p4.csv};
\addlegendentry{LDF $p=4$}

\addplot+[
  color=cE,
  mark=*,
  mark options={fill=cE,draw=cE},
  solid,
  error bars/.cd,
  y dir=both,
  y explicit
]
table[
  x expr=\coordindex+1,
  y=prob_optimal,
  y error=sem,
  col sep=comma
]
{new_plots/prob/3-regular/popt_vs_beta_qpg_ldf_n18.csv};
\addlegendentry{LDF $p=\infty$}

\addplot+[
  color=c4,
  mark=*,
  mark options={fill=white,draw=c4,solid},
  dashed,
  error bars/.cd,
  y dir=both,
  y explicit
]
table[
  x expr=\coordindex+1,
  y=prob_optimal,
  y error=sem,
  col sep=comma
]
{new_plots/t_analysis/prob/QPG_SDF_3_regular_t_p4.csv};
\addlegendentry{SDF $p=4$}

\addplot+[
  color=cE,
  mark=*,
  mark options={fill=white,draw=cE,solid},
  dashed,
  error bars/.cd,
  y dir=both,
  y explicit
]
table[
  x expr=\coordindex+1,
  y=prob_optimal,
  y error=sem,
  col sep=comma
]
{new_plots/prob/3-regular/popt_vs_beta_qpg_sdf_n18.csv};
\addlegendentry{SDF $p=\infty$}

\addplot[
  color=gray,
  solid,
  line width=1.2pt,
  no marks,
  forget plot
]
coordinates {(10,0)(10,1)};

\node[
  gray,
  font=\scriptsize,
  fill=white,
  fill opacity=0.9,
  text opacity=1,
  inner sep=1pt,
  rotate=90,
  anchor=south
]
at (axis cs:10.18,0.48) {$t^\star$ for LDF $p=4$};

\end{axis}
\end{tikzpicture}
    \caption{
        Walk-time selection for QPG on 100 random $3$-regular graph instances with $n=18$.
        The empirical success rate is the fraction of instances for which \cref{alg:quantum_greedy} returns a minimum vertex cover after completing the full iterative procedure.
        Blue curves show Trotterised evolution with $p=4$, while red curves show exact continuous-time evolution corresponding to $p=\infty$.
        Solid lines with filled markers denote LDF, while dashed lines with open markers denote SDF.
        The vertical gray line marks the empirically selected walk time used in the subsequent benchmarks.
        Error bars show the SEM over the 100 instances.
    }
    \label{fig:qpg-3regular-performance}
\end{figure}
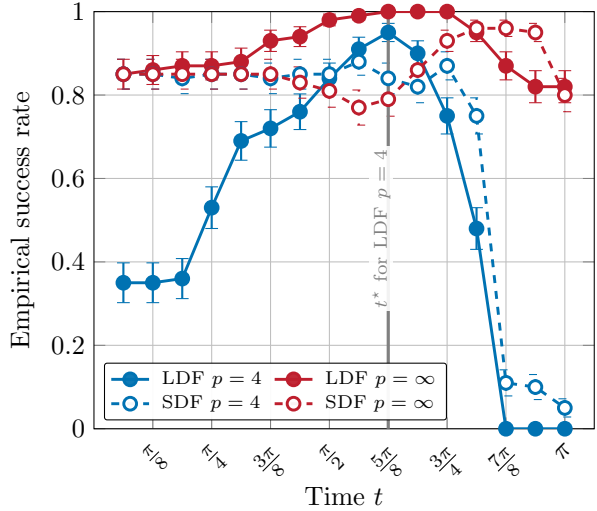
\paragraph{Largest-probability (LDF-like) strategy.}
In the first approach, the algorithm selects the vertex with the largest occupation probability,
\begin{equation*}
v^* = \argmax_{v\in V} P_v(t),
\end{equation*}
and fixes it in the vertex cover. 
The graph is then reduced by removing $v^*$ and all edges incident to it.

This strategy mirrors the classical largest-degree greedy heuristic: vertices that are more likely to be 
in the cover (i.e. with large $P_v(t)$) tend to play a central role in covering many edges. 
From a quantum perspective, this corresponds to reinforcing configurations that already have 
significant weight in the quantum state.

As shown in \cref{fig:probvst}, the ranking induced by the quantum
probabilities is not necessarily aligned with the vertex-degree ordering. In
particular, for evolution times $t \in [3\pi/4,\pi]$, the QPG algorithm may
favour vertices that are not selected by the Largest Degree First (LDF)
heuristic. In the example considered, QPG-LDF would select the vertex with the
second-highest degree, whereas the quantum probabilities assign a different
ordering. This illustrates how QPG departs from purely local degree-based
criteria and incorporates nonlocal structural information arising from the
quantum evolution.

\paragraph{Smallest-probability (SDF-like) strategy.}
An alternative approach is to select the vertex with the smallest occupation probability,
\begin{equation*}
v^* = \arg\min_{v\in V} P_v(t),
\end{equation*}
and instead fix all of its neighbours in the vertex cover.
After adding $\Nof{v^*}$ to the cover, the residual graph is obtained by removing $v^*$ together with all vertices in $\Nof{v^*}$.

This rule is the quantum analogue of the classical minimum-degree heuristic. A small value of 
$P_v(t)$ indicates that the quantum state assigns low weight to configurations where $v$ is in 
the cover. Consequently, to satisfy the edge constraints, its neighbours must be included instead. 
This leads to a more aggressive update step, where multiple vertices may be fixed simultaneously.

The two strategies reflect different ways of extracting information from the quantum state:
\begin{itemize}
    \item The LDF-like rule exploits \emph{high-probability structure}, selecting vertices that are 
    strongly supported by the quantum distribution.
    
    \item The SDF-like rule exploits \emph{low-probability structure}, interpreting unlikely vertices 
    as indicators that their neighbours should be chosen instead.
\end{itemize}
For each algorithm we must choose an initial walk time $t$, which sets the distribution over feasible covers induced by the CTQW. \Cref{fig:qpg-3regular-performance} reports the performance of the two strategies for $t \in [0,\pi]$ on $100$ $3$-regular instances with $n=18$ vertices, comparing the exact continuous-time walk ($p=\infty$) with its finite-depth Trotter approximation
at $p=4$.
Corresponding scans for the $4$-regular and connected Erd\H{o}s--R\'enyi ($\rho=0.3$) ensembles exhibit the same qualitative dependence on $t$ and are omitted for brevity.

The first feature to note is that, for QPG, the gap between $p=4$ and
$p=\infty$ is substantial: unlike the other criteria discussed later, a depth
of $p=4$ does not yet reproduce the exact walk. This is a direct manifestation
of the Trotter error analysed in \cref{sec:Trotter}, which scales as
$\mathcal{O}(|E| t^2/p)$ and grows with the walk time.
Here an always-in vertex means a vertex contained in every minimum vertex cover, while a never-in vertex means a vertex contained in no minimum vertex cover.
Because the QPG score relies on the fine contrast between per-vertex marginal occupations, in particular the gap between always-in and never-in vertices, it is especially sensitive to this error.

In the exact ($p=\infty$) case the two strategies behave quite differently in their dependence on $t$.
The LDF strategy peaks sharply around $t\approx5\pi/8$, where it attains an empirical success rate of approximately $100\%$.
The SDF strategy reaches its best performance at a later time,
$t \approx 7\pi/8$, and with a slightly lower peak value than LDF. 
The later SDF optimum is consistent with its more aggressive update
rule, which commits the entire neighbourhood of the selected vertex at each step and is therefore less tied to the precise location of the marginal-contrast
maximum.

At finite depth $p=4$ the picture changes completely. The LDF curve
retains a sensitivity to $t$: it still exhibits a clear peak whose
height is comparable to the exact case, so that for a suitably chosen walk time
the trotterised LDF strategy remains close to optimal despite the 
Trotter error. The SDF curve, by contrast, is essentially flat across a broad
range of $t$, reflecting the insensitivity of its neighbourhood update rule to the
detailed shape of the quantum distribution. This flatness persists until a time
$t \approx 3\pi/4$, beyond which both the $p=4$ LDF and SDF
curves drop in performance, as shown in \cref{fig:qpg-3regular-performance}.

The location of the LDF and SDF peaks, and the mechanism that produces it, can be understood through the layered-walk picture.
The CTQW propagates a wavefront outward from $\ket{\Omega}$ through layers of decreasing Hamming weight, and the amplitude on the minimum-cover boundary $\mathcal{L}_{\tau(G)}$ peaks before $t=\pi/2$.
One might expect this
to be the optimal time for the greedy as well, but empirically the LDF peak
$t^{\star}_{\text{LDF}} \approx 5\pi/8$ and SDF peak $t^{\star}_{\text{SDF}} \approx 7\pi/8$ occur \emph{after} the wavefront has
departed the bottom layer.
The reason is that the QPG strategy exploits not the absolute occupation of the bottom layers but the \emph{contrast} between per-vertex marginal occupations: the probability gap between always-in and never-in vertices.

This contrast is built up by a combination of coherent reflection and selective
trapping. As the wavefront returns from $\mathcal{L}_{\tau(G)}$ and spreads
back across the upper layers $\mathcal{L}_{\tau}, \ldots, \mathcal{L}_{n}$, the
irregularity of $\inducedG$ causes destructive interference to
suppress amplitude on certain paths while constructive interference reinforces
others. In particular, a fraction of the amplitude remains trapped near the
bottom layers on configurations that contain the always-in vertices and exclude
the never-in ones. It is this
interference rather than direct occupation of the optimal
layer, that converts the global structure of $\inducedG$ into a
per-vertex ranking accessible through local measurements.

The two variants thus offer
complementary trade-offs: LDF makes conservative, fine-grained updates, while
SDF enforces constraints more aggressively by fixing multiple vertices at once,
affecting both the number of recursive steps and the structure of the
intermediate quantum states.

\subsection{Quantum Energy Greedy algorithm}

While the previous strategy relies on marginal vertex probabilities, a more global selection criterion can be obtained by evaluating the expected cost of the walk state.
The Quantum Energy Greedy (QEG) algorithm scores each vertex by the expected cover size obtained when that vertex is fixed in the cover.

For the unweighted MVC instances considered in the main benchmarks, the cost Hamiltonian is
\begin{equation*}
    H_C
    =
    \sum_{\ell\in V}\frac{I-Z_\ell}{2}
    =
    \frac{|V|}{2}
    -
    \frac{1}{2}\sum_{\ell\in V} Z_\ell .
\end{equation*}
On a feasible computational basis state $\ket{C}$, this gives $H_C\ket{C}=|C|\ket{C}$.

To determine which vertex should be fixed next, we evaluate the expected cost obtained by fixing each vertex $v$ in the cover.
This is implemented by omitting the projected mixer term $\widehat X_v$, so that qubit $v$ remains in the state $\ket{1}$ throughout the evolution.
We therefore define the conditioned mixer Hamiltonian
\begin{equation*}
    H_{\mathrm{MVC}}^{(v)}
    :=
    \sum_{u\in V\setminus\{v\}}\widehat X_u
\end{equation*}
and the corresponding conditioned walk state
\begin{equation*}
    \ket{\psi^v(t)}
    :=
    e^{itH_{\mathrm{MVC}}^{(v)}}\ket{\Omega}.
\end{equation*}
The QEG score is then
\begin{equation*}
    E_v(t)
    :=
    \bra{\psi^v(t)}H_C\ket{\psi^v(t)} .
\end{equation*}
Small values of $E_v(t)$ indicate that fixing $v$ in the cover leads to a low expected cover size after the conditioned walk.

For finite Trotter depth $p$, the conditioned walk state in the algorithms below is prepared using the same first-order product formula as in the unconditioned case.
The notation $E_v(t)$ is kept for the resulting score, with $p=\infty$ denoting the exact conditioned evolution.

Using the Hadamard lemma, this score can be rewritten as
\begin{equation}
\begin{aligned}
    E_v(t)
    &=
    \frac{|V|}{2}
    -
    \frac{1}{2}
    \sum_{\ell\in V}
    \sum_{j=0}^{\infty}
    \frac{(-it)^j}{j!} \\
    &\qquad\qquad\times
    \bra{\Omega}
    \left[
        H_{\mathrm{MVC}}^{(v)},
        Z_\ell
    \right]_j
    \ket{\Omega},
\end{aligned}
\label{eq:energy_campbell}
\end{equation}
where $[A,B]_j$ denotes the $j$-fold nested commutator.

\paragraph{Graph-wide information structure.}
Both QPG and QEG exploit the nonlocal information generated by the constrained walk, but they aggregate this information differently.
The QPG score $P_v(t)$ is a single-vertex marginal evaluated in the same unconditioned walk state for every candidate vertex.
By contrast, QEG prepares a separate state $\ket{\psi^v(t)}$ conditioned on fixing each candidate vertex $v$ and evaluates the full cost Hamiltonian in that state.
For the unweighted problem, the resulting score can be written as $E_v(t)=\sum_{\ell\in V}P_\ell^{(v)}(t)$, where $P_\ell^{(v)}(t)$ is the marginal occupation of vertex $\ell$ in the conditioned state $\ket{\psi^v(t)}$.
QEG therefore compares the graph-wide occupation response to each candidate fixation, rather than ranking a candidate only by its own marginal occupation.
As detailed in \cref{sec:StructureLieAlgebra}, the nested-commutator expansions of both scores inherit the same light-cone behaviour and even-order selection rule.

Because QEG prepares one conditioned circuit for each candidate vertex, it requires $\mathcal{O}(|V_r|)$ circuit evaluations at recursion step $r$ and at most $\mathcal{O}(|V|^2)$ over a complete run, excluding the number of measurement shots.

\paragraph{Minimum energy (LDF-like) strategy.}

In the first approach, the algorithm selects the vertex whose fixation produces 
the largest decrease in the expected cost, namely
\begin{equation*}
v^* = \arg\min_{v\in V} E_v(t),
\end{equation*}
and fixes it in the vertex cover.
The graph is then reduced by removing $v^*$  and all edges incident to it.
\begin{figure}[t]
    \centering

\definecolor{c1}{HTML}{1D4F91}  
\definecolor{c2}{HTML}{D81B60}  
\definecolor{c3}{HTML}{2A9D8F}  
\definecolor{c4}{HTML}{F1C232}  
\definecolor{c5}{HTML}{5F5F5F}  
\definecolor{c6}{HTML}{F28E2B}  
\definecolor{c7}{HTML}{222222}  

\pgfplotsset{
  good curve/.style={
    no markers,
    solid,
    line width=2.4pt,
    opacity=0.98
  },
  bad curve/.style={
    no markers,
    solid,
    line width=1.8pt,
    opacity=0.90
  }
}

\begin{tikzpicture}
\begin{axis}[
  width=8.2cm,
  height=5.1cm,
  xlabel={Time $t$},
  ylabel={$E_v(t)$},
  xmin=0,
  xmax=3.14159265,
  xtick={0,0.7853981,1.5707963,2.3561944,3.14159265},
  xticklabels={$0$,$\pi/4$,$\pi/2$,$3\pi/4$,$\pi$},
  axis lines=left,
  axis line style={->},
  grid=none,
  tick label style={font=\small},
  label style={font=\normalsize},
  clip=false
]

\addplot+[good curve, color=c1]
  table[x=beta,y=v0,col sep=comma]
  {final_reduced_plotting/evolution_times/E_vs_beta.csv};

\addplot+[good curve, color=c3]
  table[x=beta,y=v2,col sep=comma]
  {final_reduced_plotting/evolution_times/E_vs_beta.csv};

\addplot+[good curve, color=c4]
  table[x=beta,y=v3,col sep=comma]
  {final_reduced_plotting/evolution_times/E_vs_beta.csv};

\addplot+[good curve, color=c6]
  table[x=beta,y=v5,col sep=comma]
  {final_reduced_plotting/evolution_times/E_vs_beta.csv};
  
\addplot+[bad curve, color=c2, dotted]
  table[x=beta,y=v1,col sep=comma]
  {final_reduced_plotting/evolution_times/E_vs_beta.csv};

\addplot+[bad curve, color=c5, dotted]
  table[x=beta,y=v4,col sep=comma]
  {final_reduced_plotting/evolution_times/E_vs_beta.csv};

\addplot+[bad curve, color=c7, dotted]
  table[x=beta,y=v6,col sep=comma]
  {final_reduced_plotting/evolution_times/E_vs_beta.csv};

\node[
  anchor=south west,
  fill=none,
  fill opacity=0.84,
  text opacity=1,
  inner sep=0.6pt,
  rounded corners=1pt
]
at (rel axis cs:0.15,0.60)
{
\begin{tikzpicture}[
  every node/.style={
    circle,
    draw=black,
    line width=0.42pt,
    minimum size=3.2mm,
    inner sep=0pt,
    font=\tiny
  },
  vgood/.style={
    draw=black,
    line width=0.70pt
  },
  vbad/.style={
    draw=black!60,
    line width=0.42pt
  },
  edge/.style={
    draw=black,
    line width=0.42pt
  }
]

\begin{scope}[x=0.6cm,y=0.6cm,rotate=-45]

\node[vgood, fill=c1, text=white] (1) at ( 0.20, 0.00) {1};
\node[vbad,  fill=c2, text=white] (2) at (-0.55,-1.25) {2};
\node[vgood, fill=c3, text=white] (3) at (-1.70,-1.05) {3};
\node[vgood, fill=c4, text=black] (4) at (-1.45,-0.05) {4};
\node[vbad,  fill=c5, text=white] (5) at (-0.90, 1.30) {5};
\node[vgood, fill=c6, text=black] (6) at ( 0.15, 1.40) {6};
\node[vbad,  fill=c7, text=white] (7) at ( 1.05, 0.75) {7};

\draw[edge] (1) -- (2);
\draw[edge] (1) -- (3);
\draw[edge] (1) -- (4);
\draw[edge] (1) -- (5);
\draw[edge] (1) -- (6);
\draw[edge] (1) -- (7);
\draw[edge] (2) -- (3);
\draw[edge] (2) -- (4);
\draw[edge] (3) -- (4);
\draw[edge] (4) -- (5);
\draw[edge] (5) -- (6);
\draw[edge] (6) -- (7);

\end{scope}
\end{tikzpicture}
};

\end{axis}
\end{tikzpicture}
    \caption{
        Vertex Energy $E_v(t)$ as a function of the continuous quantum walk time $t$ for the depicted small example graph.
        Similarly as for the probability case, as $t$ increases, higher-order terms in the commutator expansion become relevant, the quantum walk spread in the layered graph, leading to a progressive separation of energies profiles encoding increasingly nonlocal structural information.
    }
    \label{fig:Vertex_energy}
\end{figure}
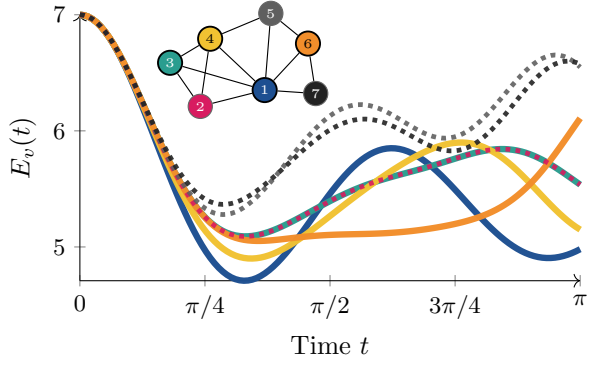

This strategy can be interpreted as a quantum analogue of the classical largest-degree greedy heuristic. At each step, fixing each vertex to be in the cover results in a new layered graph with depth dependent on the specific vertex choice. We can now evolve the CTQW quantum state on the conditioned graph up for a certain time t and measure how much the expectation value of the cost Hamiltonian decreases when $v$ is included in the cover set. The vertex with the largest energy reduction is selected and added to the cover.
Similarly as for QPG, we observe in \cref{fig:Vertex_energy}, that the ranking induced by the quantum conditional energies is not necessarily aligned with the vertex degree ordering.
From a quantum perspective, this rule compares how fixing different vertices changes the graph-wide occupation profile generated by the constrained walk.
A small value of $E_v(t)$ indicates that the walk conditioned on fixing $v$ produces a low expected cover size, so selecting the minimiser favours the fixation with the lowest aggregate cost.

\paragraph{Maximum energy (SDF-like) strategy.}
An alternative approach is to select the vertex whose fixation produces the largest increase (or smallest decrease) in the expected energy,
\begin{equation*}
    v^*
    =
    \argmax_{v\in V} E_v(t),
\end{equation*}
and instead fix all of its neighbours in the vertex cover.
After adding $\Nof{v^*}$ to the cover, the residual graph is obtained by removing $v^*$ together with all vertices in $\Nof{v^*}$.

\begin{figure}[t]
    \centering
\begin{tikzpicture}

\definecolor{c4}{RGB}{0,114,178}
\definecolor{cE}{RGB}{190,30,45}

\begin{axis}[
  width=8.2cm,
  height=7.1cm,
  xlabel={Time $t$},
  ylabel={Empirical success rate},
  xmin=0,
  xmax=17,
  xtick={2,4,6,8,10,12,14,16},
  xticklabels={
    $\frac{\pi}{8}$,
    $\frac{\pi}{4}$,
    $\frac{3\pi}{8}$,
    $\frac{\pi}{2}$,
    $\frac{5\pi}{8}$,
    $\frac{3\pi}{4}$,
    $\frac{7\pi}{8}$,
    $\pi$
  },
  x tick label style={
    font=\normalsize,
    rotate=45,
    anchor=north east
  },
  ymin=0,
  ymax=1,
  ytick={0,0.2,0.4,0.6,0.8,1.0},
  grid=both,
  tick align=inside,
  tick label style={font=\normalsize},
  label style={font=\normalsize},
  title style={font=\normalsize},
  xlabel style={yshift=7pt},
  every axis plot/.append style={
    line width=1.1pt,
    mark size=2.4pt,
    error bars/y dir=both,
    error bars/y explicit=true,
    error bars/error bar style={line width=0.8pt},
    error bars/error mark options={
      rotate=90,
      mark size=2.4pt
    }
  },
  legend columns=2,
  legend style={
    at={(axis description cs:0.02,0.02)},
    anchor=south west,
    draw=black,
    rounded corners=1pt,
    fill=white,
    fill opacity=0.95,
    text opacity=1,
    font=\scriptsize,
    column sep=0.08cm,
    row sep=-1pt,
    inner xsep=2pt,
    inner ysep=1pt,
    cells={anchor=west}
  },
  legend image code/.code={
    \draw[
      mark repeat=2,
      mark phase=2,
      #1
    ]
    plot coordinates {
      (0cm,0cm)
      (0.25cm,0cm)
      (0.50cm,0cm)
    };
  }
]

\addplot+[
  color=c4,
  mark=*,
  mark options={fill=c4,draw=c4},
  solid,
  error bars/.cd,
  y dir=both,
  y explicit
]
table[
  x expr=\coordindex+1,
  y=prob_optimal,
  y error=sem,
  col sep=comma
]
{new_plots/t_analysis/prob/QEG_LDF_3_regular_t_p4.csv};
\addlegendentry{LDF $p=4$}

\addplot+[
  color=cE,
  mark=*,
  mark options={fill=cE,draw=cE},
  solid,
  error bars/.cd,
  y dir=both,
  y explicit
]
table[
  x expr=\coordindex+1,
  y=prob_optimal,
  y error=sem,
  col sep=comma
]
{new_plots/prob/3-regular/reduced_t_analysis_qeg_ldf_exact.csv};
\addlegendentry{LDF $p=\infty$}

\addplot+[
  color=c4,
  mark=*,
  mark options={fill=white,draw=c4,solid},
  dashed,
  error bars/.cd,
  y dir=both,
  y explicit
]
table[
  x expr=\coordindex+1,
  y=prob_optimal,
  y error=sem,
  col sep=comma
]
{new_plots/prob/3-regular/reduced_t_analysis_qeg_sdf_p4.csv};
\addlegendentry{SDF $p=4$}

\addplot+[
  color=cE,
  mark=*,
  mark options={fill=white,draw=cE,solid},
  dashed,
  error bars/.cd,
  y dir=both,
  y explicit
]
table[
  x expr=\coordindex+1,
  y=prob_optimal,
  y error=sem,
  col sep=comma
]
{new_plots/prob/3-regular/reduced_t_analysis_qeg_sdf_exact.csv};
\addlegendentry{SDF $p=\infty$}

\addplot[
  color=gray,
  solid,
  line width=1.2pt,
  no marks,
  forget plot
]
coordinates {(10,0)(10,1)};

\node[
  gray,
  font=\scriptsize,
  fill=white,
  fill opacity=0.9,
  text opacity=1,
  inner sep=1pt,
  rotate=90,
  anchor=south
]
at (axis cs:10.18,0.48) {$t^\star$ for LDF $p=4$};

\end{axis}
\end{tikzpicture}
    \caption{
        Walk-time selection for QEG on 100 random $3$-regular graph instances with $n=18$.
        The empirical success rate is the fraction of instances for which \cref{alg:quantum_greedy} returns a minimum vertex cover after completing the full iterative procedure.
        Blue curves show Trotterised evolution with $p=4$, while red curves show exact continuous-time evolution corresponding to $p=\infty$.
        Solid lines with filled markers denote LDF, while dashed lines with open markers denote SDF.
        The vertical gray line marks the selected walk time $t^\star$ used in the subsequent benchmarks.
        Error bars show the SEM error over the 100 instances.
    }
    \label{fig:qeg-3regular-performance}
\end{figure}

The intuition behind this rule is complementary to the previous one.
If fixing $v$ in the cover leads to a high expected cost, then including $v$ is comparatively unfavourable according to the conditioned quantum score.
The SDF-like rule therefore excludes $v$ from the cover and instead includes all of its neighbours, thereby satisfying the incident edge constraints.

As with QPG, the choice of walk time $t$ sets the distribution over feasible covers used by the greedy criterion.
\Cref{fig:qeg-3regular-performance} reports the performance for $t \in [0,\pi]$ on $100$ $3$-regular instances with $n=18$ vertices, comparing the exact continuous-time walk ($p=\infty$) with the finite-depth
approximation at $p=4$.
Corresponding scans for the $4$-regular and connected Erd\H{o}s--R\'enyi ($\rho=0.3$) ensembles exhibit the same qualitative dependence on $t$ and are omitted for brevity.

The dependence on $t$ mirrors the pattern established for QPG: the LDF
strategy peaks sharply around $t \approx 5\pi/8$, reflecting the same
interference mechanism described in the previous section, while SDF
stays broadly flat over a wide range before both curves decline past
$t \gtrsim 3\pi/4$ as the marginal orderings invert. 

The critical empirical difference from QPG emerges at finite Trotter depth.
Whereas QPG at $p=4$ differs substantially from its exact continuous-time limit, QEG at the same depth remains close to the corresponding exact results on the tested instances.
The QEG score evaluates the full conditioned cost Hamiltonian and therefore aggregates occupation contributions from all vertices.
On the tested instances, this graph-wide criterion produces vertex rankings that are empirically more stable under low-depth Trotterisation; however, the extensive summation alone does not imply cancellation of Trotter errors and therefore provides no analytic robustness guarantee.

This empirical stability makes QEG attractive for hardware settings in which the Trotter depth must remain small, although it comes with greater measurement overhead: the conditioned energies require $\mathcal{O}(|V|)$ circuit evaluations per greedy step rather than the $\mathcal{O}(1)$ required by QPG.

\begin{figure*}
    \centering
    \begin{tikzpicture}
\definecolor{G}{RGB}{0,0,0}
\definecolor{cqpg}{RGB}{0,114,178}
\definecolor{cqeg}{RGB}{213,94,0}
\definecolor{cqfg}{RGB}{0,158,115}

\begin{groupplot}[
  group style={
    group size=2 by 1,
    horizontal sep=2.2cm,
    vertical sep=0.5cm,
    x descriptions at=edge bottom,
  },
  width=7cm,
  height=5cm,
  xlabel={$p$},
  xmin=0.5, xmax=5.5,
  xtick={1,2,3,4,5},
  xticklabels={$1$,$2$,$3$,$4$,$\infty$},
  grid=both,
  tick align=inside,
  tick label style={font=\normalsize},
  label style={font=\normalsize},
  title style={font=\normalsize},
  every axis plot/.append style={
    line width=1.0pt,
    mark size=3.5pt,
    error bars/y dir=both,
    error bars/y explicit=true,
    error bars/error bar style={line width=0.8pt},
    error bars/error mark options={rotate=90,mark size=3.5pt}
  },
]

\nextgroupplot[
  title={Approximation ratio},
  ymin=1,
  ymax=1.07,
  ytick={1.00,1.02,1.04,1.06},
  legend to name=paxislegend,
  legend columns=4,
  legend style={
    draw=black,
    rounded corners=2pt,
    fill=white,
    font=\scriptsize,
    column sep=0.7cm,
    row sep=0.1cm,
    inner xsep=4pt,
    inner ysep=3pt,
    cells={anchor=west}
  },
  legend image post style={scale=0.5},
]

\addplot+[color=cqpg,mark=*,mark options={fill=cqpg,draw=cqpg},solid,error bars/.cd,y dir=both,y explicit]
  table[x=p,y=mean,y error expr={\thisrow{std}/10},col sep=comma]
  {new_plots/plots/p_axis_data/p_axis_ratio_qpg_ldf.csv};
\addlegendentry{QPG LDF}

\addplot+[color=cqeg,mark=square*,mark options={fill=cqeg,draw=cqeg},solid,error bars/.cd,y dir=both,y explicit]
  table[x=p,y=mean,y error expr={\thisrow{std}/10},col sep=comma]
  {new_plots/plots/p_axis_data/p_axis_ratio_qeg_ldf.csv};
\addlegendentry{QEG LDF}


\addplot+[color=cqpg,mark=*,mark options={fill=white,draw=cqpg,solid},dashed,error bars/.cd,y dir=both,y explicit]
  table[x=p,y=mean,y error expr={\thisrow{std}/10},col sep=comma]
  {new_plots/plots/p_axis_data/p_axis_ratio_qpg_sdf.csv};
\addlegendentry{QPG SDF}

\addplot+[color=cqeg,mark=square*,mark options={fill=white,draw=cqeg,solid},dashed,error bars/.cd,y dir=both,y explicit]
  table[x=p,y=mean,y error expr={\thisrow{std}/10},col sep=comma]
  {new_plots/plots/p_axis_data/p_axis_ratio_qeg_sdf.csv};
\addlegendentry{QEG SDF}


\nextgroupplot[
  title={Empirical success rate},
  ymin=0,
  ymax=1,
  ytick={0,0.2,0.4,0.6,0.8,1.0},
]

\addplot+[color=cqpg,mark=*,mark options={fill=cqpg,draw=cqpg},solid,error bars/.cd,y dir=both,y explicit]
  table[x=p,y=prob_optimal,y error=sem,col sep=comma]
  {new_plots/plots/p_axis_data/p_axis_prob_qpg_ldf.csv};

\addplot+[color=cqeg,mark=square*,mark options={fill=cqeg,draw=cqeg},solid,error bars/.cd,y dir=both,y explicit]
  table[x=p,y=prob_optimal,y error=sem,col sep=comma]
  {new_plots/plots/p_axis_data/p_axis_prob_qeg_ldf.csv};


\addplot+[color=cqpg,mark=*,mark options={fill=white,draw=cqpg,solid},dashed,error bars/.cd,y dir=both,y explicit]
  table[x=p,y=prob_optimal,y error=sem,col sep=comma]
  {new_plots/plots/p_axis_data/p_axis_prob_qpg_sdf.csv};

\addplot+[color=cqeg,mark=square*,mark options={fill=white,draw=cqeg,solid},dashed,error bars/.cd,y dir=both,y explicit]
  table[x=p,y=prob_optimal,y error=sem,col sep=comma]
  {new_plots/plots/p_axis_data/p_axis_prob_qeg_sdf.csv};


\end{groupplot}

\path (group c1r1.south west) -- (group c2r1.south east) coordinate[midway] (legendmid);
\node[anchor=north, yshift=-0.9cm] at (legendmid) {\pgfplotslegendfromname{paxislegend}};

\end{tikzpicture}
    \caption{QPG and QEG performances as a function of the number of Trotter layers $p$ on 100 3-regular graph instances of size $n=20$. Left: mean approximation ratio.
    Right: empirical success rate, defined as the fraction of instances for which the algorithm returns an optimal cover.
    $p=\infty$ shows the exact continuous-time evolution. Solid lines with filled markers denote LDF, while dashed lines with open markers denote SDF. Error bars show the SEM over the 100 graphs.
    }
    \label{fig:TrotterEffect}
\end{figure*}
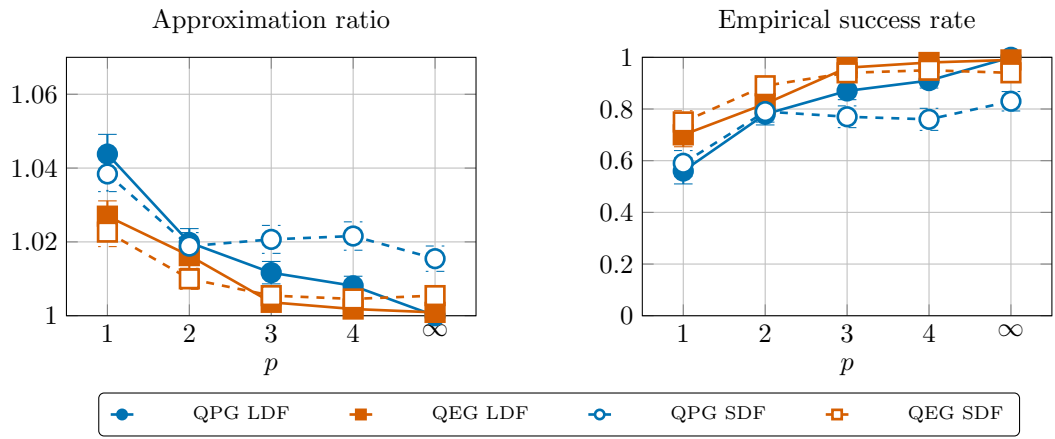

\section{Comparison and benchmarking}\label{sec:Results}
We evaluate our algorithms on a benchmark of 100 graphs per size across three
graph families: 3-regular, 4-regular, and Erd\H{o}s--R\'enyi connected random graphs with
connection probability $\rho = 0.3$, generated at sizes
$n \in \{8, 10, 12, 14, 16, 18, 20\}$.
The two main quantum algorithms, QPG and
QEG are each run at circuit depths
$p \in \{1, 2, 3, 4\}$ and in the exact ($p = \infty$) limit, and for every
method we consider both the largest-degree-first (LDF) and smallest-degree-first
(SDF) greedy strategies.

We compare against three classical baselines spanning constructive heuristics, a worst-case-guaranteed approximation, and a state-of-the-art local-search solver. The first two are the classical largest-degree-first
(\emph{Greedy LDF}) and smallest-degree-first (\emph{Greedy SDF}) heuristics that
motivate our quantum selection rules. The third baseline is the 
factor-2 approximation based on a maximal matching, which adds both endpoints of
every matched edge to the cover and is guaranteed to return a cover of size at most twice the optimum \cite{vazirani2001approximation}; it is the only baseline carrying a worst-case guarantee. The fourth is \emph{FastVC}~\cite{cai2015fastvc}, a well-known local-search heuristic for minimum vertex cover on large and massive graphs.
Performance is assessed using the mean approximation ratio and the empirical success rate, defined as the fraction of benchmark instances for which the algorithm returns an optimal vertex cover.
Throughout, the per-vertex scores are evaluated as exact expectation values.
On hardware, these scores must be estimated from a finite number of measurement shots, and the cost of resolving the gap between the top-ranked vertex and its competitor contributes to the true per-step runtime.
The results for all three graph families are reported in \cref{fig:comparison_p4}.
Additionally, our results on 3-regular graphs of sizes 8-20 are comparable to those reported in \cite{luiz2025scalable}; in particular, our QEG-LDF variant achieves favourable mean approximation ratios while preserving feasibility by construction at every step. The relative weakness of FastVC on these instances is consistent with the observations of \cite{luiz2025scalable}, who likewise report that FastVC does not retain its advantage on small regular graphs.

\subsection{Trotter effect}
\label{subsec:trotter_effect}

\begin{figure*}
    \centering
    \input{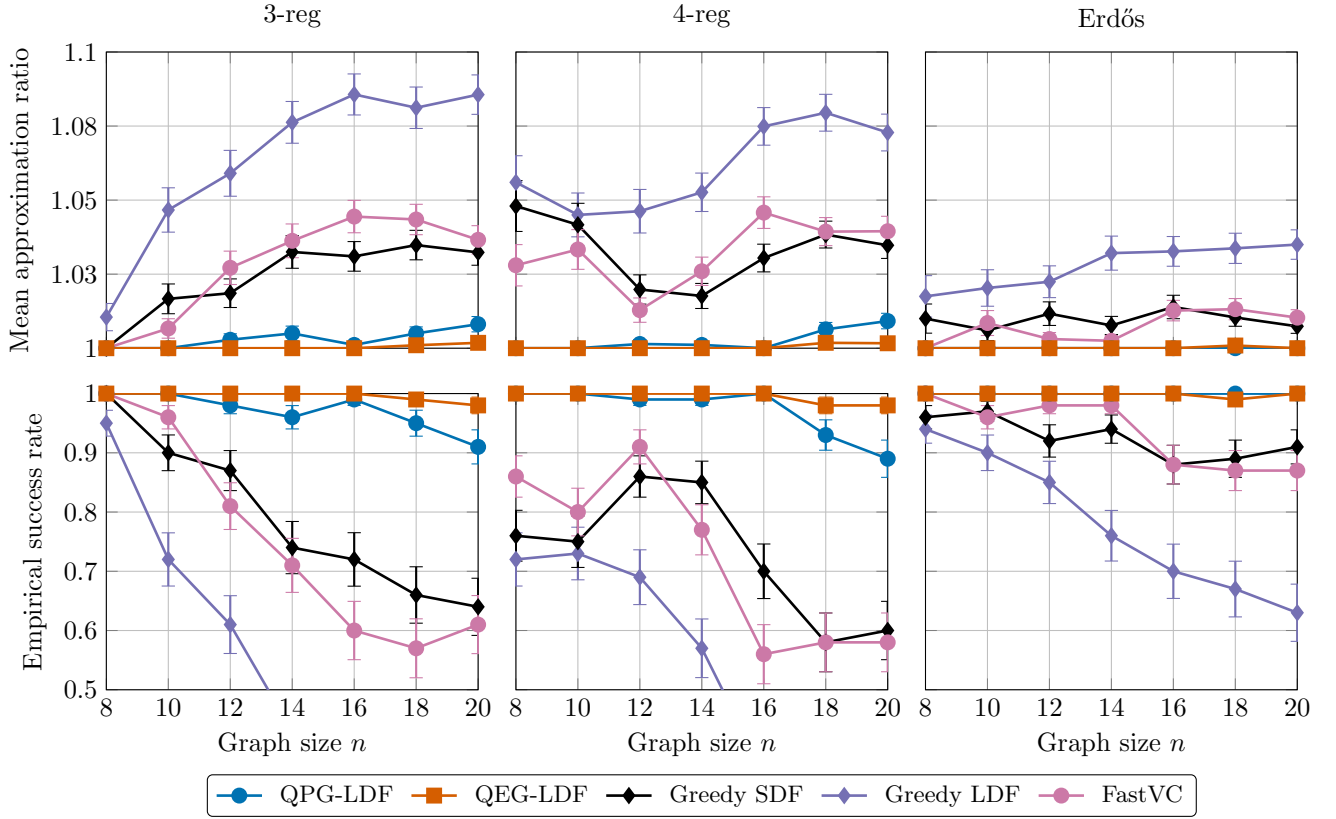}
    \caption{
        Empirical exact-recovery rate (bottom row) and mean approximation ratio (top row) at circuit depth $p=4$ across graph families.
        Each column corresponds to a different graph family: $3$-regular (left), $4$-regular (centre), and connected Erd\H{o}s--R\'enyi with $\rho=0.3$ (right) for $n=18$.
        The quantum curves are obtained from \cref{alg:quantum_greedy} for QPG-LDF (blue) and QEG-LDF (orange), with $p=4$ and the graph-family-specific walk times $t^\star$ selected at $p=4$ according to \cref{subsec:walk_time_selection}.
        Also shown are Greedy LDF (purple), Greedy SDF (black), and FastVC (pink).
        Error bars show the SEM over the 100 graphs.
    }
    \label{fig:comparison_p4}
\end{figure*}

An important practical question is how the performance of the proposed
algorithms changes when the continuous-time walk is approximated by a
finite-depth first-order Suzuki-Trotter expansion. In
\cref{fig:TrotterEffect} we report the mean approximation ratio and the
empirical success rate as functions of the
Trotter depth $p$ for QPG and QEG on 100 random 3-regular instances with $n=20$. The final point, $p=\infty$, corresponds to the exact continuous
evolution.
For the LDF strategy, both metrics improve almost monotonically as the
number of Trotter layers increases. Each additional layer reduces the
difference between the trotterised dynamics and the exact walk, leading
to a more accurate ordering of the vertex probabilities that drive the
greedy selection. Consequently, the finite-depth curves progressively
approach the exact limit. The improvement is most pronounced in the
first few layers and rapidly saturates: already at $p=4$ the performance
is close to that of the exact evolution. This behaviour agrees with the
analysis of \cref{sec:Trotter}, where the approximation error scales as
$\mathcal{O}(|E|t^2/p)$. For the moderate walk times considered here, only
a small number of Trotter steps is therefore sufficient to suppress the
error below the level at which it significantly affects the greedy
decision process.
The SDF strategy exhibits a different behaviour. The performance improves
from $p=1$ to $p=2$, already surpassing the corresponding classical
heuristics, but remains nearly constant for larger depths and up to the
exact limit. A possible explanation is that the SDF update rule is more
aggressive: low-scoring vertices trigger the inclusion of neighbouring
vertices into the cover, causing larger reductions of the instance at
each step. As a consequence, the algorithm becomes less sensitive to the
fine structure of the probability ranking generated by the quantum walk,
and further improvements in the accuracy of the Trotter approximation
have only a limited impact on the final solution quality.
The same qualitative behaviour is observed across the full range of graph
sizes considered in the benchmarks and is not specific to the $n=20$
instances shown in \cref{fig:TrotterEffect}. In particular, the rapid
convergence of the finite-depth approximation toward the exact limit for
the LDF variants, together with the earlier saturation observed for SDF,
persists for different system sizes. Similar trends are also found for
Erd\H{o}s--R\'enyi graphs with connection probability $\rho=0.3$, suggesting
that the effect of the Trotter depth is largely independent of the graph
family and reflects a general feature of the constrained quantum
dynamics underlying the algorithms.

\subsection{Comparison across graph families}
\label{subsec:comparison}

\Cref{fig:comparison_p4} compares the strongest LDF variant of each quantum scoring method, QPG-LDF and QEG-LDF, with the classical Greedy LDF, Greedy SDF, and FastVC baselines.
The quantum results are evaluated at Trotter depth $p=4$ using graph-family-specific walk times selected at Trotter depth $p=4$ on the independent calibration ensembles described in \cref{subsec:walk_time_selection}.
All points in \cref{fig:comparison_p4} are evaluated on held-out benchmark ensembles that are disjoint from the calibration ensembles.

The figure is intended as a comparison between the strongest quantum-informed variants and the classical methods rather than as an exhaustive comparison of every quantum variant.
The QPG-SDF and QEG-SDF variants are therefore not included in \cref{fig:comparison_p4}.
Their dependence on the walk time and Trotter depth is examined separately in \cref{fig:qpg-3regular-performance,fig:qeg-3regular-performance,fig:TrotterEffect}.

Across the tested graph families, QEG-LDF provides the strongest overall performance among the quantum-informed methods.
Its performance at $p=4$ remains close to that of the exact continuous-time evolution, in agreement with the Trotter-depth results in \cref{subsec:trotter_effect}.
QPG-LDF also improves on the corresponding local degree-based selection rule over the tested instances, although its marginal-probability ranking is more sensitive to finite-depth Trotterisation.

The connected Erd\H{o}s--R\'enyi instances are empirically easier for all methods considered.
At $p=4$, their mean approximation ratios are closer to one and their empirical exact-recovery rates are higher than those obtained for the regular ensembles of the same size.
One possible explanation is that these denser graphs generally have larger minimum vertex covers and hence a smaller induced-graph depth $D(G)=n-\tau(G)$ between the all-in-the-cover state and the minimum-cover layer.
Their greater degree heterogeneity may also produce more differentiated vertex scores and reduce the frequency of ambiguous greedy decisions.
The present experiments do not isolate the respective contributions of density, induced-graph depth, and degree heterogeneity, so these explanations should be regarded as hypotheses rather than established mechanisms.

\subsection{Walk-time selection}
\label{subsec:walk_time_selection}

The walk time $t$ is treated as a hyperparameter and is fixed before evaluation on the benchmark instances.
For each graph family $f$, we generate an independent calibration ensemble $\mathcal{T}_f$ containing 100 graphs with $n=18$ vertices.
Each calibration ensemble is disjoint from every benchmark ensemble used in \cref{fig:comparison_p4}, including the corresponding benchmark ensemble with $n=18$.

We calibrate the walk time separately for each score function $\mathsf{s}\in\{\mathrm{QPG},\mathrm{QEG}\}$, reduction rule $R\in\{\mathrm{LDF},\mathrm{SDF}\}$, and graph family $f$.
For each combination $(\mathsf{s},R,f)$, we perform a grid search over
\begin{equation}
\mathcal{G}_t
=
\left\{
0,\frac{\pi}{16},\frac{2\pi}{16},\ldots,\pi
\right\}.
\end{equation}
The selected walk time is
\begin{equation}
t^\star_{\mathsf{s},R,f}
=
\argmax_{t\in\mathcal{G}_t}
\widehat{P}_{\mathrm{opt}}
\bigl(t;\mathcal{T}_f,\mathsf{s},R,p=4\bigr)
\label{eq:selected_walk_time}
\end{equation}
where $\widehat{P}_{\mathrm{opt}}$ is the empirical fraction of calibration instances for which the complete recursive procedure in \cref{alg:quantum_greedy} returns a minimum vertex cover.
Ties are resolved in favour of the smallest maximising grid point.

After calibration, each value $t^\star_{\mathsf{s},R,f}$ is held fixed without further adjustment.
The same value is used for every benchmark size $n\in\{8,10,12,14,16,18,20\}$, every finite Trotter depth, and every recursion level of the greedy reduction.
In particular, the walk time is not rescaled as the residual graph shrinks.

\subsection{Choice of variant}
\label{subsec:variant_choice}
Across all tested scoring rules and greedy strategies, the LDF variants provide
the most direct vertex-cover interpretation: a high-scoring vertex is fixed into
the cover and the instance is reduced. Among these, QEG-LDF gives the strongest
and the most promising performance of the algorithms tested in
our benchmarks. In particular as reported in \cref{fig:comparison_p4}, QEG performs substantially better at low Trotter depth with performance very close to the exact ones for every type of graph studied. This is consistent with the fact that the energy score evaluates
the effect of fixing a vertex on the global cost expectation, whereas QPG uses
only a marginal inclusion probability.

The SDF variants reduce the number of recursive steps, but they also produce more aggressive updates and therefore provide less fine-grained optimisation.
In the present benchmarks, the SDF variants are useful for comparison and for illustrating the flexibility of the framework.

The benchmarks in this section evaluate the quantum scores directly for finite Trotter depths and for the exact continuous-time limit.
For the special single-layer case $p=1$, the same QPG and QEG scores can instead be estimated by sequential Monte Carlo sampling using the quantum-tree representation developed in \cref{sec:montecarlo}.
This sampling procedure avoids explicit summation over the exponentially many feasible covers and enables the larger-scale $p=1$ simulations up to $n=300$ reported there.

\section{Conclusion}
\label{sec:conclusion}

We have introduced a constraint-preserving hybrid quantum-classical framework for the minimum vertex cover and maximum independent set problems. The central construction is a projected Pauli-$X$ Hamiltonian, $H_{\mathrm{MVC}}=\sum_i \Pi_{\mathcal{N}(i)}X_i$, whose dynamics remain entirely within the feasible vertex-cover subspace. Restricted to this subspace, the Hamiltonian is the adjacency matrix of an induced feasible-state graph, so its evolution realises a continuous-time quantum walk over feasible covers rather than over the original graph.

This viewpoint clarifies the role of the quantum component. The constrained walk transports amplitude from the all-in-the-cover state toward lower-weight feasible layers, but direct sampling of an optimal cover is exponentially suppressed in the tested regime. The useful output of the quantum evolution is therefore not a full solution, but a set of per-vertex scores whose local expectation values encode nonlocal structure of the feasible-state graph. These scores drive a recursive classical reduction procedure.

Among the scoring rules studied here, the energy-based criterion gives the strongest performance.
By evaluating the conditional cost after fixing a vertex, it aggregates the resulting marginal occupations across all vertices, whereas the probability-based criterion uses only the candidate vertex's marginal occupation.
Across the tested random graph families, the resulting quantum-informed greedy algorithms outperform their corresponding classical greedy baselines at empirically selected walk times, while the strongest variant remains close to the exact continuous-time benchmark under low-depth Trotterisation.
The main practical trade-off is between measurement cost and ranking stability. Marginal-probability scores require fewer circuit evaluations but are more sensitive to finite-depth errors, whereas energy-based scores require more measurements but produce more robust rankings. For bounded-degree graphs, each Trotter layer can be implemented at constant depth, but the end-to-end resource cost also depends on recursion depth and on the number of shots needed to resolve score gaps.
The key open question is therefore not whether the feasible quantum walk can generate useful rankings on the tested instances, but how the statistical separations between competing vertex scores scale with problem size. Resolving this question, together with adaptive walk-time selection and broader comparisons against iterative quantum methods and stronger classical heuristics, will determine the regime in which constraint-preserving quantum walks can provide a practical advantage for combinatorial optimisation.

\section{Availability of data and code}
\label{sec:data-code}
To support reproducibility and independent verification of our results, we have made all relevant data and source code publicly available at
\url{https://github.com/OpenQuantumComputing/Quantum-Greedy-MVC}.

\begin{acknowledgments}
    This work has been funded by the Research Council of Norway and partners through project number 332023 (NeQst) and 296205 (FME NTRANS).

    Ruben Pariente Bassa contributed to the conceptualisation, methodology, software development, formal analysis, investigation, visualisation and manuscript writing.
    Finley A. Quinton contributed to the conceptualisation, methodology, software development, formal analysis, investigation, visualisation and manuscript writing.
    Franz G. Fuchs contributed to the conceptualisation, scientific guidance, supervision, funding acquisition, visualisation, and manuscript writing.
    Pascal Halffmann has contributed to the conceptualisation, manuscript writing, scientific guidance, and supervision.

    During manuscript preparation, the authors used OpenAI's ChatGPT to assist with language editing, stylistic improvement, consistency checks, and drafting the accompanying Lean~4 formalization.
    All scientific ideas, methods, analyses, results, and conclusions were developed and critically assessed by the authors.
    The authors reviewed all AI-assisted text and code, verified the formalization using the pinned Lean~4 environment, and take full responsibility for the manuscript and accompanying materials.
\end{acknowledgments}

\bibliography{references}

\newpage

\onecolumngrid

\appendix

\section{Quantum Fidelity Greedy algorithm}
\label[appendix]{sec:QFG}
The main greedy strategies presented in the paper select vertices based on probabilities or energy variations. Here, we introduce an alternative approach that leverages the algebraic structure of the projected Pauli operators associated with the vertices. We develop it separately for two reasons. First, as the benchmarks of \cref{sec:Results} and \cref{fig:qfg-3regular-performance} show, the fidelity score yields a consistently less discriminating ranking than QPG and QEG while incurring the same $\mathcal{O}(|V|^2)$ measurement cost as QEG, so it is the least practical of the three. Second, and more fundamentally, the echo protocol does not share the clean continuous-time quantum-walk interpretation of the other two scores: rather than reading a single observable off the walk state, it uses a four-fold echo sequence whose expansion mixes commutators of two different seed operators (see \cref{sec:StructureLieAlgebra}). For these reasons QFG is best presented as a conceptual variant that illustrates the flexibility of the projected-operator framework rather than as a primary algorithm.

For any vertex cover $C$, its complement $V\setminus C$ is an independent set, and the projected mixers associated with the vertices in this independent set commute pairwise.
Conversely, an independent set identifies a mutually commuting subset of the projected mixers.
The complement of a maximal independent set is an inclusion-wise minimal vertex cover, but it is not necessarily a minimum-cardinality vertex cover.
QFG uses the echo fidelity as a heuristic compatibility score intended to guide the greedy reduction toward high-quality solutions, without guaranteeing a maximum independent set or a minimum vertex cover.

To quantify the compatibility of a candidate vertex $v$ with the remaining active vertices, let $S_v=V\setminus\{v\}$ and $H_{\bar v}=\sum_{i\in S_v}\hat{X}_i$.
The corresponding exact collective evolution and candidate evolution are
\begin{equation*}
    A_v(t_a)
    =
    e^{it_aH_{\bar v}}, \quad
    B_v(t_b)
    =
    e^{it_b\hat{X}_v}.
\end{equation*}
For finite Trotter depth $p$, the collective evolution, echo unitary, and fidelity score are respectively
\begin{align*}
    A_{v,p}(t_a)
    &=
    \left(
    \prod_{i\in S_v}
    e^{i(t_a/p)\hat{X}_i}
    \right)^p,
    \\
    U_{v,p}(t_a,t_b)
    &=
    A_{v,p}(t_a)
    B_v(t_b)
    A_{v,p}(t_a)^\dagger
    B_v(t_b)^\dagger,
    \\
    F_{v,p}(t_a,t_b)
    &=
    \left|
    \bra{\Omega}
    U_{v,p}(t_a,t_b)
    \ket{\Omega}
    \right|^2.
\end{align*}
The product in $A_{v,p}(t_a)$ is evaluated using the same fixed vertex ordering as the other trotterised circuits.
The candidate evolution $B_v(t_b)$ contains only one generator and therefore requires no Trotter decomposition.
For $p=\infty$, $A_{v,p}(t_a)$ reduces to the exact evolution $A_v(t_a)$.
If $A_v(t_a)$ and $B_v(t_b)$ commute, the exact echo is the identity and the fidelity equals one.
Conversely, unit fidelity only shows that this echo returns $\ket{\Omega}$ up to a phase and does not by itself prove that the operators commute on the full Hilbert space.
\Cref{fig:qfg-fidelity-vs-time} shows the resulting exact fidelity profiles for the example graph as the common echo time is varied.

\begin{figure}
    \centering

\definecolor{c1}{HTML}{1D4F91}  
\definecolor{c2}{HTML}{D81B60}  
\definecolor{c3}{HTML}{2A9D8F}  
\definecolor{c4}{HTML}{F1C232}  
\definecolor{c5}{HTML}{5F5F5F}  
\definecolor{c6}{HTML}{F28E2B}  
\definecolor{c7}{HTML}{222222}  

\pgfplotsset{
  good curve/.style={
    no markers,
    solid,
    line width=2.4pt,
    opacity=0.98
  },
  bad curve/.style={
    no markers,
    solid,
    line width=1.8pt,
    opacity=0.90
  }
}

\begin{tikzpicture}
\begin{axis}[
  width=8.2cm,
  height=5.1cm,
  xlabel={Time $t$},
  ylabel={$F_v(t)$},
  xmin=0,
  xmax=3.14159265,
  xtick={0,0.7853981,1.5707963,2.3561944,3.14159265},
  xticklabels={$0$,$\pi/4$,$\pi/2$,$3\pi/4$,$\pi$},
  axis lines=left,
  axis line style={->},
  grid=none,
  tick label style={font=\small},
  label style={font=\normalsize},
  clip=false
]

\addplot+[good curve, color=c1]
  table[x=beta,y=v0,col sep=comma]
  {final_reduced_plotting/evolution_times/F_vs_beta.csv};

\addplot+[good curve, color=c3]
  table[x=beta,y=v2,col sep=comma]
  {final_reduced_plotting/evolution_times/F_vs_beta.csv};

\addplot+[good curve, color=c4]
  table[x=beta,y=v3,col sep=comma]
  {final_reduced_plotting/evolution_times/F_vs_beta.csv};

\addplot+[good curve, color=c6]
  table[x=beta,y=v5,col sep=comma]
  {final_reduced_plotting/evolution_times/F_vs_beta.csv};
  
\addplot+[bad curve, color=c2, dotted]
  table[x=beta,y=v1,col sep=comma]
  {final_reduced_plotting/evolution_times/F_vs_beta.csv};

\addplot+[bad curve, color=c5, dotted]
  table[x=beta,y=v4,col sep=comma]
  {final_reduced_plotting/evolution_times/F_vs_beta.csv};

\addplot+[bad curve, color=c7, dotted]
  table[x=beta,y=v6,col sep=comma]
  {final_reduced_plotting/evolution_times/F_vs_beta.csv};

\node[
  anchor=south west,
  fill=none,
  fill opacity=0.84,
  text opacity=1,
  inner sep=0.6pt,
  rounded corners=1pt
]
at (rel axis cs:0.35,0.55)
{
\begin{tikzpicture}[
  every node/.style={
    circle,
    draw=black,
    line width=0.42pt,
    minimum size=3.2mm,
    inner sep=0pt,
    font=\tiny
  },
  vgood/.style={
    draw=black,
    line width=0.70pt
  },
  vbad/.style={
    draw=black!60,
    line width=0.42pt
  },
  edge/.style={
    draw=black,
    line width=0.42pt
  }
]

\begin{scope}[x=0.6cm,y=0.6cm,rotate=-45]

\node[vgood, fill=c1, text=white] (1) at ( 0.20, 0.00) {1};
\node[vbad,  fill=c2, text=white] (2) at (-0.55,-1.25) {2};
\node[vgood, fill=c3, text=white] (3) at (-1.70,-1.05) {3};
\node[vgood, fill=c4, text=black] (4) at (-1.45,-0.05) {4};
\node[vbad,  fill=c5, text=white] (5) at (-0.90, 1.30) {5};
\node[vgood, fill=c6, text=black] (6) at ( 0.15, 1.40) {6};
\node[vbad,  fill=c7, text=white] (7) at ( 1.05, 0.75) {7};

\draw[edge] (1) -- (2);
\draw[edge] (1) -- (3);
\draw[edge] (1) -- (4);
\draw[edge] (1) -- (5);
\draw[edge] (1) -- (6);
\draw[edge] (1) -- (7);
\draw[edge] (2) -- (3);
\draw[edge] (2) -- (4);
\draw[edge] (3) -- (4);
\draw[edge] (4) -- (5);
\draw[edge] (5) -- (6);
\draw[edge] (6) -- (7);

\end{scope}
\end{tikzpicture}
};

\end{axis}
\end{tikzpicture}
    \caption{
        Exact echo-fidelity scores $F_{v,\infty}(t,t)$ for the vertices of the depicted small example graph, shown as functions of the common echo time $t=t_a=t_b$.
        The curve colours correspond to the vertex colours in the inset, and the separation of the scores produces the vertex ranking used by QFG.
    }
\label{fig:qfg-fidelity-vs-time}
\end{figure}
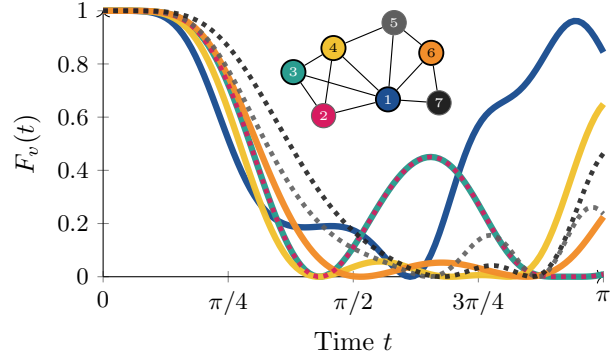

Unlike QPG and QEG, QFG probes how strongly the candidate mixer $\hat{X}_v$ fails to commute with the active block generated by $H_{\bar v}$, as observed through its action on the reference state $\ket{\Omega}$.
The corresponding short-time expansion contains nested commutators of $H_{\bar v}$ with $\hat{X}_v$, including contributions from the neighbourhood projector in $\hat{X}_v=\Pi_{\Nof{v}}X_v$.
A high fidelity indicates that the echo approximately returns $\ket{\Omega}$ and therefore suggests that the candidate mixer is relatively compatible with the active block.
A low fidelity indicates that the noncommutativity has a stronger effect on $\ket{\Omega}$ and therefore assigns the candidate a stronger conflict score.
The LDF-like and SDF-like reduction rules then determine whether the candidate itself or its neighbours are added to the cover.
Because the fidelity is evaluated only on $\ket{\Omega}$, it is a state-dependent compatibility witness rather than a general operator-level measure of commutativity.
At finite $p$, the score also includes the Trotter error in $A_{v,p}(t_a)$.

The Quantum Fidelity Greedy algorithm has a scaling of the number of circuit evaluations similar to the QEG algorithm, since at each recursive step we have to evaluate a number of conditioned fidelities that is equal to the number of vertices.

\paragraph{Smallest-fidelity (LDF-like) strategy.}
For fixed $p$, $t_a$, and $t_b$, the LDF-like variant selects
\begin{equation*}
v^*
=
\argmin_{v\in V}F_{v,p}(t_a,t_b).
\end{equation*}
A low fidelity means that the candidate mixer produces a comparatively large disturbance of the reference-state echo and therefore signals strong incompatibility with the active block.
The selected vertex $v^*$ is added to the cover and removed from the residual graph together with its incident edges.
From the operator perspective, this update removes a strongly conflicting mixer from the active block and thereby reduces the noncommuting structure that remains to be resolved.

\paragraph{Largest-fidelity (SDF-like) strategy.}
The SDF-like variant instead selects
\begin{equation*}
v^*
=
\argmax_{v\in V}F_{v,p}(t_a,t_b).
\end{equation*}
A high fidelity means that the candidate mixer produces comparatively little disturbance of the reference-state echo and therefore suggests compatibility with the active block.
The selected vertex is retained outside the cover, while all vertices in its open neighbourhood $\Nof{v^*}$ are added to the cover.
The residual graph is then obtained by removing the closed neighbourhood $\{v^*\}\cup\Nof{v^*}$.
From the independent-set perspective, this strategy greedily retains mutually nonadjacent vertices and therefore constructs a maximal independent set.
The projected mixers associated with the retained vertices commute pairwise and hence generate an Abelian operator algebra.
The echo fidelity determines the selection order and thereby influences which maximal independent set is obtained.
The complementary vertex cover is inclusion-wise minimal but is not necessarily a minimum-cardinality cover.

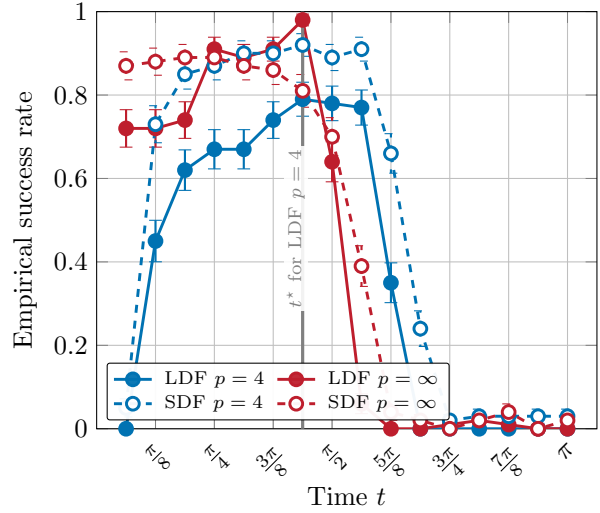
\begin{figure}
    \centering
\begin{tikzpicture}

\definecolor{c4}{RGB}{0,114,178}
\definecolor{cE}{RGB}{190,30,45}

\begin{axis}[
  width=8.2cm,
  height=7.1cm,
  xlabel={Time $t$},
  ylabel={Empirical success rate},
  xmin=0,
  xmax=17,
  xtick={2,4,6,8,10,12,14,16},
  xticklabels={
    $\frac{\pi}{8}$,
    $\frac{\pi}{4}$,
    $\frac{3\pi}{8}$,
    $\frac{\pi}{2}$,
    $\frac{5\pi}{8}$,
    $\frac{3\pi}{4}$,
    $\frac{7\pi}{8}$,
    $\pi$
  },
  x tick label style={
    font=\normalsize,
    rotate=45,
    anchor=north east
  },
  ymin=0,
  ymax=1,
  ytick={0,0.2,0.4,0.6,0.8,1.0},
  grid=both,
  tick align=inside,
  tick label style={font=\normalsize},
  label style={font=\normalsize},
  title style={font=\normalsize},
  xlabel style={yshift=7pt},
  every axis plot/.append style={
    line width=1.1pt,
    mark size=2.4pt,
    error bars/y dir=both,
    error bars/y explicit=true,
    error bars/error bar style={line width=0.8pt},
    error bars/error mark options={
      rotate=90,
      mark size=2.4pt
    }
  },
  legend columns=2,
  legend style={
    at={(axis description cs:0.02,0.02)},
    anchor=south west,
    draw=black,
    rounded corners=1pt,
    fill=white,
    fill opacity=0.95,
    text opacity=1,
    font=\scriptsize,
    column sep=0.08cm,
    row sep=-1pt,
    inner xsep=2pt,
    inner ysep=1pt,
    cells={anchor=west}
  },
  legend image code/.code={
    \draw[
      mark repeat=2,
      mark phase=2,
      #1
    ]
    plot coordinates {
      (0cm,0cm)
      (0.25cm,0cm)
      (0.50cm,0cm)
    };
  }
]

\addplot+[
  color=c4,
  mark=*,
  mark options={fill=c4,draw=c4},
  solid,
  error bars/.cd,
  y dir=both,
  y explicit
]
table[
  x expr=\coordindex+1,
  y=prob_optimal,
  y error=sem,
  col sep=comma
]
{new_plots/t_analysis/prob/QFG_LDF_3_regular_t_p4.csv};
\addlegendentry{LDF $p=4$}

\addplot+[
  color=cE,
  mark=*,
  mark options={fill=cE,draw=cE},
  solid,
  error bars/.cd,
  y dir=both,
  y explicit
]
table[
  x expr=\coordindex+1,
  y=prob_optimal,
  y error=sem,
  col sep=comma
]
{new_plots/prob/3-regular/reduced_t_analysis_qfg_ldf_exact.csv};
\addlegendentry{LDF $p=\infty$}

\addplot+[
  color=c4,
  mark=*,
  mark options={fill=white,draw=c4,solid},
  dashed,
  error bars/.cd,
  y dir=both,
  y explicit
]
table[
  x expr=\coordindex+1,
  y=prob_optimal,
  y error=sem,
  col sep=comma
]
{new_plots/t_analysis/prob/QFG_SDF_3_regular_t_p4.csv};
\addlegendentry{SDF $p=4$}

\addplot+[
  color=cE,
  mark=*,
  mark options={fill=white,draw=cE,solid},
  dashed,
  error bars/.cd,
  y dir=both,
  y explicit
]
table[
  x expr=\coordindex+1,
  y=prob_optimal,
  y error=sem,
  col sep=comma
]
{new_plots/prob/3-regular/reduced_t_analysis_qfg_sdf_exact.csv};
\addlegendentry{SDF $p=\infty$}

\addplot[
  color=gray,
  solid,
  line width=1.2pt,
  no marks,
  forget plot
]
coordinates {(7,0)(7,1)};

\node[
  gray,
  font=\scriptsize,
  fill=white,
  fill opacity=0.9,
  text opacity=1,
  inner sep=1pt,
  rotate=90,
  anchor=south
]
at (axis cs:7.18,0.48) {$t^\star$ for LDF $p=4$};

\end{axis}
\end{tikzpicture}
    \caption{
        Echo-time selection for QFG on 100 random $3$-regular graph instances with $n=16$.
        The empirical success rate is the fraction of instances for which \cref{alg:quantum_greedy} returns a minimum vertex cover after completing the full iterative procedure.
        The common echo time is varied with $t=t_a=t_b$.
        Blue curves show the Trotterised echo protocol with $p=4$, while red curves show the exact echo protocol corresponding to $p=\infty$.
        Solid lines with filled markers denote LDF-like variants, while dashed lines with open markers denote SDF-like variants.
        The vertical gray line marks the selected echo time $t^\star$ used in the subsequent benchmarks.
        Error bars show the binomial standard error over the 100 instances.
    } 
    \label{fig:qfg-3regular-performance}
\end{figure}
The common echo time $t=t_a=t_b$ controls the sensitivity of the fidelity score to the incompatibility between each candidate mixer and the active block.
\Cref{fig:qfg-3regular-performance} shows the performance of both QFG variants over $t\in[0,\pi]$ on 100 3-regular instances with $n=16$.
Both variants are sensitive to the echo time and attain their highest empirical success rates before $t=\pi/2$.
Corresponding scans for the $4$-regular and connected Erd\H{o}s--R\'enyi ($\rho=0.3$) ensembles exhibit the same qualitative dependence on $t$ and are omitted for brevity.

At Trotter depth $p=4$, the QFG performance remains close to that of the exact echo protocol over the tested range.
This agreement indicates that the vertex ranking induced by the echo fidelity is comparatively stable under the finite-depth approximation for these instances.

QFG nevertheless underperforms both QPG and QEG, including in the exact $p=\infty$ limit, indicating that the state-dependent compatibility score provides a less discriminating vertex ranking.
For suitable echo times, both QFG variants remain competitive with and can outperform their corresponding classical greedy baselines.

At recursion step $r$, QFG requires one candidate-specific echo evaluation for each vertex in the residual graph and therefore uses $\mathcal{O}(|V_r|)$ circuit evaluations.
Summed over a complete run, this gives at most $\mathcal{O}(|V|^2)$ circuit evaluations, excluding the number of measurement shots required to estimate each fidelity.

\section{Nested commutators of projected Pauli \texorpdfstring{$X$}{X} generators}
\label[appendix]{sec:StructureLieAlgebra}

In this appendix, we analyze the algebraic structure generated by projected single-qubit Pauli $X$ operators acting within the vertex-cover subspace $\mathcal{H}_{\mathrm{VC}}$. Rather than attempting a full characterisation of the resulting dynamical Lie algebra, we focus on structural properties that directly govern the behaviour of the quantum algorithms introduced in the main text.

Our approach is guided by the observation that all relevant observables, marginal selection probabilities, conditional energy estimates, and fidelity measures, can be expressed in terms of the diagonal contributions of nested commutators of the projected generators $\{\hat{X}_i\}_{i \in V}$ with respect to $Z_v$ or other Projected Paulis. The key question is therefore not the full Lie closure, but how operator support, symmetry, and constraint structure propagate under repeated commutation.

To this end, we proceed constructively. We first define the constraint-preserving projected operators and identify the graph-theoretic conditions under which projected Pauli strings are nonvanishing.  

Building on these results, we analyse the structure of nested commutators, deriving an operator light-cone that bounds the growth of support and proving selection rules that restrict which terms contribute to expectation values.
This shows how the recursive greedy scores can receive contributions from increasingly distant graph regions, in contrast to their classical degree-based counterparts.

\subsection{Constraint-preserving operators}
\label[appendix]{subsec:constraint_ops}

We begin by reviewing the projection formalism that ensures all operators respect the vertex-cover constraint. The projector onto the valid vertex-cover subspace is
\begin{equation*}
    \Pi_{\mathrm{VC}}
    = \sum_{C \in \mathrm{VC}} \ket{C}\!\bra{C}
    = \prod_{(i,j)\in E}
      \bigl(\mathbb{I} - \ket{00}\!\bra{00}_{ij}\bigr),
\end{equation*}
where the product form explicitly enforces the edge constraint $C_i + C_j \geq 1$ for every edge $(i,j) \in E$.

To implement Pauli rotations that preserve the vertex-cover constraint, we sandwich single-qubit operators between projectors.
For a vertex $i\in V$ with open neighbourhood $\Nof{i}$, define the residual induced graph $G_i^{\mathrm{res}}=G[V\setminus(\{i\}\cup\Nof{i})]$.
The projected Pauli-$X$ operator then decomposes as
\begin{equation*}
    \Pi_{\mathrm{VC}}X_i\Pi_{\mathrm{VC}}
    =
    \Pi_{\mathrm{VC}_{G_i^{\mathrm{res}}}}
    \otimes
    \left(
    \Pi_{\Nof{i}}\otimes X_i
    \right)
    =
    \Pi_{\mathrm{VC}_{G_i^{\mathrm{res}}}}
    \otimes\hat{X}_i.
\end{equation*}
Here $\Pi_{\Nof{i}}=\bigotimes_{j\in\Nof{i}}\ket{1}\!\bra{1}_j$ requires every neighbour of $i$ to be in the cover, while $\Pi_{\mathrm{VC}_{G_i^{\mathrm{res}}}}$ enforces the remaining edge constraints on vertices outside the closed neighbourhood $\{i\}\cup\Nof{i}$.
The operator $\hat{X}_i=\Pi_{\Nof{i}}X_i$ acts nontrivially only on basis states in which every neighbour of $i$ belongs to the cover, because $\Pi_{\Nof{i}}$ annihilates all other basis states.
On this control subspace, flipping vertex $i$ cannot uncover an incident edge, while all edges not incident to $i$ remain unchanged.
The operator therefore preserves the feasible vertex-cover subspace and its orthogonal complement, which is equivalently expressed by $[\hat{X}_i,\Pi_{\mathrm{VC}}]=0$.

Not every induced projected Pauli string is nonzero, and whether it survives projection depends on the graph structure of the qubits that it flips.
Let $B\subseteq V$ denote the flipping set and write $P=\left(\bigotimes_{i\in B}\sigma_i^{X/Y}\right)\otimes\left(\bigotimes_{j\in V\setminus B}\sigma_j^{I/Z}\right)$, where $\sigma_i^{X/Y}\in\{X_i,Y_i\}$ and $\sigma_j^{I/Z}\in\{\mathbb{I}_j,Z_j\}$.
The corresponding projected Pauli operator is $\Pi_{\mathrm{VC}}P\Pi_{\mathrm{VC}}$.
Using the product form of the vertex-cover projector, we obtain
\begin{equation*}
\Pi_{\mathrm{VC}}P\Pi_{\mathrm{VC}}
=
\prod_{(i,j)\in E}
\left(\mathbb{I}-\ket{00}\!\bra{00}_{ij}\right)
P
\prod_{(k,\ell)\in E}
\left(\mathbb{I}-\ket{00}\!\bra{00}_{k\ell}\right).
\end{equation*}

We analyse how each edge projector interacts with the Pauli operators, distinguishing three cases based on whether the edge endpoints lie in the flipping set $B$:

\begin{enumerate}
    \item \textbf{Both endpoints outside $B$} ($i,j \notin B$): The operator acts as $\sigma_i^{I/Z} \otimes \sigma_j^{I/Z}$, which commutes with $\ket{00}\!\bra{00}_{ij}$. Thus
    \begin{equation*}
    \bigl(\mathbb{I}-\ket{00}\!\bra{00}_{ij}\bigr)\,
    \bigl(\sigma_i^{I/Z} \otimes \sigma_j^{I/Z}\bigr)\,
    \bigl(\mathbb{I}-\ket{00}\!\bra{00}_{ij}\bigr)
    = \bigl(\sigma_i^{I/Z} \otimes \sigma_j^{I/Z}\bigr)\,
      \bigl(\mathbb{I}-\ket{00}\!\bra{00}_{ij}\bigr).
    \end{equation*}
    
    \item \textbf{One endpoint in $B$} ($i\in B$, $j\notin B$):
    Here the operator acts as $\sigma_i^{X/Y}\otimes\sigma_j^{I/Z}$.
    Direct evaluation gives
    \begin{equation*}
    \left(\mathbb{I}-\ket{00}\!\bra{00}_{ij}\right)
    \left(\sigma_i^{X/Y}\otimes\sigma_j^{I/Z}\right)
    \left(\mathbb{I}-\ket{00}\!\bra{00}_{ij}\right)
    =
    s_j\,\sigma_i^{X/Y}\otimes\ket{1}\!\bra{1}_j,
    \end{equation*}
    where $s_j=1$ when $\sigma_j^{I/Z}=\mathbb{I}_j$ and $s_j=-1$ when $\sigma_j^{I/Z}=Z_j$.
    Thus, the projector $\ket{1}\!\bra{1}_j$ requires the neighbour $j$ to be occupied.
    
    \item \textbf{Both endpoints in $B$} ($i,j \in B$): The operator acts as $\sigma_i^{X/Y} \otimes \sigma_j^{X/Y}$. Both qubits flip simultaneously, and the edge constraint is preserved only if they have opposite $Z$ eigenvalues:
    \begin{equation*}
    \bigl(\mathbb{I}-\ket{00}\!\bra{00}_{ij}\bigr)\,
    \bigl(\sigma_i^{X/Y} \otimes \sigma_j^{X/Y}\bigr)\,
    \bigl(\mathbb{I}-\ket{00}\!\bra{00}_{ij}\bigr) = \bigl(\sigma_i^{X/Y} \otimes \sigma_j^{X/Y}\bigr)\,
       \frac{\mathbb{I} - Z_i Z_j}{2}\,
    \end{equation*}
    The factor $\frac{\mathbb{I} - Z_i Z_j}{2}$ projects onto the  sector $Z_i Z_j = -1$.
\end{enumerate}

Define the residual induced graph outside the flipping set and its open neighbourhood by $G_B^{\mathrm{res}}=G[V\setminus(B\cup\Nof{B})]$.
Combining these results over all edges and collecting the neighbour projectors $\ket{1}\bra{1}_j$ for $j\in\Nof{B}$, we obtain
\begin{equation*}
\Pi_{\mathrm{VC}}P\Pi_{\mathrm{VC}}
=
\eta(B)
\left[
\left(\bigotimes_{i\in B}\sigma_i^{X/Y}\right)
\left(\prod_{(a,b)\in E(G[B])}\frac{\mathbb{I}-Z_aZ_b}{2}\right)
\right]
\otimes
\Pi_{\Nof{B}}
\otimes
\left[
\Pi_{\mathrm{VC}_{G_B^{\mathrm{res}}}}
\left(\bigotimes_{j\in V\setminus(B\cup\Nof{B})}\sigma_j^{I/Z}\right)
\right],
\end{equation*}
where $\eta(B)=(-1)^{|\{j\in\Nof{B}:\sigma_j=Z\}|}$ is the phase contributed by the $Z$ operators acting on the occupied neighbours of $B$, and $E(G[B])$ denotes the edge set of the subgraph induced by $B$.

\textbf{The bipartite constraint.}
The product of projectors $\prod_{(a,b) \in E(G[B])} \frac{\mathbb{I} - Z_a Z_b}{2}$ encodes the crucial structural requirement. Each factor projects onto the subspace where $Z_a$ and $Z_b$ have different eigenvalue. For the entire product to be nonzero, there must exist a simultaneous eigenstate where all pairs of adjacent vertices in $B$ have different $Z$ eigenvalue, equivalently, a proper $2$-colouring of $G[B]$.

Such a colouring exists if and only if $G[B]$ is bipartite.
If $G[B]$ contains an odd cycle $(v_1,v_2,\ldots,v_{2k+1},v_1)$, every edge of the cycle requires its endpoints to have opposite $Z$ eigenvalues.
Propagating these constraints around the odd cycle would require $Z_{v_1}=-Z_{v_1}$, which is impossible.
Therefore
\begin{equation*}
\prod_{(a,b)\in E(G[B])}\frac{\mathbb{I}-Z_aZ_b}{2}=0
\quad\Longleftrightarrow\quad
G[B]\text{ is non-bipartite},
\end{equation*}
which establishes the criterion for a projected Pauli operator to survive the projection.

\textbf{Summary.}
The projection guarantees constraint preservation by construction.
For a Pauli string with flipping set $B$, the projected operator is supported only on configurations in which every vertex in $\Nof{B}$ is occupied, and it can be nonzero only if $G[B]$ is bipartite.
The bipartiteness condition is formalized in the following lemma.

\begin{lemma}[Bipartiteness of the Flipping Support]
\label{lem:bipartite}
Let $P = \bigl(\bigotimes_{i \in B} \sigma_i^{X/Y}\bigr)
        \otimes
        \bigl(\bigotimes_{j \notin B} \sigma_j^{I/Z}\bigr)$
be a Pauli operator, where $B \subseteq V$ is the set of qubits on which $P$ acts as $X$ or $Y$. Then 
\begin{equation*}
\Pi_{\mathrm{VC}}\, P\, \Pi_{\mathrm{VC}} \neq 0
\quad \Longrightarrow \quad
\text{the induced subgraph } G[B] \text{ is bipartite}.
\end{equation*}
\end{lemma}

\begin{proof}
The operator $P$ flips the computational basis state of every qubit in $B$ simultaneously. For the projected operator $\Pi_{\mathrm{VC}}\,P\,\Pi_{\mathrm{VC}}$ to be nonzero, there must exist at least one valid vertex-cover configuration $\ket{C}$ such that $P\ket{C}$ is also a valid vertex cover.

Consider any edge $(u,v) \in E$ with both endpoints in $B$. The simultaneous flip changes both $C_u$ and $C_v$. For the edge constraint $C_u + C_v \geq 1$ to remain satisfied before and after the flip, the two vertices must have had complementary initial occupations: one occupied ($C = 1$) and one unoccupied ($C = 0$). This complementarity condition must hold for every edge within $B$.

This is precisely the requirement for a proper $2$-colouring of the induced subgraph $G[B]$, which exists if and only if $G[B]$ is bipartite. If $G[B]$ contains an odd cycle, no valid $2$-colouring exists, meaning no vertex-cover configuration can be mapped to another valid cover by the simultaneous flip. Therefore $\Pi_{\mathrm{VC}}\,P\,\Pi_{\mathrm{VC}} = 0$ whenever $G[B]$ is non-bipartite.
\end{proof}

\subsection{Nested commutators}\label[appendix]{nested_com}
The QPG selection probability $P_v(t)$ in \cref{eq:prob_campbell} is governed by nested commutators of $H_{\mathrm{MVC}}$ with $Z_v$.
The QEG energy estimate $E_v(t)$ in \cref{eq:energy_campbell} instead involves nested commutators of the candidate-conditioned Hamiltonian $H_{\mathrm{MVC}}^{(v)}$ with the Pauli-$Z$ terms in the cost Hamiltonian $H_C=\frac{|V|}{2}\mathbb{I}-\frac{1}{2}\sum_{i\in V}Z_i$.
For QFG, the echo expansion involves commutators generated by the candidate mixer $\hat{X}_v$ and the active-block Hamiltonian $H_{\bar v}=\sum_{i\in S_v}\hat{X}_i$, where $S_v=V\setminus\{v\}$.

The first nontrivial commutators already reveal higher-order constraint-preserving transitions.
For vertices $i$ and $j$, one obtains
\begin{equation}
[\hat{X}_i,\hat{X}_j]
=
-\frac{i}{2}
\Pof{i,j}
\otimes
(X_iY_j-Y_iX_j)\delta_E(i,j)
=
-\frac{i}{2}\widehat{XY}_{ij}\delta_E(i,j).
\label{eq:projected-mixer-commutator}
\end{equation}
Here $\delta_E(i,j)=1$ if $(i,j)\in E$ and $\delta_E(i,j)=0$ otherwise.
Thus, nonadjacent generators commute, whereas adjacent generators produce a two-body operator supported on the corresponding edge and conditioned on the surrounding neighbourhood, as illustrated in \cref{fig:twostar}.

This pattern continues under further nested commutators: starting from single-vertex generators $\hat{X}_i$, one generates increasingly nonlocal operators supported on connected subgraphs, while still preserving feasibility at every step. In this sense, the algebra generated by the projected Pauli operators provides a hierarchy of allowed transitions inside $\mathcal{H}_{\mathrm{VC}}$, ranging from local conditional flips to collective multi-qubit moves. This higher-order structure is the algebraic motivation of the nonlocal information extracted later by the greedy selection rules.

\begin{figure*}
    \centering
    \begin{subfigure}{0.42\linewidth}
        \centering
        \scalebox{0.8}{$
[\hat{\sigma}_5, \hat{\sigma}_6] = 
\vcenter{\hbox{
\begin{tikzpicture}[
scale=.75,
  line/.style={draw=black!60, line width=0.8pt},
  outer/.style={circle, fill=cyan!35, draw=black, minimum size=8mm},
  center/.style={circle, fill=none, draw=black, line width=0.8pt, minimum size=8mm}
]
    \begin{scope}[xshift=-3cm]
      \node[center,
        path picture={
          \path[fill=red!35]
            (path picture bounding box.north west) rectangle
            (path picture bounding box.east |- path picture bounding box.center);
          \path[fill=cyan!35]
            (path picture bounding box.west |- path picture bounding box.center)
            rectangle
            (path picture bounding box.south east);
        }
      ] (cL) at (0,0) {$q_5$};
    
      \foreach \i/\q in {45/1,135/2,225/3,315/4} {
        \node[outer] (nL\i) at (\i:2.3) {$q_{\q}$};
        \draw[line] (cL) -- (nL\i);
      }
    \end{scope}
    
    \begin{scope}[xshift=3cm]
      \node[center,
          path picture={
        \path[fill=cyan!35]
          (path picture bounding box.north west) rectangle
          (path picture bounding box.east |- path picture bounding box.center);
        \path[fill=red!35]
          (path picture bounding box.west |- path picture bounding box.center)
          rectangle
          (path picture bounding box.south east);
      }
      ] (cR) at (0,0) {$q_6$};
    
      \foreach \i/\q in {45/7,135/8,225/9,315/10} {
        \node[outer] (nR\i) at (\i:2.3) {$q_{\q}$};
        \draw[line] (cR) -- (nR\i);
      }
    \end{scope}
    
    \draw[line] (cL) -- (cR);
\end{tikzpicture}
}}
$}
        \caption{}
    \end{subfigure}
    \hspace{0.03\linewidth}
    \begin{subfigure}{0.42\linewidth}
        \centering
        \input{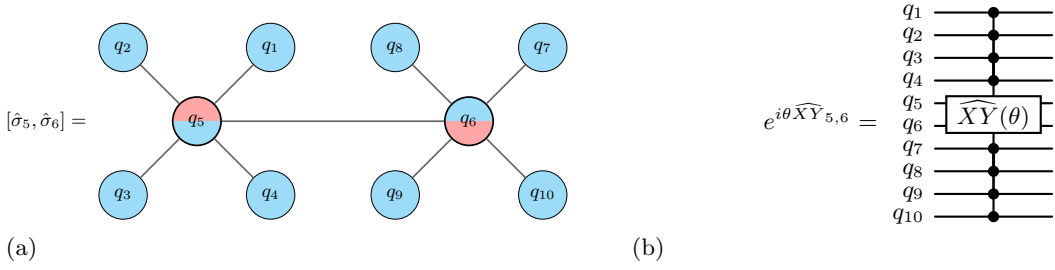}
        \vspace{-1\baselineskip}
        \caption{}
    \end{subfigure}
    \caption{
    (a) Graphical representation of two adjacent vertices, $q_5$ and $q_6$, and their respective neighbourhoods. (b) Corresponding quantum circuit for the higher-order multi-controlled operation generated by the nested commutator $[\hat{\sigma}_5, \hat{\sigma}_6]$, implementing a controlled rotation ${\widehat{XY}_{5,6}}(\theta)$,where $\widehat{XY}$ represent the XY mixer term, conditioned on all neighbouring qubits being in the state $|1\rangle$.
    }
    \label{fig:twostar}
\end{figure*}

To treat the relevant expansions uniformly, define the $r$-fold nested commutator of an operator $B$ with an operator $A$ recursively by
\begin{equation*}
[A,B]_0:=B,
\qquad
[A,B]_r:=[A,[A,B]_{r-1}]
\quad\text{for }r\geq 1.
\end{equation*}
The QPG expansion contains $[H_{\mathrm{MVC}},Z_v]_r$, while the QEG expansion contains $[H_{\mathrm{MVC}}^{(v)},Z_\ell]_r$.
The QFG echo expansion contains nested commutators generated by $H_{\bar v}$ and $\hat{X}_v$, with the decomposition $\hat{X}_v=\Pi_{\Nof{v}}X_v$ introducing contributions seeded by both $X_v$ and $\Pi_{\Nof{v}}$.

\paragraph{Pauli-support propagation.}
For $T\subseteq\Nof{i}$, define $Z_T:=\prod_{j\in T}Z_j$.
Expanding the neighbourhood projector gives
\begin{equation*}
    \widehat X_i
    =
    \Pi_{\Nof{i}}X_i
    =
    2^{-\deg(i)}
    \sum_{T\subseteq\Nof{i}}
    (-1)^{|T|}X_iZ_T.
\end{equation*}
For a Pauli string $Q$, let $\operatorname{supp}(Q)$ denote its full support and let $\operatorname{supp}_{XY}(Q):=\{j:Q_j\in\{X_j,Y_j\}\}$ denote its off-diagonal support.
The commutator $[X_iZ_T,Q]$ can be nonzero only if $Q$ has a $Y$ or $Z$ factor at $i$, or if $Q$ has an $X$ or $Y$ factor at some vertex in $T$.
When the commutator is nonzero, the resulting Pauli string has full support contained in $\operatorname{supp}(Q)\cup\{i\}\cup T$ and off-diagonal support contained in $\operatorname{supp}_{XY}(Q)\cup\{i\}$.
The support of an operator is defined as the union of the supports of the Pauli strings occurring with nonzero coefficients in its Pauli expansion.

\begin{proposition}[Operator light cone]
    \label{prop:lightcone}
    For $C_0=Z_v$ and $C_r=[H_{\mathrm{MVC}},C_{r-1}]$, one has $\operatorname{supp}(C_r)\subseteq\Ball{r}{v}$ for every $r\geq0$.
\end{proposition}

\begin{proof}
    For every Pauli string $Q$ occurring in $C_r$, we establish simultaneously that
    \begin{equation*}
        \operatorname{supp}(Q)\subseteq\Ball{r}{v}
    \end{equation*}
    and, for $r\geq1$, that
    \begin{equation*}
        \operatorname{supp}_{XY}(Q)\subseteq\Ball{r-1}{v}.
    \end{equation*}
    At $r=0$, one has $C_0=Z_v$, whose full support is $\{v\}=\Ball{0}{v}$ and whose off-diagonal support is empty.
    Assume that these inclusions hold at order $r$, and consider a Pauli string generated at order $r+1$ by a nonzero commutator $[X_iZ_T,Q]$, where $Q$ occurs in $C_r$ and $T\subseteq\Nof{i}$.
    If $Q$ has a $Y$ or $Z$ factor at $i$, then $i\in\operatorname{supp}(Q)\subseteq\Ball{r}{v}$.
    If $Q$ instead has an $X$ or $Y$ factor at some $j\in T$, then $j\in\operatorname{supp}_{XY}(Q)$.
    At $r=0$, the off-diagonal support of $Q=Z_v$ is empty, so only the first case occurs.
    For $r\geq1$, one has $j\in\Ball{r-1}{v}$ and hence $i\in\Ball{r}{v}$.
    Thus every nonzero contribution satisfies $i\in\Ball{r}{v}$.
    Because $T\subseteq\Nof{i}$, its full support is contained in
    \begin{equation*}
        \operatorname{supp}(Q)\cup\{i\}\cup T \subseteq \Ball{r+1}{v}.
    \end{equation*}
    Its off-diagonal support is contained in $\operatorname{supp}_{XY}(Q)\cup\{i\}\subseteq\Ball{r}{v}$.
    Cancellations between Pauli strings cannot enlarge either support.
    The claimed inclusion follows by induction.
\end{proof}

The support bookkeeping and symplectic Pauli step are additionally verified in Lean~4 using the pinned project provided with the accompanying code described in \cref{sec:data-code}.

\begin{corollary}[Global reach of the light-cone bound]
    \label{cor:global}
    If $G$ is connected and $r\geq\operatorname{diam}(G)$, then $\Ball{r}{v}=V$, so the light-cone bound no longer excludes any vertex from the support of $C_r$.
\end{corollary}

Consequently, at commutator orders of at least $\operatorname{diam}(G)$, the bound permits global support, although the actual support may be smaller.

\subsection{Parity selection rules}
    \label[appendix]{subsec:diagonal_survival}

The QPG and QEG observables contain expectation values of nested commutators in the real reference state $\ket{\Omega}=\ket{1}^{\otimes|V|}$.
A transpose-symmetry argument eliminates all odd-order contributions.

\begin{lemma}[Parity selection]
    \label{lem:alternating_symmetry}
    Let $A$ and $C_0$ be real symmetric operators in the computational basis, and define $C_r=[A,C_{r-1}]$ for $r\geq1$.
    Then $C_r^T=(-1)^rC_r$.
    Consequently, $\bra{\Omega}C_{2k+1}\ket{\Omega}=0$ for every real state $\ket{\Omega}$ and every $k\geq0$.
\end{lemma}

\begin{proof}
    The claim holds at $r=0$ because $C_0^T=C_0$.
    Assume that $C_{r-1}^T=(-1)^{r-1}C_{r-1}$.
    Using $A^T=A$ gives
    \begin{equation*}
    C_r^T
    =
    [A,C_{r-1}]^T
    =
    [C_{r-1}^T,A]
    =
    (-1)^r[A,C_{r-1}]
    =
    (-1)^rC_r.
    \end{equation*}
    For odd $r$, the operator $C_r$ is antisymmetric.
    Because $\ket{\Omega}$ is real, one has $\bra{\Omega}C_r\ket{\Omega}=\bra{\Omega}C_r^T\ket{\Omega}=-\bra{\Omega}C_r\ket{\Omega}$, and hence the expectation value vanishes.
\end{proof}

\begin{corollary}[Even-order reduction]
    \label{cor:odd_selection}
    Both $H_{\mathrm{MVC}}=\sum_j\hat X_j$ and $H_{\mathrm{MVC}}^{(v)}=\sum_{j\neq v}\hat X_j$ are real symmetric, so only even-order commutators contribute to the QPG and QEG expansions:
    \begin{align}
    P_v(t)
    &=
    \frac{1}{2}
    -
    \frac{1}{2}
    \sum_{k=0}^{\infty}
    \frac{(-it)^{2k}}{(2k)!}
    \bra{\Omega}
    \big[H_{\mathrm{MVC}},Z_v\big]_{2k}
    \ket{\Omega},
    \label{eq:prob_odd}
    \\
    E_v(t)
    &=
    \frac{|V|}{2}
    -
    \frac{1}{2}
    \sum_{k=0}^{\infty}
    \sum_{\ell\in V}
    \frac{(-it)^{2k}}{(2k)!}
    \bra{\Omega}
    \big[H_{\mathrm{MVC}}^{(v)},Z_\ell\big]_{2k}
    \ket{\Omega}.
    \label{eq:energy_odd}
    \end{align}
\end{corollary}

\section{CTQW state analysis}
\label[appendix]{sec:CTQW_state_analysis}

In this section we analyze different properties of the continuous quantum walk of the Vertex cover state graph initialise in the root state $\ket{\Omega}$ . In particular we will focus on three different aspects that are relevant for understanding its performance: (i) the scaling of gate counts and circuit depth required to implement the CTQW respect to the number of vertices and Trotter steps, (ii) The Trotter error analysis, and (iii) the entanglement scaling induced by its multi-controlled structure.

As we show, the circuit can remain shallow for bounded-degree graphs due to parallelisation, while the numerical results exhibit strong entanglement and the operator analysis permits contributions from increasingly distant graph regions.

We first quantify the gate complexity and parallelisation scaling of the state preparation. 
\subsection{Gate count and circuit depth}
\label[appendix]{sec:scaling}
 
We quantify the cost of preparing the trotterised continuous-time quantum
walk (CTQW) state.
A single Trotter step of the walk unitary is
\begin{equation}
    U_{\mathrm{step}}(\tau)
    =\prod_{i\in V}e^{i\tau \Pof{i}X_i}=
    \prod_{i\in V} C_{N(i)}R_{X_i}(\tau),
    \qquad
    \tau=\frac{t}{p}.
    \label{eq:trotterstep}
\end{equation}
The full state-preparation circuit is
\begin{equation*}
U_p(t)=U_{\mathrm{step}}(\tau)^p .
\end{equation*}
Throughout, we write $n=|V|$, $m=|E|$, $d_i=|N(i)|$, and
$\Delta=\max_i d_i$.

\paragraph{The elementary block is a controlled-\texorpdfstring{$\mathrm{SU}(2)$}{SU(2)} gate.}
The gate cost of implementing the Trotter step defined in
\cref{eq:trotterstep} is governed by the per-vertex block
$C_{N(i)}R_{X_i}(\tau)$, a $d_i$-controlled
single-qubit $X$-rotation.
The crucial structural observation is that the target gate
$RX(\tau)=e^{i\tau X}$ is special unitary, $RX(\tau)\in\mathrm{SU(2)}$, and
moreover has a real-valued diagonal $(\cos\tau,\cos\tau)$. This places each
block in the most favourable class of multi-controlled gates: whereas a
general multi-controlled $U(2)$ gate requires a number of two-qubit gates
that is \emph{quadratic} in the number of controls in the ancilla-free
setting~\cite{barenco1995elementary}, multi-controlled $\mathrm{SU(2)}$
gates admit \emph{linear} ancilla-free decompositions~\cite{vale2023su2}.

\paragraph{Cost of a single multi-control.}
\Cref{tab:mcucost} summarises the cost of one $d_i$-controlled
$\mathrm{SU}(2)$ rotation in three resource regimes. The table reports
two-qubit gate count, T-count, and circuit depth separately, since these
resources scale differently. The T-count includes Clifford+$T$ synthesis of
the elementary single-qubit rotations to precision $\epsilon$.
 
\begin{table}
  \centering
  \renewcommand{\arraystretch}{1.35}
  \begin{tabular}{l c c c}
    \hline\hline
    Regime & CNOTs & T gates & Depth \\
    \hline
    Ancilla-free, general $U(2)$~\cite{barenco1995elementary}
      & $\mathcal{O}(d_i^2)$ & $\mathcal{O}\!\big(d_i^2\log\tfrac1\varepsilon\big)$ & $\mathcal{O}(d_i^2)$ \\
    Ancilla-free, $\mathrm{SU(2)}$~\cite{vale2023su2,dasilva2022lineardepth}
      & $\mathcal{O}(d_i)$ & $\mathcal{O}\!\big(d_i\log\tfrac1\varepsilon\big)$ & $\mathcal{O}(d_i)$ \\
    $\mathcal{O}(d_i)$ clean ancillas~\cite{claudon2024polylog}
      & $\mathcal{O}(d_i)$ & $\mathcal{O}\!\big(d_i\log\tfrac1\varepsilon\big)$ & $\mathcal{O}(\log d_i)$ \\
    \hline\hline
  \end{tabular}
  \caption{Cost of one $d_i$-controlled $\mathrm{SU(2)}$ rotation.  The factor $\log(1/\varepsilon)$
  is the per-rotation Clifford+T synthesis cost~\cite{rossselinger2016}: the
  $\mathrm{SU(2)}$ decomposition uses $\mathcal{O}(d_i)$ elementary single-qubit
  rotations, each synthesised to precision $\varepsilon$.}
  \label{tab:mcucost}
\end{table}
 
\paragraph{Total gate count per Trotter step.}
Summing the multi-control cost over all vertices and using the identity
$\sum_{i\in V} d_i = 2m$, the two-qubit gate count of one Trotter step is
\begin{equation*}
  \mathrm{CNOT}\big(U_{\mathrm{step}}\big)
  \;=\; \sum_{i\in V} \mathcal{O}(d_i) \;=\; \mathcal{O}(m) \;=\; \mathcal{O}(n\Delta),
\end{equation*}
linear in the number of edges. The fault-tolerant non-Clifford cost is
$\mathcal{O}\!\big(m\log\tfrac1\varepsilon\big)$ by the same summation. 

 \paragraph{Depth and parallelisation.}
The Trotter step defined in \cref{eq:trotterstep} is implemented as a product
of neighbourhood-controlled rotations. Its circuit depth is determined by the
depth of the individual blocks $C_{N(i)}R_{X_i}(\tau)$ and by the extent to
which blocks with disjoint support can be applied in parallel.
Two multi-control rotations can be applied in parallel \emph{if and only if their
qubit supports are disjoint}, i.e.
\begin{equation*}
    \big(\{i\}\cup \Nof{i}\big)\,\cap\,\big(\{j\}\cup \Nof{j}\big)=\varnothing.
\end{equation*}
This fails precisely when $i$ and $j$ are adjacent or share a common
neighbour, i.e. when $\mathrm{dist}_G(i,j)\le 2$. Equivalently, $i$ and $j$
have overlapping support exactly when they are adjacent in the square graph
$G^2=(V,E^2)$, which has an edge between any two vertices at graph distance
at most two in $G$. A proper vertex colouring of $G^2$ therefore partitions
the blocks into colour classes whose members have mutually disjoint support
and can be scheduled simultaneously.
 
\begin{proposition}[Trotter-step depth]
\label{prop:depth}
Let $\chi(G^2)$ denote the chromatic number of the square graph and let
$d_{\mathrm{gate}}$ be the depth of a single $\Delta$-controlled
$\mathrm{SU(2)}$ rotation in the chosen resource regime. Then one Trotter
step can be scheduled in
\begin{equation*}
    \mathrm{Depth}\big(U_{\mathrm{step}}\big)\;\le\;\chi(G^2)\,d_{\mathrm{gate}}.
\end{equation*}
This bound assumes that the ancillary resources required by all blocks within a colour class are available concurrently.
Using the standard bound $\chi(G^2)\le \Delta^2+1$, this gives
$\mathrm{Depth}(U_{\mathrm{step}})=\mathcal{O}(\Delta^3)$ in the ancilla-free
$\mathrm{SU(2)}$ regime ($d_{\mathrm{gate}}=\mathcal{O}(\Delta)$) and
$\mathcal{O}(\Delta^2\log\Delta)$ with $\mathcal{O}(\Delta)$ clean ancillas per block
($d_{\mathrm{gate}}=\mathcal{O}(\log\Delta)$).
\end{proposition}
 
\begin{proof}
    Take a proper vertex colouring of $G^2$ with $\chi(G^2)$ colour classes.
    Within each colour class, the corresponding blocks have pairwise disjoint
    support and can therefore be applied in parallel, with depth at most
    $d_{\mathrm{gate}}$. Concatenating the $\chi(G^2)$ classes yields the stated
    depth bound. The maximum degree of $G^2$ is at most
    $\Delta(\Delta-1)+\Delta\le\Delta^2$, so greedy colouring gives
    $\chi(G^2)\le\Delta^2+1$; substituting the per-regime $d_{\mathrm{gate}}$
    from \cref{tab:mcucost} gives the two explicit bounds.
\end{proof}
 
\begin{corollary}[Constant-depth layers for bounded degree]
\label{cor:bounded}
If $\Delta=\mathcal{O}(1)$, then $\chi(G^2)=\mathcal{O}(1)$ and $d_{\mathrm{gate}}=\mathcal{O}(1)$, so each
Trotter step has depth $\mathrm{Depth}(U_{\mathrm{step}})=\mathcal{O}(1)$, independent
of the system size $n$. The full state-preparation circuit $U_p(t)$ then has
depth $\mathcal{O}(p)$, with a two-qubit gate count $\mathcal{O}(pm)=\mathcal{O}(pn)$ and non-Clifford
count $\mathcal{O}\!\big(pn\log\tfrac1\varepsilon\big)$.
\end{corollary}

 \paragraph{Full circuit and ancilla overhead.}
For general $\Delta$, the complete $p$-step preparation circuit has two-qubit gate count $\mathcal{O}(pm)=\mathcal{O}(pn\Delta)$ and depth $\mathcal{O}\!\bigl(p\,\chi(G^2)d_{\mathrm{gate}}\bigr)$.
The ancilla-free $\mathrm{SU}(2)$ decomposition requires no clean ancillas and gives depth $\mathcal{O}\!\bigl(p\,\chi(G^2)\Delta\bigr)$.
The logarithmic-depth implementation requires $\mathcal{O}(\Delta)$ clean ancillas for every block executed concurrently.
Executing all blocks in each colour class in parallel can therefore require $\mathcal{O}(n)$ clean ancillas in the worst case and gives depth $\mathcal{O}\!\bigl(p\,\chi(G^2)\log\Delta\bigr)$.
Reusing a single register of $\mathcal{O}(\Delta)$ clean ancillas reduces the ancilla overhead but requires the corresponding blocks to be executed sequentially, giving a worst-case depth of $\mathcal{O}(pn\log\Delta)$.

\subsection{Trotter error analysis}
\label[appendix]{sec:Trotter}
We now analyse the error introduced by the trotterisation of the continuous-time quantum walk evolution.  
The exact continuous-time evolution is
\begin{equation*}
U(t)=e^{itH_{\mathrm{MVC}}}.
\end{equation*}

Since the local terms $\hat{X}_i$ do not generally commute, the exact evolution is approximated through a first-order Suzuki--Trotter decomposition:
\begin{equation*}
U_p(t)
=
\left(
\prod_{i \in V}
e^{i\frac{t}{p} \hat{X}_i}
\right)^p ,
\end{equation*}
where $p$ denotes the number of Trotter steps.

To quantify the approximation quality, we evaluate the state error with respect to the exact evolved state initialised in the root state $\ket{\Omega}$:
\begin{equation*}
\epsilon_p(t)
=
\left\|
U_p(t)\ket{\Omega}
-
e^{itH_{\mathrm{MVC}}}\ket{\Omega}
\right\|_2 ,
\end{equation*}.

For a first-order product formula, the operator norm error satisfies the standard bound
\begin{equation*}
\left\|
U_p(t)-e^{itH_{\mathrm{MVC}}}
\right\|
=
\mathcal{O}\!\left(
\frac{t^2}{p}
\sum_{i<j}\|[\hat{X}_i,\hat{X}_j]\|
\right).
\end{equation*}
The scaling of the Trotter error is therefore controlled by the commutator structure of the constrained mixer terms.

Disjoint support is sufficient for two mixer terms to commute, but it is not necessary because nonadjacent vertices at distance two can have overlapping neighbourhood projectors.
The exact commutator in \cref{eq:projected-mixer-commutator} shows that $[\hat{X}_i,\hat{X}_j]=0$ whenever $(i,j)\notin E$.
Thus, at most $|E|$ pairwise commutators contribute to the first-order product-formula bound.
Moreover, \cref{eq:projected-mixer-commutator} gives $\|[\hat{X}_i,\hat{X}_j]\|\leq1$.
Because the state error is bounded by the corresponding operator-norm error, one obtains
\begin{equation*}
    \epsilon_p(t)
    \leq
    \left\|
    U_p(t)-e^{itH_{\mathrm{MVC}}}
    \right\|
    =
    \mathcal{O}\!\left(\frac{|E|t^2}{p}\right).
\end{equation*}

For bounded-degree graph families, this worst-case global-state bound scales as $\mathcal{O}(n/p)$ at fixed evolution time, so maintaining a size-independent bound may require increasing the Trotter depth with system size.
This does not directly determine the accuracy of the greedy algorithm, which depends on the induced vertex ranking rather than on global-state fidelity.
The numerical benchmarks therefore assess the stability of the algorithmic performance under low-depth Trotterisation directly.

This scaling is consistent with the numerical results shown in \cref{fig:trotter_error}.  
The left panel shows the dependence of the state error on the number of Trotter steps $p$ for a fixed 3-regular graph with $n=18$ vertices. The observed behaviour follows the expected power-law decay
\begin{equation*}
\epsilon_p(t)\sim p^{-1},
\end{equation*}
characteristic of first-order Trotter formulas.
The centre panel reports the dependence of the error on the evolution time $t$ for different values of $p$. The growth is approximately quadratic in time for sufficiently small $t$, in agreement with the perturbative prediction
\begin{equation*}
\epsilon_p(t)\propto \frac{t^2}{p}.
\end{equation*}
The right panel shows the dependence of the error on graph size for ensembles of random $3$-regular graphs at fixed evolution time and fixed Trotter depth.
Over the tested size range, the error increases approximately linearly with $n$, consistent with $|E|=3n/2$ and the corresponding product-formula bound.

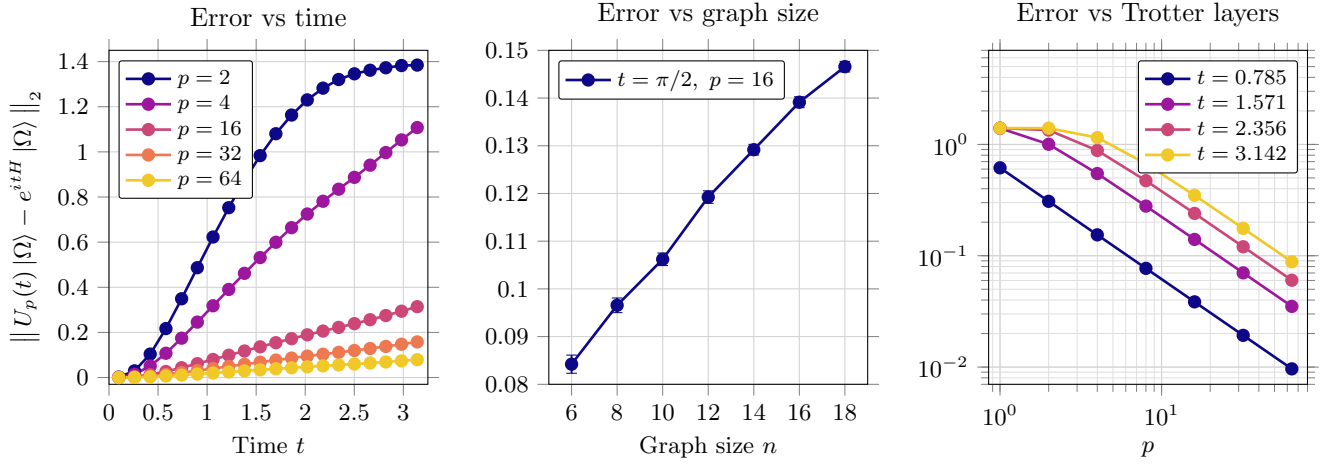
\begin{figure}
     \centering
%
%
\begin{tikzpicture}

\definecolor{cP2}{RGB}{13,8,135}     
\definecolor{cP4}{RGB}{156,23,158}   
\definecolor{cP16}{RGB}{204,71,120}  
\definecolor{cP32}{RGB}{237,121,83}  
\definecolor{cP64}{RGB}{240,200,40}  

\begin{groupplot}[
  group style={
    group size=3 by 1,
    horizontal sep=1.6cm,
    vertical sep=0.5cm,
  },
  width=5.8cm,
  height=6.0cm,
  grid=both,
  grid style={line width=0.3pt, draw=gray!25},
  major grid style={line width=0.4pt, draw=gray!35},
  tick align=outside,
  tick label style={font=\small},
  label style={font=\small},
  title style={font=\normalsize},
  every axis plot/.append style={
    line width=1.0pt,
    mark size=2.0pt,
    error bars/y dir=both,
    error bars/y explicit=true,
    error bars/error bar style={line width=0.7pt},
    error bars/error mark options={rotate=90,mark size=2.0pt}
  },
]

\nextgroupplot[
  title={Error vs time},
  xlabel={Time  $t$},
  ylabel={$\bigl\|\,U_p(t)\,|\Omega\rangle - e^{itH}\,|\Omega\rangle\,\bigr\|_2$},
  xmin=0, xmax=3.25,
  xtick={0,0.5,1.0,1.5,2.0,2.5,3.0},
  ymin=-0.03, ymax=1.45,
  ytick={0,0.2,0.4,0.6,0.8,1.0,1.2,1.4},
  legend pos=north west,
  legend cell align=left,
  legend style={
    draw=black, rounded corners=1pt, font=\footnotesize,
    inner xsep=2pt, inner ysep=1.5pt, row sep=-1pt
  },
]

\addplot[color=cP2,  mark=*, mark options={fill=cP2,  draw=cP2},  solid]
table[x=beta, y=mean_error, y error=sem_error, col sep=comma]{figures/trotter/err_vs_beta_p2.csv};
\addlegendentry{$p=2$}

\addplot[color=cP4,  mark=*, mark options={fill=cP4,  draw=cP4},  solid]
table[x=beta, y=mean_error, y error=sem_error, col sep=comma]{figures/trotter/err_vs_beta_p4.csv};
\addlegendentry{$p=4$}

\addplot[color=cP16, mark=*, mark options={fill=cP16, draw=cP16}, solid]
table[x=beta, y=mean_error, y error=sem_error, col sep=comma]{figures/trotter/err_vs_beta_p16.csv};
\addlegendentry{$p=16$}

\addplot[color=cP32, mark=*, mark options={fill=cP32, draw=cP32}, solid]
table[x=beta, y=mean_error, y error=sem_error, col sep=comma]{figures/trotter//err_vs_beta_p32.csv};
\addlegendentry{$p=32$}

\addplot[color=cP64, mark=*, mark options={fill=cP64, draw=cP64}, solid]
table[x=beta, y=mean_error, y error=sem_error, col sep=comma]{figures/trotter/err_vs_beta_p64.csv};
\addlegendentry{$p=64$}

\nextgroupplot[
  title={Error vs graph size},
  xlabel={Graph size $n$},
  xmin=5, xmax=19,
  xtick={6,8,10,12,14,16,18},
  ymin=0.080, ymax=0.150,
  ytick={0.08,0.09,0.10,0.11,0.12,0.13,0.14,0.15},
  yticklabel style={/pgf/number format/fixed, /pgf/number format/precision=2},
  scaled y ticks=false,
  legend pos=north west,
  legend cell align=left,
  legend style={
    draw=black, rounded corners=1pt, font=\footnotesize,
    inner xsep=2pt, inner ysep=1.5pt
  },
]

\addplot[color=cP2, mark=*, mark options={fill=cP2, draw=cP2}, solid]
table[x=n, y=mean_error, y error=sem_error, col sep=comma]
{figures/trotter/state_error_vs_n_data.csv};
\addlegendentry{$t=\pi/2,\ p=16$}

\nextgroupplot[
  title={Error vs Trotter layers},
  xlabel={$p$},
  xmode=log, ymode=log,
  log basis x=10, log basis y=10,
  xmin=0.85, xmax=80,
  ymin=7e-3, ymax=7,
  legend pos=north east,
  legend cell align=left,
  legend style={
    draw=black, rounded corners=1pt, font=\footnotesize,
    inner xsep=2pt, inner ysep=1.5pt, row sep=-1pt
  },
]

\addplot[color=cP2,  mark=*, mark options={fill=cP2,  draw=cP2},  solid]
table[x=p, y=mean_error, y error=sem_error, col sep=comma]{figures/trotter/err_vs_p_t0785.csv};
\addlegendentry{$t=0.785$}

\addplot[color=cP4,  mark=*, mark options={fill=cP4,  draw=cP4},  solid]
table[x=p, y=mean_error, y error=sem_error, col sep=comma]{figures/trotter/err_vs_p_t1571.csv};
\addlegendentry{$t=1.571$}

\addplot[color=cP16, mark=*, mark options={fill=cP16, draw=cP16}, solid]
table[x=p, y=mean_error, y error=sem_error, col sep=comma]{figures/trotter/err_vs_p_t2356.csv};
\addlegendentry{$t=2.356$}

\addplot[color=cP64, mark=*, mark options={fill=cP64, draw=cP64}, solid]
table[x=p, y=mean_error, y error=sem_error, col sep=comma]{figures/trotter/err_vs_p_t3142.csv};
\addlegendentry{$t=3.142$}

\end{groupplot}
\end{tikzpicture}
     \caption{
        Trotter error analysis for the constrained continuous-time quantum walk.
        Left: Norm-2 distance between the exact CTQW state and the trotterised state as a function of the number of Trotter steps $p$ for a fixed 3-regular graph with $n=18$ vertices.
        Center: Dependence of the state error on the evolution time $t$ for different Trotter depths $p$ on the same graph.
        Right: Scaling of the Trotter error with graph size for ensembles of random 3-regular graphs at fixed evolution time and fixed Trotter depth.
    }
    \label{fig:trotter_error}
\end{figure}

\subsection{Entanglement analysis}
\label[appendix]{subsec:Entanglement}

\Cref{fig:volume_law} shows the time dependence of the mean von Neumann entropy $S_{\rm blocks}$ for connected blocks of size $k=n/2$ and the scaling of its maximum with system size.
Over the tested system sizes, the approximately linear growth of the maximum block entropy with $n$ is consistent with volume-law entanglement.

The circuit structure provides a mechanism by which such extended correlations can develop.
Each term $\mathrm{C}_{\Nof{i}} RX_i(\beta_i)$ can conditionally entangle qubit $i$ with the qubits in its neighbourhood $\Nof{i}$.
Repeated application of these graph-local operations allows correlations and operator support to spread across progressively larger regions of the graph.

In the Heisenberg picture, this propagation is described by the nested commutator expansion generated by $H_{\mathrm{MVC}}=\sum_i\hat{X}_i$.
As established in \cref{prop:lightcone}, each successive commutation can enlarge the support of an operator by at most one graph-distance shell.
Consequently, local observables such as $\langle Z_i\rangle$ can receive contributions involving vertices at increasing graph distance as the commutator order grows.
The observed entropy scaling provides numerical evidence that the constrained quantum walk generates extended correlations within this graph-local propagation structure.

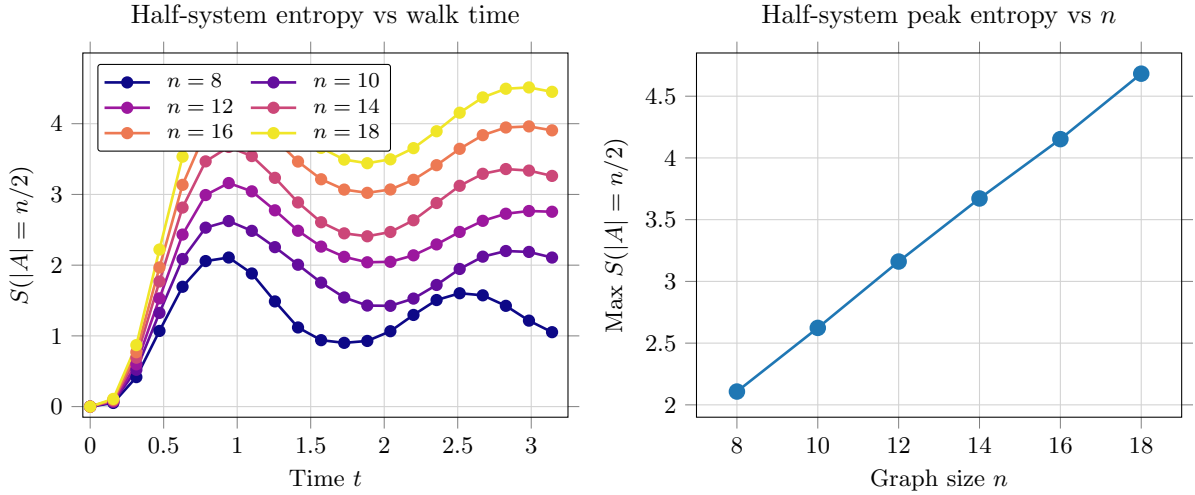
\begin{figure}[t]
    \centering
%
%
\begin{tikzpicture}

\definecolor{cN8}{RGB}{13,8,135}      
\definecolor{cN10}{RGB}{100,13,157}   
\definecolor{cN12}{RGB}{156,23,158}   
\definecolor{cN14}{RGB}{204,71,120}   
\definecolor{cN16}{RGB}{237,121,83}   
\definecolor{cN18}{RGB}{240,230,40}   
\definecolor{cBlue}{RGB}{31,119,180}  

\begin{groupplot}[
  group style={
    group size=2 by 1,
    horizontal sep=1.7cm,
    vertical sep=0.5cm,
  },
  width=8.0cm,
  height=6.4cm,
  grid=both,
  grid style={line width=0.3pt, draw=gray!25},
  major grid style={line width=0.4pt, draw=gray!35},
  tick align=outside,
  tick label style={font=\small},
  label style={font=\small},
  title style={font=\normalsize},
  every axis plot/.append style={
    line width=1.0pt,
    mark size=1.8pt,
    error bars/y dir=both,
    error bars/y explicit=true,
    error bars/error bar style={line width=0.7pt},
    error bars/error mark options={rotate=90,mark size=1.8pt}
  },
]

\nextgroupplot[
  title={Half-system entropy vs walk time},
  xlabel={Time $t$},
  ylabel={$S(|A|=n/2)$},
  xmin=-0.05, xmax=3.25,
  xtick={0,0.5,1.0,1.5,2.0,2.5,3.0},
  ymin=-0.15, ymax=5.0,
  ytick={0,1,2,3,4},
  legend pos=north west,
  legend cell align=left,
  legend columns=2,
  legend style={
    draw=black, rounded corners=1pt, font=\footnotesize,
    inner xsep=2pt, inner ysep=1.5pt, column sep=4pt, row sep=-1pt
  },
]

\addplot[color=cN8,  mark=*, mark options={fill=cN8,  draw=cN8},  solid]
table[x=t, y=mean_S_n8,  y error=sem_S_n8,  col sep=comma]{figures/entanglement/entanglement_data_pivoted.csv};
\addlegendentry{$n=8$}

\addplot[color=cN10, mark=*, mark options={fill=cN10, draw=cN10}, solid]
table[x=t, y=mean_S_n10, y error=sem_S_n10, col sep=comma]{figures/entanglement/entanglement_data_pivoted.csv};
\addlegendentry{$n=10$}

\addplot[color=cN12, mark=*, mark options={fill=cN12, draw=cN12}, solid]
table[x=t, y=mean_S_n12, y error=sem_S_n12, col sep=comma]{figures/entanglement/entanglement_data_pivoted.csv};
\addlegendentry{$n=12$}

\addplot[color=cN14, mark=*, mark options={fill=cN14, draw=cN14}, solid]
table[x=t, y=mean_S_n14, y error=sem_S_n14, col sep=comma]{figures/entanglement/entanglement_data_pivoted.csv};
\addlegendentry{$n=14$}

\addplot[color=cN16, mark=*, mark options={fill=cN16, draw=cN16}, solid]
table[x=t, y=mean_S_n16, y error=sem_S_n16, col sep=comma]{figures/entanglement/entanglement_data_pivoted.csv};
\addlegendentry{$n=16$}

\addplot[color=cN18, mark=*, mark options={fill=cN18, draw=cN18}, solid]
table[x=t, y=mean_S_n18, y error=sem_S_n18, col sep=comma]{figures/entanglement/entanglement_data_pivoted.csv};
\addlegendentry{$n=18$}

\nextgroupplot[
  title={Half-system peak entropy vs $n$},
  xlabel={Graph size $n$},
  ylabel={Max $S(|A|=n/2)$},
  xmin=7, xmax=19,
  xtick={8,10,12,14,16,18},
  ymin=1.9, ymax=4.85,
  ytick={2.0,2.5,3.0,3.5,4.0,4.5},
  yticklabel style={/pgf/number format/fixed, /pgf/number format/precision=1},
  scaled y ticks=false,
]

\addplot[color=cBlue, mark=*, mark options={fill=cBlue, draw=cBlue}, solid, mark size=2.6pt]
table[x=n, y=mean_S_at_peak, y error=sem_S_at_peak, col sep=comma]
{figures/entanglement/entanglement_peak_summary.csv};

\end{groupplot}
\end{tikzpicture}
    \caption{
    Mean von Neumann block entropy $S_{\rm blocks}$ averaged over 40 connected subgraphs.
    Left: $S_{\rm blocks}$ as a function of walk time $t$ for connected blocks of size $k=n/2$.
    Right: Scaling of the maximum block entropy $S_{\rm blocks}^{\max}$ with total system size $n$ at $t=\pi/4$.
    }
    \label{fig:volume_law}
\end{figure}

\section{Monte Carlo simulation for \texorpdfstring{$p=1$}{p=1}}
\label[appendix]{sec:montecarlo}

Although the constrained quantum-walk state can exhibit substantial entanglement, a simplification occurs in the single-layer case $p=1$.
At this depth, the circuit admits an efficient classical representation that enables Monte Carlo estimation of expectation values without explicitly constructing the full quantum state.
In what follows, we formalise the quantum tree representation underlying the $p=1$ state, derive the pseudo-Boolean form of the cost function, and describe the Monte Carlo sampling procedure. We then present numerical results for the performance of our QPG and QEG algorithms using the Monte Carlo estimator for regular graphs and Erd\H{o}s--R\'enyi graphs up to size 300.

\subsection{Quantum tree representation of the feasible subspace}
For $p=1$, the evolved quantum state admits a natural interpretation as a weighted superposition over feasible vertex covers.
The final quantum state can be written as
\begin{equation*}
\ket{\psi(t)}
=
\prod_{i \in V}
\mathrm{C}_{\Nof{i}} RX_i(t)
\ket{\Omega},
\end{equation*}
where the ordering of vertices induces a sequential structure in the circuit.

\paragraph{Quantum tree construction.}

The circuit naturally induces a rooted tree structure, which we refer to as the \emph{quantum state tree}. Each level of the tree corresponds to a vertex $i \in V$ in the prescribed ordering, and each node represents a partial assignment consistent with the vertex-cover constraint.

At level $i$, the controlled rotation $\mathrm{C}_{\Nof{i}} RX(t)$ generates a binary branching depending on the constraints:
\begin{itemize}
    \item If all the neighbours of the vertex $i$ are in the cover:
    \begin{itemize}
    \item with amplitude $\cos(t)$, the qubit remains in state $\ket{1}$ (the vertex remains in the cover), corresponding to the left branch if all the neighbours are in the cover;
    \item with amplitude $i\sin(t)$, the qubit flips to $\ket{0}$ (the vertex is removed from the cover), corresponding to the right branch.
    \end{itemize}
    \item Otherwise the multicontrol don't activate and the qubit i stay in the state $\ket{1}$ with amplitude $1$'
    .
\end{itemize}
\Cref{fig:quantum-state-tree} illustrates this construction for a four-vertex graph and shows the corresponding circuit.
\begin{figure*}[t]
    \centering
    
    \begin{minipage}{0.48\textwidth}
        \centering
        \begin{tikzpicture}[
  scale=0.90, transform shape, 
  x=0.85cm, y=0.95cm,          
  v/.style={circle, fill=cyan!35, draw=black, minimum size=5mm, inner sep=1pt},
  e/.style={draw, line width=0.9pt}
]

\node[v] (r)   at (-1,4)    {1111};

\node[v] (a)   at (-4,2.5) {1111};

\node[v] (c)   at (-7,1.2) {1111};

\node[v] (e1)  at (-9,0)   {1111};
\node[v] (e2)  at (-5,0)   {1011};

\node[v] (f1)  at (-10.5,-1.2) {1111};
\node[v] (f2)  at (-7.5,-1.2)  {1101};

\node[v] (g1)  at (-6.2,-1.2)  {1011};
\node[v] (g2)  at (-3.8,-1.2)  {1001};

\node[v] (d)   at (-2,1.2) {1110};
\node[v] (d2)  at (-2,0)   {1110};
\node[v] (d3)  at (-2,-1.2){1110};

\node[v] (b)   at (1,2.5)  {0111};
\node[v] (b2)  at (1,1.2)  {0111};
\node[v] (b3)  at (1,0)    {0111};
\node[v] (b4)  at (1,-1.2) {0111};

\draw[e] (r) -- node[midway, sloped, above] {$\cos\theta_1$} (a);
\draw[e] (r) -- node[midway, sloped, above] {$\sin\theta_1$} (b);

\draw[e] (a) -- node[midway, sloped, above] {$\cos\theta_2$} (c);
\draw[e] (a) -- node[midway, sloped, above] {$\sin\theta_2$} (d);

\draw[e] (c) -- node[midway, sloped, above] {$\cos\theta_3$} (e1);
\draw[e] (c) -- node[midway, sloped, above] {$\sin\theta_3$} (e2);

\draw[e] (e1) -- node[midway, sloped, above,yshift=3pt] {$\cos\theta_4$} (f1);
\draw[e] (e1) -- node[midway, sloped, above,yshift=3pt] {$\sin\theta_4$} (f2);

\draw[e] (e2) -- node[midway, sloped, above,yshift=3pt] {$\cos\theta_4$} (g1);
\draw[e] (e2) -- node[midway, sloped, above,yshift=3pt] {$\sin\theta_4$} (g2);

\draw[e] (d)  -- (d2) -- (d3);
\draw[e] (b)  -- (b2) -- (b3) -- (b4);

\end{tikzpicture}
    \end{minipage}
    \hfill
    \begin{minipage}{0.48\textwidth}
        \centering
        \begin{tikzpicture}[scale=1, every node/.style={circle, draw=none, fill=cyan!35, draw=black, minimum size=7mm, inner sep=0pt}]
  \node (v2) at (-0.3, 1.3) {$q_2$};
  \node (v1) at ( 0.0, 0.0) {$q_1$};
  \node (v4) at ( 1.6, 0.9) {$q_4$};
  \node (v3) at ( 2.1,-0.6) {$q_3$};

  \draw (v2) -- (v1);
  \draw (v2) -- (v4);
  \draw (v1) -- (v4);
  \draw (v1) -- (v3);
  \draw (v4) -- (v3);
\end{tikzpicture}
    
        \vspace{0.5cm}
    
        \input{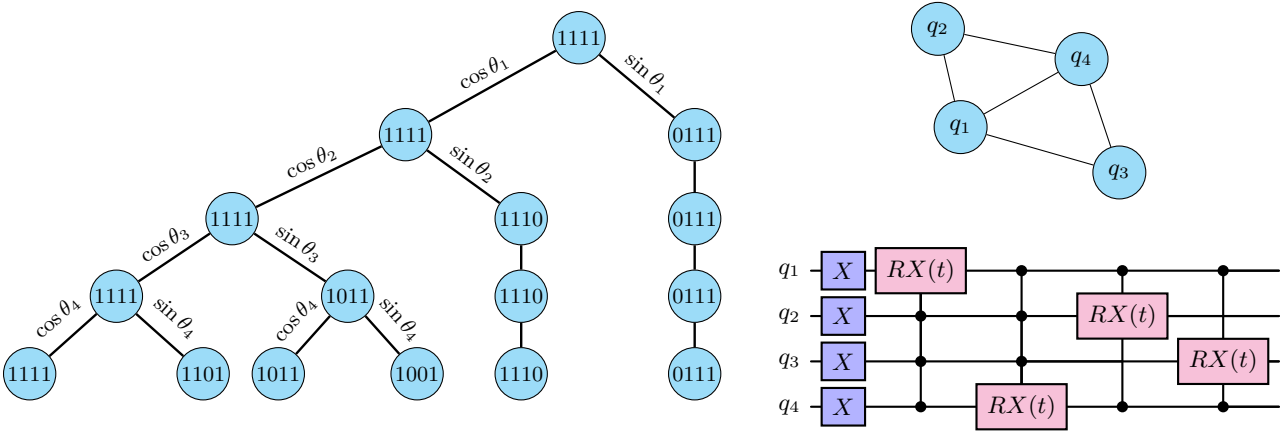}
    \end{minipage}
    
    \caption{
    Quantum-state-tree representation for a four-vertex example.
    The unbalanced tree is shown on the left, while the corresponding graph and quantum circuit are shown on the right.
    The vertex-cover constraints prune infeasible branches and thereby produce the unbalanced tree structure.
    }
    \label{fig:quantum-state-tree}
\end{figure*}

Due to the vertex-cover constraints, not all branches correspond to valid configurations. Branches that would violate feasibility are effectively \emph{pruned} by the projectors implicit in the controlled rotations. As a result, the quantum tree is generally \emph{unbalanced}, with branching occurring only along paths corresponding to valid vertex-cover assignments.
This construction is closely related to the class of \emph{quantum tree generators} introduced in the context of Quantum Branch-and-Bound (QBNB) and constrained variants of QAOA, where feasible solutions are encoded as amplitudes along root-to-leaf paths of a decision tree \cite{Christiansen_2025,wilkening2025constraintorientedbiasedquantumsearch,wilkening2025constraintorientedbiasedquantumsearch}. In those approaches, soft linear constraints are typically enforced through the use of additional ancilla registers and penalty mechanisms, which guide the evolution toward feasible regions of the search space.

In contrast, the present construction implements a \emph{hard constraint embedding} directly at the circuit level: feasibility is enforced locally by the multi-controlled rotations, and invalid branches are naturally pruned without requiring auxiliary qubits or penalty terms. As a result, the quantum tree is generated intrinsically within the computational register, leading to a more resource-efficient representation of the feasible subspace while preserving the same path-amplitude interpretation exploited in QBNB frameworks.

Each leaf of the quantum tree corresponds to a feasible vertex cover $C \in \mathrm{VC}(G)$. The amplitude associated with the computational basis state $\ket{C}$ is uniquely determined by the sequence of left and right branches along the corresponding path $\gamma_C$:
\begin{equation*}
\ket{\psi(t)}
=
\sum_{C \in \mathrm{VC}(G)}
\left(
\prod_{j \in R_{\gamma_C}} i\sin(t)
\prod_{k \in L_{\gamma_C}} \cos(t)
\right)
\ket{C},
\end{equation*}
where $R_{\gamma_C}$ and $L_{\gamma_C}$ denote, respectively, the sets of branching indices at which the path $\gamma_C$ retains or removes the current vertex; indices corresponding to forced retention belong to neither set.

This representation follows from the reachability properties of the constrained mixer: every feasible vertex cover can be obtained from the reference state $\ket{\Omega}$ through a sequence of admissible flips, and the circuit assigns a coherent amplitude to each such path.

\subsection{Cost function as pseudo-Boolean function}
The expectation value of the cost Hamiltonian $H_C =\frac{|V|}{2}- \frac{1}{2} \sum_{i \in V}  Z_i$ with respect to the quantum state $\ket{\psi(t)}$ can be expressed as
\begin{equation*}
\bra{\psi(t)} H_C \ket{\psi(t)}
=
\sum_{C \in \mathrm{VC}(G)} \cost(C) 
\prod_{j \in R_{\gamma_C}} \sin^2(t)
\prod_{k \in L_{\gamma_C}} \cos^2(t),
\end{equation*}
where $R_{\gamma_C}$, $L_{\gamma_C}$ correspond to the sets of qubits that can be flipped along the path $\gamma_C$ in the quantum tree.

Introducing the change of variables
$x := \cos^2(t),  x \in [0,1]$,
the expectation value can be expressed as a relaxed pseudo-Boolean function.
\begin{equation*}
F(x)
=
\sum_{C \in \mathrm{VC}(G)}
\cost(C)
\prod_{j \in R_{\gamma_C}} (1 - x)
\prod_{k \in L_{\gamma_C}} x.
\end{equation*}

This form admits a natural reinterpretation as a \emph{path integral over the quantum tree}:
\begin{equation*}
F(x)
=
\sum_{\gamma \in \Gamma} C_\gamma
\prod_{i \in V} f_i(\gamma_i; x),
\end{equation*}
where $\Gamma$ is the set of all feasible root-to-leaf paths in the quantum tree, $C_\gamma$ is the cost associated with the leaf reached by $\gamma$, and
\begin{equation*}
f_i(\gamma_i; x) =
\begin{cases}
x, & \text{if vertex $i$ is retained along path } \gamma,\\
1 - x, & \text{if vertex $i$ is flipped along path } \gamma.
\end{cases}
\end{equation*}

defining the path probabilities $w_\gamma(x)=\prod_{i \in V} f_i(\gamma_i; x)$ we obtain
\begin{equation}
  F(x)=\sum_{\gamma\in\Gamma}\mathcal C_\gamma\, w_\gamma(x),
  \qquad
  \sum_{\gamma\in\Gamma} w_\gamma(x)=1,
  \label{eq:pathintegral}
\end{equation}
with normalisation $\sum_\gamma w_\gamma=1$.
Hence $\{w_\gamma(x)\}_{\gamma\in\Gamma}$ is a probability distribution over
paths, and $F(x)=\mathbb E_{\gamma\sim w}[\mathcal C_\gamma]$ is a genuine
expectation, with \emph{no} hidden normalisation constant. The per-vertex
marginal probability that vertex $v$ lies in the cover is likewise
\begin{equation}
  P_v(x)=\sum_{\gamma\in\Gamma} w_\gamma(x)\,\mathbb I[v\in\gamma],
  \label{eq:marginal}
\end{equation}
which coincides with the QPG score evaluated at $p=1$.

\subsection{Sequential Monte Carlo sampler}

\Cref{eq:pathintegral,eq:marginal} show that both QEG and
QPG scores at $p=1$ are expectations under the path measure $w_\gamma(x)$.
Because $w_\gamma$ factorises over the tree, a path can be drawn by a single
top-down traversal that visits vertices in the fixed order and respects
feasibility, with no global summation over the exponentially many covers.

\begin{figure}[t]
    \refstepcounter{algbox}
    \label{alg:mcsample}
    \noindent\textbf{Algorithm~\thealgbox: Monte Carlo sampler for $p=1$}
    \vspace{2pt}
    \hrule
    \vspace{4pt}
    \begin{algorithmic}[1]
    \State \textbf{Input:} graph $G=(V,E)$, parameter $x=\cos^2 t$, vertex order
           $1,\dots,n$
    \State Initialise $C\gets V$ \Comment{all vertices in the cover}
    \For{$i=1$ to $n$}
      \If{removing $i$ keeps $C\setminus\{i\}$ a valid cover}
           \Comment{branching node}
        \State draw $u\sim \mathcal U(0,1)$
        \If{$u\ge x$} $C\gets C\setminus\{i\}$ \Comment{remove with probability $1-x$}
        \EndIf
      \EndIf \Comment{otherwise $i$ is retained deterministically}
    \EndFor
    \State \textbf{return} feasible cover $C$ and its $\cost(C) = \sum_{i\in C}c_i$
    \end{algorithmic}
\end{figure}

A vertex is removed only when both (i) a uniform draw exceeds $x$ and (ii) the
removal preserves feasibility; the feasibility test at line~4 is exactly the
pruning of forced nodes, so the leaf returned is sampled with probability
$w_{\gamma_C}(x)=\abs{a_C(t)}^2$. The traversal touches each vertex once and the
feasibility check inspects only the neighbours of $i$, giving $\mathcal{O}(n+|E|)$ time
per sample.

\paragraph{Estimators.}
From $N$ independent samples $C^{(1)},\dots,C^{(N)}$ drawn by
\cref{alg:mcsample}, the cost expectation and the per-vertex marginals
are estimated by
\begin{equation*}
  \widehat F
  =\frac1N\sum_{s=1}^{N}\mathcal C_{C^{(s)}},
  \qquad
  \widehat P_v
  =\frac1N\sum_{s=1}^{N}\mathbb I\!\left[v\in C^{(s)}\right].
\end{equation*}
Since each $C^{(s)}\sim w(x)$, both estimators are unbiased,
$\mathbb E[\widehat F]=F(x)$ and $\mathbb E[\widehat P_v]=P_v(x)$, with variance
$\mathcal{O}(1/N)$. Crucially the procedure
scales in general quadratically in $n$ and linearly in $N$, making it tractable for graphs with hundreds
of vertices, where exact preparation of $\ket{\psi(t)}$ is infeasible.

\paragraph{Greedy scores at $p=1$.}
The QPG-LDF/SDF rules use $\widehat P_v$ directly in place of the measured
marginals. The QEG rules require the conditional cost obtained by fixing a
vertex; at $p=1$ this is implemented by running \cref{alg:mcsample}
on the graph with the corresponding mixer term removed (vertex $v$ pinned to the
cover), and estimating $\widehat F$ on the resulting tree. As with the full
algorithm, this incurs one estimator per candidate vertex per greedy step, i.e. 
$\mathcal{O}(n)$ traversals per step and $\mathcal{O}(n^2)$ per run.

\subsection{Numerical results up to \texorpdfstring{$n=300$}{n=300}}

\Cref{fig:montecarlo} reports the mean approximation ratio of the $p=1$ surrogate for QPG and QEG, in both LDF and SDF variants, on $20$ instances per size for $3$-regular, $4$-regular, and Erd\H{o}s--R\'enyi ($\rho=0.3$) graphs, with sizes up to $n=300$.
Each score is estimated from $N=10^{4}$ Monte Carlo samples drawn using \cref{alg:mcsample}, with the vertices visited in descending-degree order and the walk time fixed at the family-specific value $t^\star$ selected separately for the $p=1$ Monte Carlo experiment.
The qualitative ordering established in the main text persists in this larger-size study: QEG-LDF remains the strongest variant across all three families, with no observed degradation over the simulated range.
These results demonstrate the performance of the Monte Carlo surrogate at larger graph sizes; they do not, on their own, establish a quantum advantage.

\begin{figure}
    \centering
    \pgfplotsset{compat=1.18}
\usepgfplotslibrary{groupplots}

\begin{tikzpicture}
\definecolor{G}{RGB}{0,0,0}
\definecolor{p1}{RGB}{0,114,178}
\definecolor{p2}{RGB}{213,94,0}
\definecolor{p3}{RGB}{0,158,115}
\definecolor{p4}{RGB}{204,121,167}
\definecolor{p5}{RGB}{86,180,233}
\definecolor{p6}{RGB}{230,159,0}

\pgfkeys{/pgf/number format/.cd,1000 sep={}}

\pgfplotsset{
  qpgldf/.style={
    color=p1,
    mark=o,
    mark options={draw=p1,fill=white},
    solid,
  },
  qegldf/.style={
    color=p2,
    mark=square,
    mark options={draw=p2,fill=p2},
    solid,
  },
  qpgsdf/.style={
    color=p5,
    mark=o,
    mark options={draw=p5,fill=white,solid},
    dashed,
    error bars/error bar style={solid},
    error bars/error mark options={rotate=90,mark size=1.5pt,solid},
  },
  qegsdf/.style={
    color=p6,
    mark=square,
    mark options={draw=p6,fill=p6,solid},
    dashed,
    error bars/error bar style={solid},
    error bars/error mark options={rotate=90,mark size=1.5pt,solid},
  },
}

\begin{groupplot}[
  group style={
    group size=3 by 1,
    horizontal sep=0.5cm,
    y descriptions at=edge left,
  },
  width=6cm,
  height=6cm,
  xlabel={Graph size $n$},
  ylabel={Mean ratio},
  xmin=8,
  ymin=1,
  ymax=1.15,
  grid=both,
  xtick={50,100,150,200,250,300},
  tick align=inside,
  tick label style={font=\normalsize},
  label style={font=\normalsize},
  title style={font=\normalsize},
  every axis plot/.append style={
    line width=1pt,
    mark size=2pt,
    error bars/y dir=both,
    error bars/y explicit=true,
    error bars/error bar style={line width=1.0pt},
    error bars/error mark options={rotate=90,mark size=1.5pt}
  },
]

\nextgroupplot[title={3-regular}]

\addplot+[
  qpgldf,
  error bars/.cd,
  y dir=both,
  y explicit
]
table[
  x=graph_nodes,
  y=mean,
  y error expr=\thisrow{std}/sqrt(20),
  col sep=comma
]
{figures/final_reduced/MC/3-regular/reduced_conditional_smc.csv};

\addplot+[
  qegldf,
  error bars/.cd,
  y dir=both,
  y explicit
]
table[
  x=graph_nodes,
  y=mean,
  y error expr=\thisrow{std}/sqrt(20),
  col sep=comma
]
{figures/final_reduced/MC/3-regular/reduced_n2_reduced.csv};

\addplot+[
  qpgsdf,
  error bars/.cd,
  y dir=both,
  y explicit
]
table[
  x=graph_nodes,
  y=mean,
  y error expr=\thisrow{std}/sqrt(20),
  col sep=comma
]
{figures/final_reduced/MC/3-regular/reduced_conditional_smc_reverse.csv};

\addplot+[
  qegsdf,
  error bars/.cd,
  y dir=both,
  y explicit
]
table[
  x=graph_nodes,
  y=mean,
  y error expr=\thisrow{std}/sqrt(20),
  col sep=comma
]
{figures/final_reduced/MC/3-regular/reduced_n2_reduced_reverse.csv};

\nextgroupplot[title={4-regular}]

\addplot+[
  qpgldf,
  error bars/.cd,
  y dir=both,
  y explicit
]
table[
  x=graph_nodes,
  y=mean,
  y error expr=\thisrow{std}/sqrt(20),
  col sep=comma
]
{figures/final_reduced/MC/4-regular/reduced_conditional_smc.csv};

\addplot+[
  qegldf,
  error bars/.cd,
  y dir=both,
  y explicit
]
table[
  x=graph_nodes,
  y=mean,
  y error expr=\thisrow{std}/sqrt(20),
  col sep=comma
]
{figures/final_reduced/MC/4-regular/reduced_n2_reduced.csv};

\addplot+[
  qpgsdf,
  error bars/.cd,
  y dir=both,
  y explicit
]
table[
  x=graph_nodes,
  y=mean,
  y error expr=\thisrow{std}/sqrt(20),
  col sep=comma
]
{figures/final_reduced/MC/4-regular/reduced_conditional_smc_reverse.csv};

\addplot+[
  qegsdf,
  error bars/.cd,
  y dir=both,
  y explicit
]
table[
  x=graph_nodes,
  y=mean,
  y error expr=\thisrow{std}/sqrt(20),
  col sep=comma
]
{figures/final_reduced/MC/4-regular/reduced_n2_reduced_reverse.csv};

\nextgroupplot[title={Erdos-Renyi}]

\addplot+[
  qpgldf,
  error bars/.cd,
  y dir=both,
  y explicit
]
table[
  x=graph_nodes,
  y=mean,
  y error expr=\thisrow{std}/sqrt(20),
  col sep=comma
]
{figures/final_reduced/MC/erdos/reduced_conditional_smc.csv};

\addplot+[
  qegldf,
  error bars/.cd,
  y dir=both,
  y explicit
]
table[
  x=graph_nodes,
  y=mean,
  y error expr=\thisrow{std}/sqrt(20),
  col sep=comma
]
{figures/final_reduced/MC/erdos/reduced_n2_reduced.csv};

\addplot+[
  qpgsdf,
  error bars/.cd,
  y dir=both,
  y explicit
]
table[
  x=graph_nodes,
  y=mean,
  y error expr=\thisrow{std}/sqrt(20),
  col sep=comma
]
{figures/final_reduced/MC/erdos/reduced_conditional_smc_reverse.csv};

\addplot+[
  qegsdf,
  error bars/.cd,
  y dir=both,
  y explicit
]
table[
  x=graph_nodes,
  y=mean,
  y error expr=\thisrow{std}/sqrt(20),
  col sep=comma
]
{figures/final_reduced/MC/erdos/reduced_n2_reduced_reverse.csv};

\end{groupplot}

\node[
  anchor=north,
  yshift=-0.35cm,
  draw=black,
  fill=white,
  rounded corners,
  inner xsep=7pt,
  inner ysep=5pt,
  font=\small
] (manuallegend) at (current bounding box.south) {%
\begin{tabular}{@{}ll@{\hspace{0.7cm}}ll@{}}
\tikz[baseline=-0.6ex]{
  \begin{axis}[
    hide axis,
    scale only axis,
    width=0.9cm,
    height=0.25cm,
    xmin=0,
    xmax=1,
    ymin=0,
    ymax=1,
    every axis plot/.append style={line width=1pt,mark size=2pt},
  ]
  \addplot+[qpgldf] coordinates {(0,0.5) (1,0.5)};
  \end{axis}
}
&
QPG-LDF
&
\tikz[baseline=-0.6ex]{
  \begin{axis}[
    hide axis,
    scale only axis,
    width=0.9cm,
    height=0.25cm,
    xmin=0,
    xmax=1,
    ymin=0,
    ymax=1,
    every axis plot/.append style={line width=1pt,mark size=2pt},
  ]
  \addplot+[qegldf] coordinates {(0,0.5) (1,0.5)};
  \end{axis}
}
&
QEG-LDF
\\[0.25em]
\tikz[baseline=-0.6ex]{
  \begin{axis}[
    hide axis,
    scale only axis,
    width=0.9cm,
    height=0.25cm,
    xmin=0,
    xmax=1,
    ymin=0,
    ymax=1,
    every axis plot/.append style={line width=1pt,mark size=2pt},
  ]
  \addplot+[qpgsdf] coordinates {(0,0.5) (1,0.5)};
  \end{axis}
}
&
QPG-SDF
&
\tikz[baseline=-0.6ex]{
  \begin{axis}[
    hide axis,
    scale only axis,
    width=0.9cm,
    height=0.25cm,
    xmin=0,
    xmax=1,
    ymin=0,
    ymax=1,
    every axis plot/.append style={line width=1pt,mark size=2pt},
  ]
  \addplot+[qegsdf] coordinates {(0,0.5) (1,0.5)};
  \end{axis}
}
&
QEG-SDF
\end{tabular}%
};

\useasboundingbox (current bounding box.north west) (manuallegend.south east);

\end{tikzpicture}
    \caption{
        Mean approximation ratio obtained using the $p=1$ Monte Carlo surrogate on $20$ graph instances of each size and family.
        The panels show $3$-regular graphs (left), $4$-regular graphs (centre), and connected Erd\H{o}s--R\'enyi graphs with $\rho=0.3$ (right).
        Results are shown for the QPG and QEG scores with both LDF and SDF reduction rules, for graph sizes up to $n=300$. The error bars show the standard error of the mean.
    }
\label{fig:montecarlo}
\end{figure}
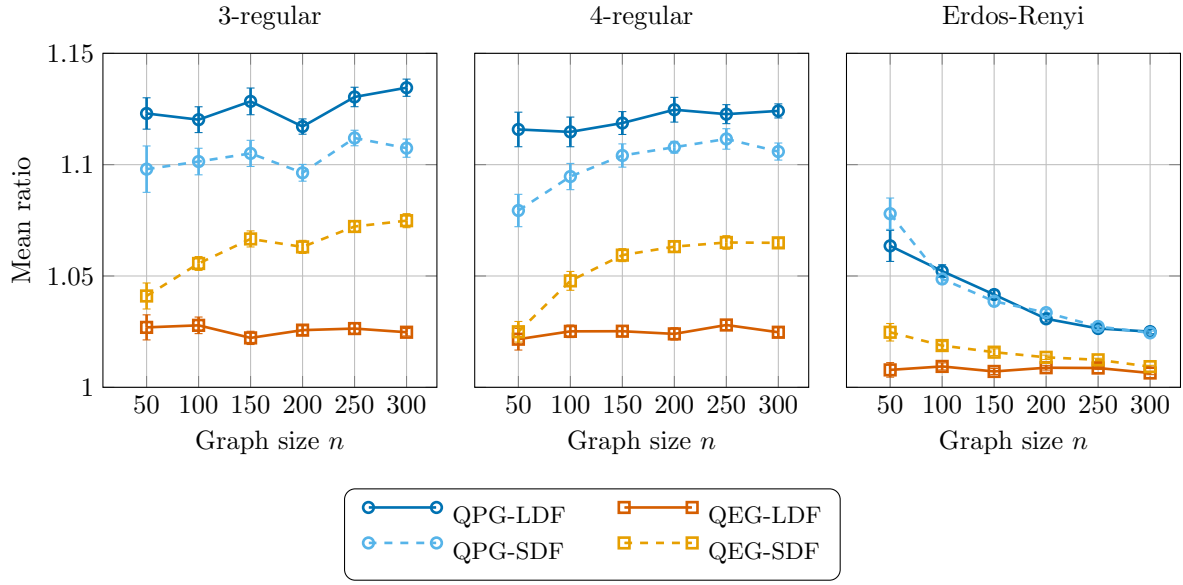

\end{document}